%% file: paper.tex
\documentclass{llncs}

\usepackage[leqno]{amsmath}
\usepackage{amsfonts}
\usepackage[dvipsnames]{xcolor}
\usepackage[clock]{ifsym}
\usepackage{bm}
\usepackage[inline]{enumitem}
\usepackage{mathpartir}
\usepackage{amssymb}

\usepackage[colorlinks = true,
            linkcolor = blue,
            citecolor = blue]{hyperref}
\usepackage{cleveref}

\usepackage{appendix}
\usepackage[strict]{changepage}

\usepackage{orcidlink}

\newtheorem{innercustomthm}{Theorem}
\newenvironment{customthm}[1]
  {\renewcommand\theinnercustomthm{#1}\innercustomthm}
  {\endinnercustomthm}

\newtheorem{innercustomcor}{Corollary}
  \newenvironment{customcor}[1]
    {\renewcommand\theinnercustomcor{#1}\innercustomcor}
    {\endinnercustomcor}

\definecolor{colabred}{RGB}{242, 60, 60}
\definecolor{colabgrey}{RGB}{98, 100, 112}
\definecolor{gogreen}{RGB}{14,173,105}
\definecolor{shinyshamrock}{RGB}{81,152,114}

\crefname{definition}{def.}{defs.}
\crefname{corollary}{cor.}{cors.}
\crefname{section}{\S}{\S}
\crefname{table}{tbl.}{tbls.}
\crefname{rule}{rule}{rules}
\crefname{case}{case}{cases}
\crefname{subcase}{subcase}{subcases}

\newcommand*{\eg}{\textit{e.g.}}
\newcommand*{\Eg}{\textit{E.g.}}
\newcommand*{\ie}{\textit{i.e.,}}
\newcommand*{\Ie}{\textit{I.e.,}}
\newcommand*{\etc}{\textit{etc.}}
\newcommand*{\etal}{\textit{et~al.}}

\newcommand*{\wrt}{w.r.t.}
\newcommand*{\Wrt}{W.r.t.}
\newcommand*{\st}{\textnormal{st:}}
\newcommand*{\ih}{\textnormal{i.h.}}

\newcommand{\selection}{\ensuremath{\oplus}}
\newcommand{\Selection}{\ensuremath{\bm \oplus}}
\newcommand{\role}[1]{\textsf{\color{colabred} \textbf{#1}}}
\newcommand{\msgLabel}[1]{\textsf{#1}}
\newcommand{\branch}{\ensuremath{\,\&}}
\newcommand{\Branch}{\ensuremath{\bm \,\&}}
\newcommand{\timeOut}{\ensuremath{{\scriptstyle\VarClock}}}
\newcommand{\emptyBuffer}{\ensuremath{\epsilon}}
\newcommand{\ctx}{\ensuremath{\mathbb{C}}}
\newcommand{\inaction}{\ensuremath{\bm 0}}
\newcommand{\roles}{\ensuremath{\bm\rho}}
\newcommand{\dom}{\textsf{dom}}
\newcommand{\rel}{\ensuremath{\mathcal{R}}}
\newcommand{\relSet}{\ensuremath{\mathfrak{R}}}
\newcommand{\recVar}[1]{\textnormal{\textrm{#1}}}
\newcommand{\bufferTracker}{\ensuremath{\Sigma}}
\newcommand{\safe}[3]{\ensuremath{(#1;#2)\textnormal{-}\fSafe(#3)}}
\newcommand{\reason}[1]{\textnormal{(#1)}}

\newcommand{\fv}{\texttt{fv}}
\newcommand{\fend}{\textnormal{\texttt{end}}}
\newcommand{\fBasic}{\textnormal{\texttt{basic}}}
\newcommand{\fPayloads}{\textnormal{\texttt{gc}}}
\newcommand{\fSafe}{\textnormal{\ensuremath{\varphi_\texttt{s}}}}
\newcommand{\fTCP}{\textnormal{\ensuremath{\varphi_\texttt{TCP}}}}
\newcommand{\fLive}{\textnormal{\ensuremath{\varphi_\texttt{L}}}}
\newcommand{\msgInsert}{\ensuremath{\rightsquigarrow}}

\newcommand{\cVar}[1]{\ensuremath{#1}}
\newcommand{\cSessionRole}[2]{\ensuremath{#1}[\role{#2}]}
\newcommand{\cInaction}{\ensuremath{\bm 0}}
\newcommand{\cRestriction}[2]{\ensuremath{(\nu #1)\,} #2}
\newcommand{\cPar}[2]{#1 \ensuremath{\,|\,} #2}
\newcommand{\cChoice}[2]{#1\ensuremath{+}#2}
\newcommand{\cSnd}[5]{#1\selection[\role{#2}]\ensuremath{\,!\,}\msgLabel{#3}\ensuremath{\langle}#4\ensuremath{\rangle .\,}#5}
\newcommand{\cBranch}[2]{#1\branch\ensuremath{\{}#2\ensuremath{\}}}
\newcommand{\cBranchL}[2]{#1\branch\ensuremath{\left\{}#2\ensuremath{\right.}}

\newcommand{\cBranchQuant}[2]{#1\branch\ensuremath{_{i\in I}\{}#2\ensuremath{\}}}
\newcommand{\cBranchQuantx}[4]{#1\branch\ensuremath{_{#3\in #4}\{}#2\ensuremath{\}}}
\newcommand{\cRcv}[4]{[\role{#1}]\ensuremath{\,?\,}\msgLabel{#2}\ensuremath{(}#3\ensuremath{).}#4}
\newcommand{\cTimeOut}[1]{\timeOut\ensuremath{.\,}#1}
\newcommand{\cDef}[2]{\textsf{def} #1 \textsf{in} #2}
\newcommand{\cDecl}[3]{\cVar{#1}\ensuremath{(}#2\ensuremath{) =} #3}
\newcommand{\cCall}[2]{\cVar{#1}\ensuremath{\langle}#2\ensuremath{\rangle}}
\newcommand{\cBuffer}[2]{\cVar{#1}\ensuremath{:}#2}
\newcommand{\cBufferCons}{\ensuremath{\cdot} }
\newcommand{\cBufferEntry}[4]{\ensuremath{(\role{#1}\triangleright\role{#2}\triangleleft\msgLabel{#3} \langle #4 \rangle)}}

\newcommand{\rSelection}{\textbf{[R-\Selection]}}
\newcommand{\rBranch}{\textbf{[R-\Branch]}}
\newcommand{\rTimeOut}{\textbf{[R-\timeOut]}}
\newcommand{\rChoice}{\textbf{[R-+]}}
\newcommand{\rCall}{\textbf{[R-\ensuremath{\bm X}]}}
\newcommand{\rCtx}{\textbf{[R-\ctx]}}
\newcommand{\rDrop}{\textbf{[R-\ensuremath{\bm\downarrow}]}}

\newcommand{\tInt}{\texttt{int}}
\newcommand{\tBool}{\texttt{bool}}
\newcommand{\tReal}{\texttt{real}}
\newcommand{\tUnit}{\texttt{unit}}
\newcommand{\tType}{\ensuremath{\mathbb{T}}}
\newcommand{\tBasic}{\ensuremath{\mathbb{B}}}
\newcommand{\tSession}{\ensuremath{\mathbb{S}}}
\newcommand{\tQueue}{\ensuremath{\mathbb{M}}}
\newcommand{\tSessionQueue}{\ensuremath{\tau}}
\newcommand{\tRcv}[4]{\role{#1}\ensuremath{\,?\,}\msgLabel{#2}\ensuremath{(}#3\ensuremath{).}#4}
\newcommand{\tSnd}[4]{\role{#1}\ensuremath{\,!\,}\msgLabel{#2}\ensuremath{(}#3\ensuremath{).}#4}
\newcommand{\tBranch}[1]{\branch\ensuremath{\left\{}#1\ensuremath{\right\}}}
\newcommand{\tBranchL}[1]{\branch\ensuremath{\left\{}#1\ensuremath{\right.}}
\newcommand{\tBranchQuant}[1]{\branch\ensuremath{_{i\in I}\{}#1\ensuremath{\}}}
\newcommand{\tBranchQuantx}[3]{\branch\ensuremath{_{#2\in #3}\{}#1\ensuremath{\}}}
\newcommand{\tSelect}[1]{\selection\ensuremath{\{}#1\ensuremath{\}}}
\newcommand{\tSelectL}[1]{\selection\ensuremath{\left\{}#1\ensuremath{\right.}}
\newcommand{\tSelectQuant}[1]{\selection\ensuremath{_{i\in I}\{}#1\ensuremath{\}}}
\newcommand{\tRec}[2]{\ensuremath{\mu}\recVar{#1}.#2}
\newcommand{\tEnd}{\textsf{\textbf{end}}}
\newcommand{\tBuffer}[4]{\role{#1}\ensuremath{\,!\,}\msgLabel{#2}\ensuremath{(}#3\ensuremath{)\cdot}#4}
\newcommand{\tTuple}[2]{(#1\ensuremath{\,;\,}#2)}

\newcommand{\tr}[1]{{\scriptsize\textnormal{\textup{[T-#1]}}}}
\newcommand{\trInaction}{\tr{\inaction}}
\newcommand{\trVar}{\tr{Var}}
\newcommand{\trVal}{\tr{Val}}
\newcommand{\trX}{\tr{X}}
\newcommand{\trSelection}{\tr{\selection}}
\newcommand{\trBranch}{\tr{\branch}}

\newcommand{\trCall}{\tr{Call}}
\newcommand{\trDef}{\tr{Def}}
\newcommand{\trPar}{\tr{\ensuremath{|}}}
\newcommand{\trChoice}{\tr{+}}
\newcommand{\trLift}{\tr{Lift}}
\newcommand{\trNew}{\tr{\ensuremath{\nu}}}
\newcommand{\trEmpty}{\tr{\emptyBuffer}}
\newcommand{\trBuffer}[1]{\tr{\ensuremath{\sigma_#1}}}

\newcommand{\crr}[1]{{\scriptsize\textnormal{\textup{[\contG-#1]}}}}
\newcommand{\crrSnd}[1]{\crr{Snd\ensuremath{_#1}}}
\newcommand{\crrCom}{\crr{Com}}
\newcommand{\crrTO}{\crr{\timeOut}}
\newcommand{\crrRec}{\crr{\tmu}}
\newcommand{\crrCong}{\crr{Cong}}

\newcommand{\contO}{\ensuremath{{\color{gogreen} \Theta}}}
\newcommand{\contG}{\ensuremath{{\color{gogreen} \Gamma}}}
\newcommand{\contGb}[1]{\ensuremath{{\color{gogreen} \Gamma_{#1}}}}
\newcommand{\contGbp}[1]{\ensuremath{{\color{gogreen} \Gamma_{#1}^\prime}}}
\newcommand{\contGp}{\ensuremath{{\color{gogreen} \Gamma^\prime}}}
\newcommand{\contGpp}{\ensuremath{{\color{gogreen} \Gamma^{\prime\prime}}}}
\newcommand{\contV}[2]{#1 : #2}
\newcommand{\contSnd}[5]{\cSessionRole{#1}{#2}\ensuremath{\,!\,}\role{#3}:\msgLabel{#4}(#5)}
\newcommand{\contCom}[4]{\ensuremath{#1}[\role{#2}][\role{#3}]:\msgLabel{#4}}
\newcommand{\contTimeout}[2]{\cSessionRole{#1}{#2}\smol\timeOut}

\newcommand{\spGen}[1]{{\color{colabred} \scriptsize\textnormal{\textup{[S-#1]}}}}
\newcommand{\spRed}{\spGen{\ensuremath{\rightarrow}}}
\newcommand{\spCom}{\spGen{Com}}
\newcommand{\spRec}{\spGen{\ensuremath{\mu}}}
\newcommand{\spR}[1]{\spGen{\ensuremath{\rel_{#1}}}}

\newcommand{\tpi}{\ensuremath{\pi}}
\newcommand{\tmu}{\ensuremath{\mu}}

\newcommand{\smol}{\ensuremath{\,}}

\newcommand{\sep}{\ensuremath{\:\:\:}}
\newcommand{\Sep}{\ensuremath{\:\:\:\:\:\:}}

\newcommand{\hl}[1]{\colorbox{lightgray}{#1}}
\newcommand{\raiseup}[1]{\raisebox{0.25ex}{#1}}
\newcommand{\quot}[1]{``#1''}

\newcommand{\raft}{\textsf{Raft}}
\newcommand{\paxos}{\textsf{Paxos}}
\newcommand{\tcp}{\texttt{TCP}}
\newcommand{\magpi}{MAG\tpi}

\usepackage[colorinlistoftodos]{todonotes}

\begin{document}


\title{MAG\tpi: Types for Failure-Prone Communication}

\author{
  Matthew~Alan {Le Brun} \orcidlink{0000-0001-7394-0122} \and 
  Ornela Dardha \orcidlink{0000-0001-9927-7875}
}

\institute{
  University of Glasgow \\ \email{m.le-brun.1@research.gla.ac.uk} \\ \email{ornela.dardha@glasgow.ac.uk}
}

\maketitle


\begin{abstract}
\emph{Multiparty Session Types} (MPST) are a typing discipline for communication-centric systems, guaranteeing communication safety, deadlock freedom and protocol compliance.
Several works have emerged which model failures and introduce fault-tolerance techniques. However, such works often make assumptions on the underlying network, \eg, TCP-based communication where messages are guaranteed to be delivered; or adopt {centralised} reliable nodes and an ad-hoc notion of reliability; or address \emph{only} a single kind of failure, \eg, node crash failures.

In this work, we develop \magpi---a Multiparty, Asynchronous and Generalised \tpi-calculus, which is the \emph{first language and type system} to accommodate in unison: ($i$) the widest range of non-Byzantine faults, including \emph{message loss, delays} and \emph{reordering}; \emph{crash failures} and \emph{link failures}; and \emph{network partitioning}; ($ii$) a novel and most general notion of \emph{reliability}, taking into account the viewpoint of \emph{each} participant in the protocol; ($iii$) the spectrum of network assumptions from the lowest UDP-based network programming to the TCP-based application level.
We prove subject reduction and session fidelity; process properties, (deadlock freedom, termination); failure-handling safety and reliability adherence.
\end{abstract}


\section{Introduction}\label{sec:intro}
\input{sections/1-intro.tex}

\section{MAG\tpi: Language}\label{sec:calculus}
\input{sections/2-calculus.tex}

\section{MAG\tpi: Type System}\label{sec:types}
\input{sections/3-types.tex}

\section{Case Study}\label{sec:examples}

\input{sections/4-examples.tex}

\section{Related Work, Conclusions and Future Work}\label{sec:related}\label{sec:conc}
\input{sections/5-related.tex}

\clearpage
\bibliographystyle{splncs04}
\bibliography{refs}

\clearpage
\changetext{}{10em}{-5em}{-5em}{}
\appendix

\section{Subject Reduction}
\input{sections/appendix/subject-reduction.tex}
\label{app:SR}

\section{Session Fidelity}\label{app:SF}
\input{sections/appendix/session-fidelity.tex}

\section{Process Properties and Decidability}\label{app:proc-props}
\input{sections/appendix/process-properties.tex}

\section{Additional Examples}\label{app:examples}
\input{sections/appendix/leader-election.tex}


\end{document}

%% file: sections/1-intro.tex

Despite large investments into fault-prevention techniques, failures still regularly 
occur in complex distributed applications.
%
%
It is widely agreed that traditional methods of verification using software testing 
do not provide high levels of confidence in the correctness of 
distributed algorithms.
This is mainly due to the \emph{non-deterministic} behaviour inherent to these 
protocols, which makes it unfeasible to manually test for all edge cases.
This problem is bypassed by using exhaustive techniques such 
as model checking~\cite{DBLP:books/daglib/0007403,DBLP:journals/fac/Rossi21}, 
capable of exploring the entirety of the state space of a program to verify 
its correctness.
However, building suitable models for complex distributed algorithms is arduous, 
expensive, and often intractable (due to the state explosion 
problem~\cite{DBLP:conf/laser/ClarkeKNZ11}).
Furthermore, even if an algorithm is successfully encoded into a suitable model 
and checked, guarantees of correctness are on the \emph{design} of the algorithm, and not 
on the software \emph{implementation}; handwritten code is still prone to 
human error. 
Contrastively, types and type systems~\cite{DBLP:books/daglib/0005958} are
lightweight forms of verification. 
Baked in programming languages, types provide guarantees directly on handwritten code 
and aid developers in implementing software which is correct by construction.
Specific to concurrent and distributed computing, 
\emph{session types}~\cite{DBLP:conf/concur/Honda93,THK94,HVK98,V12,DBLP:journals/pacmpl/ScalasY19,DBLP:journals/jacm/HondaYC16}
have quickly grown in popularity since their initial conceptualisation 
\cite{DBLP:conf/concur/Honda93}, spanning from \emph{binary}--two participants, to \emph{multiparty}--many participants.

Session types enforce that processes communicate according to a protocol specification.
Consequently, desirable properties about communication, \eg, \emph{type safety} (communication occurs error-free), \emph{protocol compliance} (or session fidelity; processes behave according to their predefined protocol), 
and \emph{deadlock freedom} (processes do not get stuck waiting), can be \emph{statically} determined by 
a type checker. 
To this aim, session types have been implemented in various programming languages, 
including Java~\cite{KouzapasDPG16,Dezani_2006}, Go~\cite{DBLP:conf/popl/LangeNTY17}, Haskell~\cite{KokkeD21,orchard2017session}, Scala~\cite{ScalasDHY17}, Rust~\cite{DBLP:conf/ecoop/LagaillardieNY22}, Elixir~\cite{DBLP:conf/agere/TaboneF21}.

To date, most session type theories are designed for \emph{concurrent}, as opposed 
to \emph{distributed} processes---\ie\ it is commonly assumed that communication failures 
do not occur.
For the few (and rapidly increasing) works that do consider failures, heavy 
assumptions are made that impede their viability for realistic complex distributed 
applications.
\Eg, \emph{asynchronous} theories~\cite{DBLP:conf/concur/MajumdarMSZ21,DBLP:journals/jacm/HondaYC16,DBLP:journals/pacmpl/ScalasY19} 
use \emph{message buffers} to model distributed communication under \quot{TCP-like} assumptions: messages are 
guaranteed to be delivered and messages from a single sender do not get 
reordered.
\emph{Affine sessions}~\cite{DBLP:journals/lmcs/MostrousV18,DBLP:journals/pacmpl/FowlerLMD19,DBLP:journals/mscs/CapecchiGY16} 
only allow failure-handling of application level failures through \emph{try-catch} blocks; 
there is no support for \emph{arbitrary} failures that may stem from hardware 
faults, network inconsistencies \etc\
\emph{Coordinator model} approaches~\cite{DBLP:conf/forte/AdameitPN17,DBLP:conf/forte/ChenVBZE16,DBLP:conf/esop/VieringCEHZ18} 
assume some degree of reliability, be it as a 
central resilient process, a reliable broadcast, or fixed synchronisation 
points.

The harsh reality is that many real-world distributed protocols (\eg, consensus algorithms) 
cannot assume \emph{any} of these conditions.
Networks introduce many points of failure into a system: nodes may crash, 
messages can be dropped, delayed or duplicated, links between nodes may fail \etc\
Designers of distributed protocols have acknowledged that failure is inevitable,
and so algorithms are designed to withstand a threshold of failure whilst still 
achieving their expected behaviour---known as 
\emph{fault-tolerance}~\cite{laprie1985dependable}.
Examples of fault-tolerant protocols (extensively) used today
include the \paxos~\cite{DBLP:journals/tocs/Lamport98} and 
\raft~\cite{DBLP:conf/usenix/OngaroO14} 
consensus algorithms, which assume the possibility of \emph{all non-Byzantine} 
faults---\ie\ node crashes, link failures, network partitions, and message inconsistencies.

Although the correctness of these algorithms has been heavily studied,
many of them are developed with limited confidence in the correctness of 
the deployable artifact, due to the reasons previously outlined.
To fill this gap, we need type-based verification, which can be made available to programming languages, thus supporting designers and developers in designing and implementing correct distributed algorithms. 
While (multiparty) session types have made great impact in modelling structured communication and guaranteeing relevant properties, their theory is not yet expressive to model these complex algorithms.
%
%

In this paper, we take steps towards filling this gap by presenting \magpi---a Multiparty, Asynchronous and Generalised \tpi-calculus---the first language and type system able to accommodate: ($i$) the widest range of non-Byzantine faults, including \emph{message loss}, \emph{delays} and \emph{reordering}; \emph{crash} and \emph{link failures}; and \emph{network partitioning}---all by using \emph{timeouts}; ($ii$) a novel and most general notion of \emph{reliability}, taking into account the viewpoint of each participant in the protocol; and ($iii$) a spectrum of network assumptions---from the lowest level of network programming based on UDP, to application level based on TCP.
\begin{example}[Ping Pong: Types]
\label{ping-pong}
We illustrate \magpi\ with a simplified version of the ping utility from the 
Internet Control Message Protocol (ICMP\footnote{\url{https://www.rfc-editor.org/rfc/rfc792}}), which is our running example. 
The ping utility consists of a total of three \emph{roles} communicating amongst each other: two roles, \role p and \role r, communicate \emph{reliably} with each other, and both communicate \emph{unreliably} with a third role \role q. 
Our definition of reliability (\cref{subsec:types-reliability}) takes into account the viewpoint of each role, thus allowing roles to have their own (possibly empty) \emph{reliability set}.
Following the assumptions above, the reliability set for \role p is $\{\role r\}$, for \role r is $\{\role p\}$, and for \role q is $\emptyset$.

Below we give the session types, denoted $\tSession_\role{r}$, $\tSession_\role{p}$ and $\tSession_\role{q}$ for roles \role  r, \role p and \role q respectively.
\footnotesize{\begin{align*}
        \begin{array}{c}
            \tSession_\role{r} = \tBranch{\tRcv{p}{ok}{}{\tEnd},\ \tRcv{p}{ko}{}{\tEnd}}
            \\[1em]
            \tSession_\role{p} = \tSnd{q}{ping}{}{
                \tBranchL{
                    \begin{array}{l}
                        \tRcv{q}{pong}{}{\tSnd{r}{ok}{}{\tEnd}}, \\
                        \cTimeOut{\tSnd{q}{ping}{}{
                            \tBranchL{
                                \begin{array}{l}
                                    \tRcv{q}{pong}{}{\tSnd{r}{ok}{}{\tEnd}}, \\
                                    \cTimeOut{\tSnd{q}{ping}{}{
                                        \tBranchL{
                                            \begin{array}{l}
                                                \tRcv{q}{pong}{}{\tSnd{r}{ok}{}{\tEnd}}, \\
                                                \cTimeOut{\tSnd{r}{ko}{}{\tEnd}}
                                            \end{array}
                                        }
                                    }}
                                \end{array}
                            }
                        }}
                    \end{array}
                }
            }
            \\[3em]
            \tSession_\role{q} = \tBranchL{
                \begin{array}{l}
                    \tRcv{p}{ping}{}{\tSnd{p}{pong}{}{\tEnd}}, \\
                    \cTimeOut{\tBranchL{
                        \begin{array}{l}
                            \tRcv{p}{ping}{}{\tSnd{p}{pong}{}{\tEnd}}, \\
                            \cTimeOut{\tBranchL{
                                \begin{array}{l}
                                    \tRcv{p}{ping}{}{\tSnd{p}{pong}{}{\tEnd}}, \\
                                    \cTimeOut{\tEnd}
                                \end{array}}
                            }
                        \end{array}}
                    }
                \end{array}
            }
        \end{array}
    \end{align*}}

Role \role r is the \emph{receiver} ($\&$--called \emph{branching}), which waits on two options: it receives from \role p either the label \msgLabel{ok} or \msgLabel{ko} and then it terminates the protocol ($\tEnd$).
Role \role p is the \emph{sender} ($\oplus$\footnote{For readability, we adopt a shorthand notation for sending towards a 
single role and for payloads of type unit, such that $\tSelect{\tSnd{s}{m}{\tUnit}{\tSession}}$
is represented by $\tSnd{s}{m}{}{\tSession}$.}--called \emph{selection}), and it tries to obtain information on the status of \role q.
It begins by sending a \msgLabel{ping} message to \role q ($\role q \, ! \, \msgLabel{ping}()$), then waits to receive from \role q.
If a \msgLabel{pong} 
is received ($\role q\, ?\, \msgLabel{pong}()$) in the top branch, then it concludes that the status of \role q is \emph{reachable} and sends this information to \role r ($\role r\, !\, \msgLabel{ok}$()), after which it terminates.
Alternatively, \role p enters a \emph{timeout branch} (${\scriptstyle\VarClock}$). For simplicity, we assume \role p will attempt to communicate with \role q three times (shown in the three-time indentation of the timeout branch) before 
assuming \role q is \emph{unreachable}; after which the session will also terminate by sending \msgLabel{ko} to \role r, followed by \tEnd.
In the same lines, the protocol for role \role q is given by the session type $\tSession_\role{q}$, where its timeout branches match the timeouts from $\tSession_\role{p}$.
\end{example}






\subsection{Contributions}
We now present our contributions \wrt\ our Multiparty, Asynchronous, and Generalised \tpi-calculus (\magpi).

\begin{enumerate}
\item
\textbf{MAG\tpi} language (\cref{sec:calculus}):
\begin{itemize}
\item
\magpi\ is \emph{the first language} to support the widest set of non-Byzantine faults, including \role{message loss}, \role{message delays} and \role{message reordering}; \role{crash failures} and \role{link failures}; and \role{network partitioning}.
\item
\magpi\ is \emph{the first language} to introduce \textbf{timeouts} in receive branches
(used for handling network failures), as well as support
\textbf{undirected branching} in a \emph{generalised} setting---the ability to simultaneously expect an 
incoming message {from more than one sender}.
\end{itemize}

\item
\textbf{MAG\tpi} type system (\cref{sec:types}):
\begin{itemize}
\item
is a \emph{conservative extension} of a {generalised} asynchronous 
MPST theory~\cite{DBLP:journals/pacmpl/ScalasY19}, benefiting from: the ability to model \emph{more protocols} than traditional syntactic theories (\eg\ global types); and the flexibility of checking desired properties, such as deadlock freedom or termination, \emph{a posteriori}---as opposed to during the design phase.

\item
supports \textbf{undirected branching/selection} and is the first type system to introduce 
\textbf{timeout branches}.

\item
supports a novel and most general \emph{reliability} definition (\cref{subsec:types-reliability}), taking into account the viewpoint of \emph{each} participating role, and is built on \emph{optional role-dependant} 
{reliability} assumptions.
\end{itemize}

\item
\textbf{Type properties} (\cref{sec:types-props}):
We prove {subject reduction} (\cref{thm:sr}) and 
{session fidelity} (\cref{thm:sf}).
We show failure-handling safety (\cref{cor:failure-handling}) and its inverse, reliability adherence (\cref{cor:rel-adherence}), which strictly connect timeouts and reliability. We prove process properties (\cref{thm:proc-props}) \eg\ deadlock-freedom, termination, 
liveness, in line with~\cite{DBLP:journals/pacmpl/ScalasY19}. Finally, as our \magpi\ type system is Turing-complete, we prove decidable type checking (\cref{thm:decidability}) and decidability of process properties for finite message buffers (\cref{thm:decidable-props}).

\item
\textbf{TCP vs. UDP} (\cref{sec:network_assumptions}):
\magpi\ is expressive enough to capture a range of network assumptions: from low-level network programming operating over the User Datagram Protocol (UDP); to application-level software operating over the Transmission Control Protocol (TCP).

\item
\textbf{Case study} (\cref{sec:examples}):
we further demonstrate the use of timeouts and undirected branching to model a Domain Name System ({DNS}) distributed protocol with a cache and in-built load-balancer; we also show the properties it satisfies, including safety and deadlock freedom.
Further examples are available in \cref{app:examples}, including a prototype specification 
of a leader election algorithm used by consensus protocols.
\end{enumerate}

%% file: sections/2-calculus.tex
We present a multiparty session \tpi-calculus, based on the 
theory of Scalas and Yoshida~\cite{DBLP:journals/pacmpl/ScalasY19}, extended to 
accurately model real-world distributed network environments.
%
We assume the \emph{lowest 
level} of abstraction---the only \emph{failure detection mechanism} available 
to a process is an upper-bound wait limit, \ie\ a \emph{timeout}.

Our calculus presents three novel features:
\begin{enumerate*}[label=\textit{(\roman*)}]
  \item the new \emph{timeout} primitive; 
  \item the capability of expecting a message from \emph{different senders}; and 
  \item operational semantics which can model \emph{various non-Byzantine failures}.
\end{enumerate*}
Timeouts can be attached to receive actions---henceforth referred to as \emph{branches}---and are used to describe an alternative process to be executed in case failures are \emph{assumed} to occur (akin to error handlers).

Failures are said to be \emph{assumed}, as opposed to detected, since we model the impossibility 
result of distinguishing between a delayed \emph{vs} lost message.
Thus, it is possible for a processes to prematurely timeout without its corresponding message 
having been lost---just like the real-world!

The benefit of our approach is that the \emph{failure detection mechanism is agnostic to 
the type of fault}, allowing us to model in unison \textbf{message loss}, \textbf{message delay}, \textbf{crash-stop failures},
\textbf{link failures}, and \textbf{network partitions}.

\subsection{Syntax}
\begin{definition}[Language Syntax]\label{def:calc-syntax}
The Multiparty, Asynchronous and Generalised \tpi-calculus syntax is defined by the grammar in \cref{fig:calc-syntax}.
\begin{figure}[t]
\textnormal{
\begin{align*}
    c &::= \sep x         \sep|\sep \cSessionRole{s}{p}  & (\textit{variable, session w/ role})\\
    d &::= \sep v         \sep|\sep c                    & (\textit{basic value, variable, session w/ role})\\
    w &::= \sep v         \sep|\sep \cSessionRole{s}{p}  & (\textit{basic value, session w/ role}) \\
 P, Q &::= \sep \inaction \sep|\sep \cRestriction{s}{P}  & (\textit{inaction, restriction}) \\
      & \!\sep|\sep \cPar{P}{Q} \sep|\sep \cChoice{P}{Q} & (\textit{composition, non-deterministic choice}) \\
      & \!\sep|\sep \cSnd{c}{q}{m}{d}{P} & (\textit{selection towards role }\role{q})\\
      & \!\sep|\sep \cBranchQuant{c}{\cRcv{q$_i$}{m$_i$}{x_i}{P_i}} & (\textit{reliable branching from roles }\role{q$_i$}) \\
      & \!\sep|\sep \cBranchQuant{c}{\cRcv{q$_i$}{m$_i$}{x_i}{P_i},\;\;\cTimeOut{Q}} & (\textit{branching from roles }\role{q$_i$}\textit{ w/ \textbf{timeout} Q}) \\
      & \!\sep|\sep \cDef{\;D\;}{\;P\;} \sep|\sep \cCall{X}{\tilde{d}} & (\textit{process definition, process call}) \\
      & \!\sep|\sep \cBuffer{s}{\sigma}   & (\textit{session buffer}) \\
    D &::= \sep \cDecl{X}{\tilde{x}}{P} & (\textit{process declaration}) \\
\sigma &::= \cBufferEntry{p}{q}{m}{w}\cBufferCons\sigma \sep|\sep \emptyBuffer & (\textit{session queue, empty session queue})
\end{align*}}
\caption{Syntax for MAG\tpi}
\label{fig:calc-syntax}
\end{figure}
\end{definition}
Communication happens over sessions ($s,s^\prime$) between a number 
of roles (\role{p}, \role{q}) ranging over set \roles.
The primitives of the calculus are \emph{sessions with roles} \cSessionRole{s}{p}, and 
basic values $v$, both of which can be abstracted using \emph{variables} ($x,y$).
Processes ($P,Q$), include the following standard constructs:
\begin{enumerate*}[label=\textit{(\roman*)}]
  \item \emph{inaction} \inaction\ represents process termination;
  \item session \emph{restriction} \cRestriction{s}{\cVar{\!\!P}} binds a new session $s$ in $P$;
  \item parallel \emph{composition} declares two concurrent processes;
  \item \emph{selection} \cSnd{\cVar{c}}{q}{m}{\cVar{d}}{\cVar{P}} uses channel $c$
        to send a message to \role{q} with label \msgLabel{m} and payload \cVar{d}---after sending, the 
        process continues according to $P$;
  \item \emph{definition} and \emph{declaration} allow processes to be assigned names, modelling recursion 
        through the use of process \emph{calls}. 
\end{enumerate*}
We now elaborate on the novelties in our language.
\begin{description}
  \item[\hl{\cBranchQuant{c}{\cRcv{q$_i$}{m$_i$}{\cVar{d}}{\cVar{P}}}}]
    Quantification over roles in a branch allows processes to receive from one in a range of other roles.
    This has practical applications in a multitude of distributed protocols, \eg\ 
    \emph{load balancers}.
    
  \item[\cBranchQuant{c}{\cRcv{q$_i$}{m$_i$}{$d$}{$P$}, \hl{\cTimeOut{$Q$}}}]
    Timeouts are used as a \emph{failure detection mechanism} in receive branches.
    If a failure is assumed to have lost or prevented the incoming message, then process $Q$ is initiated.
    It is key to note that timeouts are non-deterministic---they model an arbitrary and unknown duration of time 
    a process waits before assuming a failure has occurred.
  
  \item[\hl{$P+Q$}] 
    Non-deterministic choice randomly picks between two possible process continuations.
    We use this construct to simplify examples for better presentation.
    Concretely, it replaces the need for \emph{expressions} and \emph{if-then-else}
    constructs, which are routine and orthogonal to our formulation.
  
  \item[\hl{$s:\sigma$}]
    Message buffer for session $s$.
    An entry in the buffer \cBufferEntry{p}{q}{m}{w} models a message ``in transit'' from 
    role \role{p} to \role{q} with label \msgLabel{m} and payload $w$. This is needed to accommodate asynchrony in our language.
\end{description}

\begin{figure}[t]
\begin{tabular}{l l l}
  \rSelection & \cPar{\cSnd{\cSessionRole{s}{p}}{q}{m}{\cVar{w}}{\cVar{P}}}{\cBuffer{s}{$\sigma$}} 
                $\longrightarrow$ 
                \cPar{\cVar{P}}{\cBuffer{s}{$\sigma$ \cBufferCons \cBufferEntry{p}{q}{m}{\cVar{w}} \cBufferCons \emptyBuffer}}\\[0.8em]
  \rBranch    & \cPar{\cBranchQuant{\cSessionRole{s}{q}}{\cRcv{p$_i$}{m$_i$}{\cVar{x_i}}{\cVar{P_i}} [, \cTimeOut{\cVar{Q}}]}}{\cBuffer{s}{\cBufferEntry{p$_k$}{q}{m$_k$}{\cVar{w}} \cBufferCons $\sigma$}} \\[0.2em]
              & $\Sep\longrightarrow$
                \cPar{\cVar{P_k[^w/_{x_k}]}}{\cBuffer{s}{$\sigma$}} & for $k \in I$\\[0.8em]
  \rTimeOut   & \cPar{\cBranchQuant{\cSessionRole{s}{q}}{\cRcv{p$_i$}{m$_i$}{\cVar{x_i}}{\cVar{P_i}}, \cTimeOut{\cVar{Q}}}}{\cBuffer{s}{$\sigma$}} 
                $\longrightarrow$ 
                \cPar{\cVar{Q}}{\cBuffer{s}{$\sigma$}}\\[0.8em]
  \rChoice    & \cChoice{\cVar{P_1}}{\cVar{P_2}} $\longrightarrow$ \cVar{P_i} & for $i \in \{1,2\}$\\[0.8em]
  \rCall      & \cDef{\cDecl{X}{\cVar{x_1},$\dotsc$,\cVar{x_n}}{\cVar{P}}}{(\cPar{\cCall{X}{$w_1$,$\dotsc$,$w_n$}}{\cVar{Q}})} \\[0.2em]
              & $\Sep\longrightarrow$
                \cDef{\cDecl{X}{\cVar{x_1},$\dotsc$,\cVar{x_n}}{\cVar{P}}}{(\cPar{\cVar{P}[$^{w_1}$/$_{x_1}$]$\,\dotsb$[$^{w_n}$/$_{x_n}$]}{\cVar{Q}})} \\[0.8em]
  \rCtx       & \cVar{P} $\longrightarrow$ \cVar{P^\prime} $\implies$ \ctx[\cVar{P}] $\longrightarrow$ \ctx[\cVar{P^\prime}] \\[0.8em]
  \rDrop      & \cBuffer{s}{\cVar{h} \cBufferCons $\sigma$} $\longrightarrow$ \cBuffer{s}{$\sigma$}
\end{tabular}
\caption{Reduction rules for MAG\tpi}
\label{fig:calc-semantics}
\end{figure}

\subsection{Operational Semantics}

We begin with definitions of a reduction context and buffer congruence.

\begin{definition}[Reduction Context]\label{def:calc-semantics}
A \textbf{reduction context}\ \ctx\ abstracts away an outer environment from a process, 
and is given by:
\[
  \ctx\ ::=\ \cPar{\ctx}{P}\sep |\sep \cRestriction{s}{\ctx}\sep |\sep \textnormal{\cDef{$D$}{\ctx}}\sep |\sep [\;]
\]
Hence, $\ctx [P]$ refers to process $P$ under some arbitrary context \ctx.
\end{definition}

\begin{definition}[Buffer Congruence]
\label{def:buffer_cong}
A process containing \textbf{only} a buffer under its restriction is congruent to inaction.
Message buffers observe total reordering.
\[
  \cRestriction{s}{\cBuffer{s\!}{\!\sigma}} \;\equiv\; \inaction 
  \Sep\Sep\Sep\Sep
  \cBuffer{s\!}{\!\sigma_1\cBufferCons h_1 \cBufferCons h_2 \cBufferCons\sigma_2} \;\equiv\; \cBuffer{s\!}{\!\sigma_1\cBufferCons h_2 \cBufferCons h_1 \cBufferCons\sigma_2}
\]
\end{definition}

\begin{definition}[OS]
The operational semantics for \magpi\ is given via a \emph{reduction relation} $\longrightarrow$ inductively defined in \cref{fig:calc-semantics}, together with standard \emph{structural congruence rules} \cite{DBLP:journals/pacmpl/ScalasY19} and two \emph{buffer congruence rules} defined in \cref{def:buffer_cong}.
\end{definition}

Let us now comment on the reduction rules (\cref{fig:calc-semantics}).
Processes send messages using the \textbf{selection rule} \rSelection; this adds the sent 
message as a new entry in the session buffer, and advances the process to its continuation.
Sent messages are read from the buffer using the \textbf{branching rule} \rBranch. 
If the receiver has a valid branch matching the sender and message label, then it advances to 
the specific continuation of said branch (a timeout branch for this rule is optional).
The substitution \cVar{P_k[^w/_{x_k}]} denotes the replacement of variable $x_k$ with the payload 
value $w$ in the continuation process $P_k$.
The \textbf{timeout rule} \rTimeOut\ advances processes to their timeout branch \emph{without} changing 
the buffer.
Non-deterministic choice is resolved using the \textbf{choice rule} \rChoice, which advances the process 
to one of the two possible continuations.
The \textbf{call rule} \rCall\ replaces a process call with its defined process, substituting each parameter.
Processes can reduce under a context using the \textbf{context rule} \rCtx.
Lastly, messages can be lost from the buffer with the \textbf{drop rule} \rDrop.

We now unpack how our semantics deals with failures. The reduction rules in \cref{fig:calc-semantics} allow various forms of failures to be modelled, stemming from the versatility and elegance of the drop rule \rDrop.
The following elaborates on how this rule can be utilised to model different types of failure:
\begin{itemize}
  \item \role{Message loss} is modelled directly through the reduction rule \rDrop.
  \item \role{Crash-failure} is modelled through repeated applications of \rDrop\ for a particular role.
        \Eg, to model a crash of role \role{p}, the reduction step \rDrop\ should be applied 
        to all messages that enter the buffer matching the pattern \cBufferEntry{p}{\_}{\_}{\_} 
        (\_~symbolises a ``don't care'' value).
  \item \role{Link-failure} is modelled using a similar method; the difference being that messages between 
        \emph{two} specific recipients are dropped.
        \Eg, modelling a link-failure between roles \role{p} and \role{q} requires \rDrop\ to be applied 
        to all messages entering the buffer with the patterns \cBufferEntry{p}{q}{\_}{\_} and 
        \cBufferEntry{q}{p}{\_}{\_}.
  \item \role{Message delay} is modelled by applying rule \rTimeOut\ to a branch whilst a valid message 
        resides in the buffer.
        \Eg:
        \begin{align*}
          &\cPar{\cBranchQuant{\cSessionRole{s}{q}}{\cRcv{p$_i$}{m$_i$}{x_i}{P_i},\ \cTimeOut{Q}}\ }{\ \cBuffer{s\!}{\!\cBufferEntry{p$_k$}{q}{m$_k$}{w} \cBufferCons \sigma}} \\
          &\Sep\longrightarrow\ 
          \cPar{Q\ }{\ \cBuffer{s\!}{\!\cBufferEntry{p$_k$}{q}{m$_k$}{w} \cBufferCons \sigma\:}} &\text{for}\ k \in I.
        \end{align*}
  \item Total \role{message reordering} is modelled via \emph{buffer congruence rules} (\cref{def:buffer_cong}).
  \item \role{Network partitions} can be represented using multiple \emph{link failures}.
\end{itemize}

The granularity at which we model failures allows for degrees of customisation.
\Eg, benign fault-tolerant consensus algorithms typically assume the possibility of \emph{all} non-Byzantine 
faults, therefore all the aforementioned failures are required.
Alternatively, an application assumed to run over a trusted TCP network need not worry about single message drops, 
and hence \rDrop\ should only be applied to model node crash and link failures.

\begin{definition}[Well-formedness]\label{def:well-formed}
  To ensure that communication is possible, we require that a well-formed process has a 
  buffer for each session, \ie
  \textnormal{\[
    P = \cRestriction{s}{Q} \implies Q\equiv \cRestriction{\tilde{s^\prime}}{(\cPar{Q^\prime\:}{\cBuffer{\:s\!}{\!\sigma}})}
  \]}
\end{definition}

\Cref{def:well-formed} introduces a well-formedness condition to guarantee that 
a session always guards its buffer, hence ensuring that messages always have a queue 
to be placed in. From now on, we will only consider well-formed processes.

Before concluding this section, we recall our ping pong running example from the introduction, and present below the processes for roles \role p, \role q and \role r.
\begin{example}[Ping Pong: Processes]\label{ex:ping-procs}
\footnotesize{\begin{align*}
        \begin{array}{c}
            P_\role{p} = \cSnd{\cSessionRole{s}{p}}{q}{ping}{}{
                \cBranchL{\cSessionRole{s}{p}}{
                    \begin{array}{l}
                        \cRcv{q}{pong}{}{P_\role{p}^\textit{ok}}, \\
                        \cTimeOut{\cSnd{\cSessionRole{s}{p}}{q}{ping}{}{
                            \cBranchL{\cSessionRole{s}{p}}{
                                \begin{array}{l}
                                    \cRcv{q}{pong}{}{P_\role{p}^\textit{ok}}, \\
                                    \cTimeOut{\cSnd{\cSessionRole{s}{p}}{q}{ping}{}{
                                        P^\prime_\role{p}
                                    }}
                                \end{array}
                            }
                        }}
                    \end{array}
                }
            }
            \\[2em]
            P_\role{p}^\textit{ok} = \cSnd{\cSessionRole{s}{p}}{r}{ok}{}{\inaction}
            \Sep
            P_\role{p}^\prime = \cBranchL{\cSessionRole{s}{p}}{
                \begin{array}{l}
                    \cRcv{q}{pong}{}{P_\role{p}^\textit{ok}}, \\
                    \cTimeOut{\cSnd{\cSessionRole{s}{p}}{r}{ko}{}{\inaction}}
                \end{array}
            }
            \\[2em]
            P_\role{q} = \cBranchL{\cSessionRole{s}{q}}{
                \begin{array}{l}
                    \cRcv{p}{ping}{}{P_\role{q}^\textit{pong}},\\
                    \cTimeOut{\cBranchL{\cSessionRole{s}{q}}{
                        \begin{array}{l}
                            \cRcv{p}{ping}{}{P_\role{q}^\textit{pong}}, \\
                            \cTimeOut{\cBranch{\cSessionRole{s}{q}}{
                                \cRcv{p}{ping}{}{P_\role{q}^\textit{pong}},\ 
                                \cTimeOut{\inaction}
                            }}
                        \end{array}
                    }}
                \end{array}
            }
            \\[2em]
            P_\role{q}^\textit{pong} = \cSnd{\cSessionRole{s}{q}}{p}{pong}{}{\inaction}
            \\[1em]
            P_\role{r} = \cBranch{\cSessionRole{s}{r}}{\cRcv{p}{ok}{}{\inaction}, \cRcv{p}{ko}{}{\inaction}}
        \end{array}
    \end{align*}}
\end{example}

%% file: sections/3-types.tex

We introduce the type system for MAG\tpi,
which is a conservative extension of the \emph{generalised} asynchronous 
MPST theory~\cite[sec. 7]{DBLP:journals/pacmpl/ScalasY19}.
\emph{Generalised} MPST stray away from global protocol specifications 
(global types) and instead operate on user-defined localised specifications of each 
participating role (local types).
The benefits of working with such theory include:
\begin{enumerate*}[label=\textit{(\roman*)}]
\item
the ability to capture a larger set of viable protocols compared to traditional syntactic methods (\eg\ global types) of enforcing consistent communication;
\item
the ability to model protocols of different requirements. 
\end{enumerate*}
In particular, instead of syntactically enforcing programmers to write, \eg, deadlock-free code, a generalised theory allows programmers to unrestrictedly design protocols that are checked \emph{a posteriori} against any number of required properties, such as deadlock-freedom, termination \etc\

The novelties of our type system include:
\begin{enumerate*}[label=\textit{(\roman*)}]
  \item \emph{undirected branching/selection};
  \item \emph{timeout branches} (syntax in \cref{subsec:types-syntax}); and
  \item \emph{reliability sets}---sets of roles assumed to not fail, from the perspective of \emph{each} role (\cref{subsec:types-reliability}). 
\end{enumerate*}
Reliability sets (possibly empty) enforce the use of \emph{timeouts} for all failure-prone  
communication.

As in~\cite{DBLP:journals/pacmpl/ScalasY19}, our type system does \emph{not} use
global types, but solely relies on local types.
Consequently, typing contexts must obey a safety property to ensure subject 
reduction (\cref{subsec:types-contexts}).
Finally, we present the rules for our type system in \cref{subsec:types-rules}, and 
discuss its key properties in \cref{subsec:types-props}.

\subsection{Types}\label{subsec:types-syntax}
Our MPST theory is designed for the distributed 
computing setting.
Concretely, our type system (\cref{def:type-syntax}) is \emph{asynchronous}; it allows 
\emph{branching} (resp. \emph{selection}) from (resp. to) multiple roles; and supports 
\emph{timeout} continuation types.

\begin{definition}[Typing syntax]\label{def:type-syntax}
The typing syntax is defined using the grammar in \cref{session_types}.
\begin{figure}[t]
\centering
\begin{minipage}[t]{0.55\textwidth}
    \textbf{Basic and Session Types}\\[1em]
    \begin{tabular}{r c l}
        \tType    & ::= & \tBasic\ $|$ \tSession\ \\[0.5em]
        \tBasic   & ::= & \tInt\ $|$ \tBool\ $|$ \tReal\ $|$ \tUnit\ $|$ $\dotsb$ \\[0.5em]
        \tSession & ::= & \tBranchQuant{\tRcv{\role{p$_i$}}{\msgLabel{m$_i$}}{\tType$_i$}{\tSession$_i$}[, \cTimeOut{\tSession$^\prime$}]} \\[0.5em]
                    & $|$ & \tSelectQuant{\tSnd{\role{p$_i$}}{\msgLabel{m$_i$}}{\tType$_i$}{\tSession$_i$}} \\[0.5em]
                    & $|$ & \tRec{t}{\tSession} $|$ \recVar{t} $|$ \tEnd
    \end{tabular}
\end{minipage}
\begin{minipage}[t]{0.35\textwidth}
    \textbf{Buffer Types}\\[1em]
    \begin{tabular}{r c l}
        \tQueue   & ::= & \tBuffer{p}{m}{\tType}{\tQueue} $|$ \emptyBuffer\ \\[2em]
    \end{tabular}

    \textbf{Session-Buffer Types}\\[1em]
    \begin{tabular}{r c l}
        \tSessionQueue & ::= & \tQueue\ $|$ \tSession\ $|$ \tTuple{\tQueue}{\tSession}\\[0.5em]
    \end{tabular}
\end{minipage}
\caption{Basic Types, Session Types, Buffer Types and Session-Buffer Types}
\label{session_types}
\end{figure}
\begin{figure}[t]
    \begin{mathpar}
    \inferrule
        {\ }
        {\tType \equiv \tType}
    \and
    \inferrule
        {\ }
        {\tQueue_1\cBufferCons\tQueue_2 \equiv \tQueue_2\cBufferCons\tQueue_1}
    \and
    \inferrule
        {\ }
        {\emptyBuffer\cBufferCons\emptyBuffer \equiv \emptyBuffer}
    \and
    \inferrule
        {\tQueue \equiv \tQueue^\prime \\ \tSession \equiv \tSession^\prime}
        {\tTuple{\tQueue}{\tSession} \equiv \tTuple{\tQueue^\prime}{\tSession^\prime}}
    \end{mathpar}
    \caption{Type congruence rules}
    \label{fig:type-cong}
\end{figure}
For undirected branching and selection, $I \neq \emptyset$ and role-label tuples 
\textnormal{$(\role{p$_i$},\msgLabel{m$_i$})$} must be pairwise distinct.
Recursion variables cannot be free and must appear guarded under branching/selection types.
\end{definition}
Type \tType\ denotes either a \emph{basic} type \tBasic, or a 
\emph{session type} \tSession, and is used to type variables.
Session types describe how a channel should be used: 
\begin{enumerate*}[label=(\textit{\roman*})]
    \item \emph{undirected branching} (external choice) 
          \tBranchQuant{\tRcv{\role{p$_i$}}{\msgLabel{m$_i$}}{\tType$_i$}{\tSession$_i$}[, \cTimeOut{\tSession$^\prime$}]}
          denotes \emph{receiving} a message with label \msgLabel{m$_i$} and payload of type \tType$_i$ 
          from role \role{p$_i$}, then continuing according to \tSession$_i$. 
          The (optional) timeout continuation type \tSession$^\prime$ describes the protocol for handling 
          failure on that branch;
    \item \emph{undirected selection} (internal choice) 
          \tSelectQuant{\tSnd{\role{p$_i$}}{\msgLabel{m$_i$}}{\tType$_i$}{\tSession$_i$}}
          denotes \emph{sending} a message with label \msgLabel{m$_i$} and payload \tType$_i$
          to role \role{p$_i$}, then continuing according to \tSession$_i$;
    \item type \tEnd\ marks a channel as closed, and terminates communication.
\end{enumerate*}
A session buffer is typed using the \emph{buffer type} \tQueue.
Entries in the buffer must correspond to the type \tBuffer{p}{m}{\tType}{\tQueue}, denoting 
a message sent to \role{p} with label \msgLabel{m} and payload of type \tType.
A session with role is typed using \emph{session-buffer} types, combining
a session type and a buffer type.

Type \emph{congruence} $\equiv$ is defined in \cref{fig:type-cong}.
Notably, buffer types can be reordered,
and two session-buffer types are congruent if their individual buffer and session types are congruent.
Buffer type reordering is necessary to match the \emph{total message reordering} 
supported by the language (\cref{def:buffer_cong}).

\subsection{Reliability}\label{subsec:types-reliability}
We go on a short detour and talk about reliability.
Previous related work~\cite{DBLP:journals/corr/abs-2207-02015,DBLP:conf/forte/AdameitPN17,DBLP:journals/pacmpl/VieringHEZ21} 
have included the notion of \emph{reliability} into their type systems.
Generally, either one specific role, or a pre-defined set of roles, are 
assumed to be reliable---\ie\ no failures occur for communication involving the identified set of roles.

Our definition of reliability (\cref{def:reliability}) is the most general and the first to take into account the viewpoint of each role.
We argue that this is necessary in a distributed setting since reliability in networks is 
dependant on the physical topology of processes.
Recalling the ping utility (\cref{ping-pong}), we could imagine the processes representing roles \role{p} and \role{r} reside on the same physical hardware, thus their communication cannot be affected by network faults; and the process for \role{q} resides on geographically separated hardware, therefore its communication with both \role{p} and \role{r} is vulnerable to failure.


\begin{definition}[Reliability]\label{def:reliability}
The \textbf{reliability set} \relSet\ for a role \textnormal{$\role{p}\in\roles$} is defined as 
\textnormal{$\relSet \subseteq \roles\setminus\{\role{p}\}$}, capturing the viewpoint of \textnormal{$\role{p}$}.
\textbf{Reliability} \rel\ is defined as a function mapping roles to their reliability set, 
\ie\ \textnormal{$\rel : \roles \to \relSet$}.
\end{definition}

To better model real distributed environments, our definition of reliability allows each role to 
have its own (possibly empty) reliability set.
\begin{example}[Ping Pong: Reliability Sets]\label{ex:ping-rel}
\Wrt\ \cref{ping-pong}, as the three roles have different viewpoints on each other, then the reliability set for each of them is different. In particular, we have
$\rel(\role{p}) = \{\role r\},\  \rel(\role{r}) = \{\role p\}, \ \rel(\role{q}) = \emptyset$.
\end{example}
%
Investigating the extremes, we have:
for a set of roles \roles, {if} for all $\role{p}\in\roles\cdot\rel(\role{p})=\emptyset$, 
{then} \emph{no} communication is reliable;
conversely, {if} for all $\role{p}\in\roles\cdot\rel(\role{p})=\roles\setminus\{\role{p}\}$, 
{then} \emph{all} communication is reliable---referred to as a \emph{reliable network}.
%
%
This work only considers static configurations for \rel, thus reliability sets 
cannot change at runtime.
We find that even with this restriction, our definition is the most general compared to 
related work.

\subsection{Contexts}\label{subsec:types-contexts}
%

\begin{definition}[Type contexts]\label{def:ctx-syntax}
Context \contO\ is a partial mapping from process variables to $n$-tuples of types and context \contG\ is a partial mapping from variables to types, and sessions with roles to 
session-buffer types, both defined below:
\textnormal{
    \[
        \contO ::= \;\emptyset \sep|\sep \contO,\contV{X}{\tType_1,\dotsc,\tType_n} \Sep\Sep\Sep
        \contG ::= \;\emptyset \sep|\sep \contG,\contV{x}{\tType} \sep|\sep \contG,\contV{\cSessionRole{s}{p}}{\tSessionQueue}
    \]
}\noindent
The \textbf{composition} of contexts $(\contGb{1},\contGb{2})$ is defined iff:
\textnormal{\[
    \forall c \in \dom(\contGb{1}) \cap \dom(\contGb{2}) : \sep \contGb{i}(c) = \tQueue\ \land\ \contGb{j}(c) = \tSession 
\]}\noindent
For such $c$, \textnormal{(\contGb{1},\contGb{2})(c) = \tTuple{\tQueue}{\tSession}}.

\noindent
Contexts are \textbf{congruent} $\contGb{1} \equiv \contGb{2}$ iff: 
\textnormal{\[
    \dom(\contGb{1}) = \dom(\contGb{2})\ \land\ \forall c \in \dom(\contGb{1}): \contGb{1}(c) \equiv \contGb{2}(c)
\]}
\end{definition}
Context \contO\ is \emph{non-linear} and types \emph{process variables} by tracking the types of 
their parameters.
Context \contG\ is \emph{linear} and allows \emph{variables} to have \emph{basic} or \emph{session types}, 
and \emph{sessions with roles} to have \emph{session-buffer types}; as a program 
progresses, a role may simultaneously have both an active session type and 
messages queued in the message buffer.

Context \emph{composition} allows two contexts to coexist as long as their common 
channels map to buffer types in one context, and session types in the other.

Context \emph{congruence} holds if two contexts have the same domain and the types 
of their channels are congruent.
It is key to note that by the definitions of context composition and congruence we have
$\contV{\cSessionRole{s}{p}}{\tTuple{\tQueue}{\tSession}} \equiv 
\contV{\cSessionRole{s}{p}}{\tQueue},\contV{\cSessionRole{s}{p}}{\tSession}$.
Buffer types (resp. session-buffer types) are only used internally by the type 
system; end-users are not expected to explicitly define these types.


\begin{definition}[Context reduction]\label{def:ctx-reduction}
An \textbf{action} $\alpha$ is given as:
\textnormal{\[
    \alpha\ ::=\ \contSnd{s}{p}{q}{m}{\tType} \sep|\sep \contCom{s}{p}{q}{m} \sep|\sep \contTimeout{s}{p}
\]}\noindent
From left to right, this reads as 
\begin{enumerate*}[label=\textnormal{(\textit{\roman*})}]
    \item a \textbf{sent message};
    \item \textbf{communication} of a message; and
    \item the \textbf{timeout} of a channel.
\end{enumerate*}
\textbf{Context transition} $\xrightarrow{\alpha}_{(\bufferTracker;\rel)}$ is defined in 
\cref{fig:ctx-reduction}.
We write $\contG\xrightarrow{\alpha}_{(\bufferTracker;\rel)}\ $ iff 
$\ \exists\,\contGp : \contG\xrightarrow{\alpha}_{(\bufferTracker;\rel)}\contGp$.
We define two \textbf{context reductions} $\rightarrow_{(\bufferTracker;\rel)}$ and $\rightarrow_\bufferTracker$ as: 
\textnormal{
\begin{mathpar}
    \contG\rightarrow_{(\bufferTracker;\rel)}\contGp \textit{ holds iff } \contG\xrightarrow{\alpha}_{(\bufferTracker;\rel)} \contGp
    \and\
    \contG\rightarrow_{\bufferTracker}\contGp \textit{ holds iff } \contG\xrightarrow{\alpha}_\bufferTracker\contGp \textit{ for } \alpha \in \{\contSnd{s}{p}{q}{m}{\tType},\ \contCom{s}{p}{q}{m}\}
\end{mathpar}
}\noindent
We write $\rightarrow_{(\bufferTracker;\rel)}^+$ (resp. $\rightarrow_{\bufferTracker}^+$)
and $\rightarrow_{(\bufferTracker;\rel)}^*$ (resp. $\rightarrow_{\bufferTracker}^*$) for their transitive 
and reflexive/transitive closures respectively.
\end{definition}

A context \contG\ keeps track of open buffers using a \emph{buffer-tracker} \bufferTracker.
Whenever a new session is initialised, it is added to \bufferTracker,
details in \cref{subsec:types-rules} item \trNew. 
For now it suffices to know that 
buffer trackers restrict communication to occur only over restricted sessions, thus 
by \cref{def:well-formed} (well-formedness), it guarantees that a session buffer exists for all sessions in \bufferTracker.

\begin{figure}[t]
    \begin{mathpar}
    \inferrule [\crrTO]
        {s \in \bufferTracker \\ \exists\, k \in I : \role{q$_k$} \not\in \rel(\role{p})}
        {\cSessionRole{s}{p}:\!\tBranchQuant{\tRcv{q$_i$}{m$_i$}{\tType_i}{\tSession_i},\ \cTimeOut{\tSession^\prime}} \xrightarrow{\contTimeout{s}{p}}_{(\bufferTracker;\rel)} \cSessionRole{s}{p}:\tSession^\prime}
    \and
    \inferrule [\crrSnd{1}]
        {s \in \bufferTracker \\ k \in I}
        {\cSessionRole{s}{p}:\tSelectQuant{\tSnd{q$_i$}{m$_i$}{\tType_i}{\tSession_i}} \xrightarrow{\contSnd{s}{p}{q$_k$}{m$_k$}{\tType_k}}_{(\bufferTracker;\rel)} \cSessionRole{s}{p}:\tTuple{\tBuffer{q$_k$}{m$_k$}{\tType_k}{\emptyBuffer}}{\tSession_k}}
    \and
    \inferrule [\crrSnd{2}]
        {s \in \bufferTracker \\ k \in I}
        {\cSessionRole{s}{p}:\tTuple{\tQueue}{\tSelectQuant{\tSnd{q$_i$}{m$_i$}{\tType_i}{\tSession_i}}} \xrightarrow{\contSnd{s}{p}{q$_k$}{m$_k$}{\tType_k}}_{(\bufferTracker;\rel)} \cSessionRole{s}{p}:\tTuple{\tQueue\cBufferCons\tBuffer{q$_k$}{m$_k$}{\tType_k}{\emptyBuffer}}{\tSession_k}}
    \and
    \inferrule [\crrCom]
        {s \in \bufferTracker \\ \exists\, k \in I : (\role{p},\msgLabel{m},\tType) = (\role{p$_k$},\msgLabel{m$_k$},\tType_k)}
        {\cSessionRole{s}{p}:\tBuffer{q}{m}{\tType}{\tQueue},\ \cSessionRole{s}{q}:\!\tBranchQuant{\tRcv{p$_i$}{m$_i$}{\tType_i}{\tSession_i}\ [,\cTimeOut{\tSession^\prime}]} \xrightarrow{\contCom{s}{p}{q}{m}}_{(\bufferTracker;\rel)} \cSessionRole{s}{p}:\tQueue,\ \cSessionRole{s}{q}:\tSession_k}
    \and
    \inferrule [\crrRec]
        {\cSessionRole{s}{p}:\tSession[^{\tRec{t}{\tSession}}/_\recVar{t}] \xrightarrow{\alpha}_{(\bufferTracker;\rel)} \contGp}
        {\cSessionRole{s}{p}:\tRec{t}{\tSession} \xrightarrow{\alpha}_{(\bufferTracker;\rel)} \contGp}
    \and
    \inferrule [\crrCong]
        {\contGb{1}\xrightarrow{\alpha}_{(\bufferTracker;\rel)} \contGb{2}}
        {\contGb,\contGb{1} \xrightarrow{\alpha}_{(\bufferTracker;\rel)} \contG,\contGb{2}}
    \end{mathpar}
    \caption{Context reduction rules}
    \label{fig:ctx-reduction}
\end{figure}

Context reduction (\cref{def:ctx-reduction}) models communication at the \emph{type-level}.
Context \contG\ can reduce by sending, communicating, or timing out.
By \crrTO, $\contG = \cSessionRole{s}{p}:\!\tBranchQuant{\tRcv{q$_i$}{m$_i$}{\tType_i}{\tSession_i},\ \cTimeOut{\tSession^\prime}}$ 
can reduce to a \emph{timeout} branch continuation type $\tSession^\prime$ if $s$ is 
in the buffer-tracker (\ie\ a buffer exists for session $s$), 
and \emph{at least one} of the roles in the branch is \emph{unreliable}.
The latter prevents taking a timeout for communication that is sure to be delivered.
Reductions \crrSnd{1} and \crrSnd{2} simulate sending a message by reducing the 
selection type $\tSelectQuant{\tSnd{q$_i$}{m$_i$}{\tType_i}{\tSession_i}}$ to 
one of its continuations $\tSession_i$, and by inserting the sent message into 
the buffer type.
The difference is that \crrSnd{1} creates the buffer type if it was 
previously not specified, whereas 
\crrSnd{2} appends the message to an already existing buffer type.
Communication between two roles is simulated through \crrCom, 
where a branch type $\cSessionRole{s}{q}:\!\tBranchQuant{\tRcv{p$_i$}{m$_i$}{\tType_i}{\tSession_i}\ [,\cTimeOut{\tSession^\prime}]}$
consumes the message from a buffer type $\cSessionRole{s}{p}:\tBuffer{q}{m}{\tType}{\tQueue}$,
reducing to the continuations $\cSessionRole{s}{p}:\tQueue,\ \cSessionRole{s}{q}:\tSession_k$.
Lastly, \crrRec\ allows reduction through recursion and \crrCong\ reduces substructures of 
compatibly composed contexts.

\begin{definition}\label{def:safety-prop}
Property \fSafe\ is a $(\bufferTracker;\rel)$-\textbf{safety} property on typing contexts iff: 

\textnormal{\\
\begin{tabular}{l l}
    \spR{1} & $\fSafe(\contG,\ \cSessionRole{s}{p}:\tBranchQuant{\tRcv{q$_i$}{m$_i$}{\tType_i}{\tSession_i}}) \ \implies \ \forall i \in I : \role{q$_i$} \in \rel(\role{p})$ \\[0.2em]
    \spR{2} & $\fSafe(\contG,\ \cSessionRole{s}{p}:\tBranchQuant{\tRcv{q$_i$}{m$_i$}{\tType_i}{\tSession_i,\ \cTimeOut{\tSession^\prime}}}) \ \implies \ \exists i \in I : \role{q$_i$} \not\in \rel(\role{p})$ \\[0.2em]
    \spCom & $\fSafe(\contG,\ \contV{\cSessionRole{s}{p}}{\tBranchQuant{\tRcv{q$_i$}{m$_i$}{\tType_i}{\tSession_i}\,[,\, \cTimeOut{\tSession^\prime}]}},\, \contV{\cSessionRole{s}{q}\!}{\!\tQueue})$ \\
           & $\Sep\text{and } \tQueue \equiv \tBuffer{p}{m}{\tType}{\tQueue^\prime}$ \\
           & $\Sep\text{and } \exists\, k \in I : \role{q$_k$} = \role{q} \land \msgLabel{m$_k$} = \msgLabel{m} \ \implies \ \tType_k = \tType$ \\[0.2em]
    \spRec & $\fSafe(\contG,\ \cSessionRole{s}{p}:\tRec{t}{\tSession}) \ \implies \ \fSafe(\contG,\ \cSessionRole{s}{p}:\tSession[^{\tRec{t}{\tSession}}/_\recVar{t}])$ \\[0.2em]
    \spRed & $\fSafe(\contG) \text{ and } \contG\rightarrow_{(\bufferTracker;\rel)}\contGp \ \implies \ \fSafe(\contGp)$
\end{tabular}
}
\end{definition}

As previously mentioned, our type system is a \emph{generic} one that does not use 
syntactic methods of enforcing consistent communication.
Therefore, we define a \emph{safety} property in \cref{def:safety-prop} on type 
contexts that is used to guarantee subject reduction and other theorems (presented in \cref{sec:types-props}).

We say \fSafe\ is the \emph{largest} safety property required to guarantee subject 
reduction.
The property can be re-instantiated with more specific conditions (as 
demonstrated in \cref{sec:network_assumptions}) as per the requirements of 
the implementation.
Concretely, \spR{1} and \spR{2} ensure that timeouts are only not defined if 
communication is reliable and that timeouts are defined if communication is 
unreliable respectively.
Condition \spCom\ ensures that communicating messages have matching payload types.
Lastly, \spRec\ preserves \fSafe\ through recursion unfolding and \spRed\ requires 
safety to hold after context reduction.

\subsection{Typing Rules}\label{subsec:types-rules}
\begin{figure*}
\begin{mathpar}
\inferrule [\trInaction]
    {\fend(\contG)}
    {\contO\cBufferCons\contG\ \vdash\ \cInaction}
\and
\inferrule [\trVar]
    {\ }
    {c:\tType\ \vdash\ c:\tType}
\and
\inferrule [\trVal]
    {v \in \tBasic}
    {\emptyset\ \vdash\ v:\tBasic}
\and
\inferrule [\trX]
    {\contO(X)=\tType_1,\dotsc,\tType_n}
    {\contO\ \vdash\ X:\tType_1,\dotsc,\tType_n}
\and
\inferrule [\trSelection]
    {\contGb{1}\ \vdash\ c:\tSelectQuant{\tSnd{q$_i$}{m$_i$}{\tType_i}{\tSession_i}} \\ k \in I \\ \contGb{2}\ \vdash\ d:\tType_k \\ \contO\cdot\contG,c:\tSession_k\ \vdash\ P}
    {\contO\cdot\contG,\contGb{1},\contGb{2}\ \vdash\ \cSnd{c}{q$_k$}{m$_k$}{d}{P}}
\and
\inferrule [\trBranch]
    {\contGp\ \vdash\ c:\tBranchQuant{\tRcv{p$_i$}{m$_i$}{\tType_i}{\tSession_i}\,[,\ \cTimeOut{\tSession^\prime}]} \\ [\contO\cdot\contG,c:\tSession^\prime\ \vdash\ Q] \\ \forall i \in I \cdot \contO\cdot\contG, x_i:\tType_i, c:\tSession_i\ \vdash\ P_i}
    {\contO\cdot\contG,\contGp\ \vdash\ \cBranchQuant{c}{\cRcv{p$_i$}{m$_i$}{x_i}{P_i}\,[,\ \cTimeOut{Q}]}}
\and
\inferrule [\trCall]
    {\contO\ \vdash\ X:\tType_1,\dotsc,\tType_n \\ \fend(\contGp) \\ \forall i \in 1..n \cdot \contGb{i}\ \vdash\ d_i:\tType_i}
    {\contO\cdot\contGb{1},\dotsc,\contGb{n},\contGp\ \vdash\ \cCall{X}{d_1,\dotsc,d_n}}
\and
\inferrule [\trDef]
    {\contO,X:\tType_1,\dotsc,\tType_n\cdot x_1:\tType_1,\dotsc,x_n:\tType_n\ \vdash\ P \\ \contO,X:\tType_1,\dotsc,\tType_n\cdot\contG\ \vdash\ Q}
    {\contO\cdot\contG\ \vdash\ \cDef{\ \cDecl{X}{x_1:\tType_1,\dotsc,x_n:\tType_n}{P\ }}{\ Q}}
\and
\inferrule [\trChoice]
    {\contO\cdot\contG\ \vdash\ P_1 \\ \contO\cdot\contG\ \vdash\ P_2}
    {\contO\cdot\contG\ \vdash\ \cChoice{P_1}{P_2}}
\and
\inferrule [\trLift]
    {\contO\cdot\contG\ \vdash\ P}
    {\contO\cdot\contG\ \vdash_{\emptyset}\ P}
\and
\inferrule [\trEmpty]
    {\fPayloads(\contG)}
    {\contO\cdot\contG\ \vdash_{\{s\}}\ \cBuffer{s}{\emptyBuffer}}
\and
\inferrule [\trBuffer{1}]
    {\contO\cdot\contGp \vdash_{\{s\}}\ \cBuffer{s}{\sigma} \\ \contG\ \vdash\ w:\tType}
    {\contO\cdot\contG,\contGp,\cSessionRole{s}{p}:\tBuffer{q}{m}{\tType}{\emptyBuffer}\ \vdash_{\{s\}}\ \cBuffer{s}{\cBufferEntry{p}{q}{m}{w}\cBufferCons\sigma}}
\and
\inferrule [\trBuffer{2}]
    {\contO\cdot\contGp,\cSessionRole{s}{p}:\tQueue\ \vdash_{\{s\}}\ \cBuffer{s}{\sigma} \\ \contG\ \vdash\ w:\tType}
    {\contO\cdot\contG,\contGp,\cSessionRole{s}{p}:\tBuffer{q}{m}{\tType}{\tQueue}\ \vdash_{\{s\}}\ \cBuffer{s}{\cBufferEntry{p}{q}{m}{w}\cBufferCons\sigma}}
\and
\inferrule [\trBuffer{w}]
    {\contG = (\contGb{0}\msgInsert\contGb{1}),\contGb{2} \\ \contO\cdot\contGb{1}\ \vdash_\bufferTracker\ \cBuffer{s}{\sigma} \\ \fPayloads(\contGb{0},\contGb{2})}
    {\contO\cdot\contG\ \vdash_\bufferTracker\ \cBuffer{s}{\sigma}}
\and
\inferrule [\trPar]
    {\contO\cdot\contGb{1}\ \vdash_{\bufferTracker_1}\ P_1 \\ \contO\cdot\contGb{2}\ \vdash_{\bufferTracker_2}\ P_2 \\ \bufferTracker_1 \cap \bufferTracker_2 = \emptyset}
    {\contO\cdot\contGb{1},\contGb{2}\ \vdash_{\bufferTracker_1 \cup \bufferTracker_2}\ \cPar{P_1\ }{\ P_2}}
\and
\inferrule [\trNew]
    {\contGp = \{\cSessionRole{s}{p}:\tSessionQueue_\role{p}\}_{\role{p}\in\roles} \\ s\not\in\contG \\ (\{s\}\,;\rel)$-$\fSafe(\contGp) \\ \contO\cdot\contG,\contGp\ \vdash_{\bufferTracker}\ P}
    {\contO\cdot\contG\ \vdash_{\bufferTracker\setminus\{s\}}\ \cRestriction{s\!:\!\contGp}{P}}
\end{mathpar}
\caption{Typing rules}
\label{fig:type-rules}
\end{figure*}

\begin{figure}[t]
\begin{mathpar}
\inferrule
    {\ }
    {\fend(\emptyset)}
\and
\inferrule
    {\forall i \in 1..n \ \cdot\  \fBasic(\tType_i)\sep \lor\sep x_i:\tType_i\ \vdash\ x_i:\tEnd}
    {\fend(x:\tType_1,\dotsc,x_n:\tType_n)}
\\
\inferrule
    {\fend(\contGb{1}) \\ \fend(\contGb{2})}
    {\fend(\contGb{1},\contGb{2})}
\and
\inferrule
    {\forall i \in 1..n, \role{p} \in \roles \ \cdot\  \cSessionRole{s_i}{p}:\tSessionQueue_i\ \vdash\ \cSessionRole{s_i}{p}:\tEnd}
    {\fend(\cSessionRole{s_i}{p}:\tSessionQueue_1,\dotsc,\cSessionRole{s_i}{p}:\tSessionQueue_n)}
\end{mathpar}
\caption{Predicate \fend(\contG)}
\label{fig:end}
\end{figure}
\begin{figure}[t]
\begin{mathpar}
\inferrule
    {\ }
    {\fPayloads(\emptyset)}
\and
\inferrule
    {\fPayloads(\contG)}
    {\fPayloads(\contG,\cSessionRole{s}{p}:\emptyBuffer)}
\and
\inferrule
    {\fBasic(\tType) \\ \fPayloads(\contG,\cSessionRole{s}{p}:\tQueue)}
    {\fPayloads(\contG,\cSessionRole{s}{p}:\tBuffer{q}{m}{\tType}{\tQueue})}
\and
\inferrule
    {\contG = \contGp, \cSessionRole{s^\prime}{p$^\prime$}: \tType \\ \fPayloads(\contGp,\cSessionRole{s}{p}:\tQueue)}
    {\fPayloads(\contG,\cSessionRole{s}{p}:\tBuffer{q}{m}{\tType}{\tQueue})}
\end{mathpar}
\caption{The garbage collector predicate \fPayloads(\contG)}
\label{fig:payloads}
\end{figure}
\begin{figure}[t!]
\begin{align*}
    \cSessionRole{s}{p}:\tBuffer{q}{m}{\tType}{\emptyBuffer}\msgInsert\contG,\cSessionRole{s}{p}:\tQueue &= \contG,\cSessionRole{s}{p}:\tBuffer{q}{m}{\tType}{\tQueue} \\
    \cSessionRole{s}{p}:\tBuffer{q}{m}{\tType}{\emptyBuffer}\msgInsert\contG \textit{ when } \cSessionRole{s}{p}:\tQueue\not\in\contG &= \contG,\cSessionRole{s}{p}:\tBuffer{q}{m}{\tType}{\emptyBuffer} 
\end{align*}
\caption{Message insertion function \contGp\msgInsert\contG}
\label{fig:msg-insert}
\end{figure}

Our type system is defined by the typing rules in \cref{fig:type-rules}. Below we explain them in detail.
Typing judgements are of the form:
$\contO\cBufferCons\contG\ \vdash\ P$ reading ``process $P$ is well typed under type contexts \contO\ and \contG"; and $\contG\ \vdash\ d: \tType$ reading ``value (or variable, or channel) $d$ is of type $\tType$ under type context $\contG$".

\begin{description}
    \item[\raiseup{\trInaction}] The inaction process \inaction\ is typed by a context that is ``\fend\ typed'', determined by the 
                       predicate \fend(\contG)---defined in \cref{fig:end}.
                       The predicate holds: \begin{enumerate*}[label=(\textit{\roman*})]
                           \item if $\contG = \emptyset$; 
                           \item if \contG\ consists of variables, then it holds if all the variables are either of a \fBasic\ type,
                                 or can be typed by \tEnd; and
                           \item if \contG\ consists of sessions with roles, then it holds if all the channels can be typed by \tEnd.
                       \end{enumerate*}
    \item[\raiseup{\trVar}] A variable or session with role $c$ has type \tType\ in a context only containing the mapping of $c$ to \tType.
    \item[\raiseup{\trVal}] A value $v$ is typed by a \fBasic\ type \tBasic\ if $v$ is contained in the set of that \fBasic\ type. 
                  \Eg, $42 : \mathbb{N}$ is typed under an empty context $\emptyset$ since $42\in\mathbb{N}$.
    \item[\raiseup{\trX}] A process variable $X$ is typed to an $n$-tuple of types $\tType_1,\dotsc,\tType_n$ under context \contO, if \contO\ maps 
                the process variable to the same $n$-tuple of types.
    \item[\raiseup{\trSelection}] The \emph{selection} process $\cSnd{c}{q$_k$}{m$_k$}{d}{P}$ is typed under a context which maps the sending channel 
                        $c$ to a {selection} session type $\tSelectQuant{\tSnd{q$_i$}{m$_i$}{\tType_i}{\tSession_i}}$, where a 
                        selection option matches the send process, \ie\ $k\in I$.
                        The context should match the payload $d$ to the type indicated in the selection ($\tType_k$), and 
                        continuation process $P$ should be typed under the continuation type $\tSession_k$.
    \item[\raiseup{\trBranch}] The \emph{branching} process $\cBranchQuant{c}{\cRcv{p$_i$}{m$_i$}{x_i}{P_i}}$ is typed under a context 
                               which maps the receiving channel $c$ to a branch type $\tBranchQuant{\tRcv{p$_i$}{m$_i$}{\tType_i}{\tSession_i}}$,
                               where all roles and message labels of each branch match.
                               Every continuation process $P_i$ must be typed under the continuation type $\tSession_i$ and
                               payload typed by $\tType_i$.
                               If the process is a \emph{timeout} branch $\cBranchQuant{c}{\cRcv{p$_i$}{m$_i$}{x_i}{P_i},\ \cTimeOut{Q}}$,
                               then it should be typed by a session type also containing a timeout continuation
                               $\tBranchQuant{\tRcv{p$_i$}{m$_i$}{\tType_i}{\tSession_i},\ \cTimeOut{\tSession^\prime}}$, and the 
                               timeout process $Q$ must be typed by $\tSession^\prime$.
    \item[\raiseup{\trCall}] A process \emph{call} $\cCall{X}{d_1,\dotsc,d_n}$ is correctly typed if \contO\ types the 
                             process variable to a $n$-tuple of types $\tType_1,\dotsc,\tType_n$ and \contG\ maps each 
                             parameter $d_i$ to the corresponding $\tType_i$ (for $i \in 1..n$).
                             Any remaining channels in \contG\ cannot be open, and hence must be \fend\ typed.
    \item[\raiseup{\trDef}] Process \emph{declaration} $\cDecl{X}{x_1:\tType_1,\dotsc,x_n:\tType_n}{P\ }$ is well typed if 
                            $P$ is self-contained, \ie\ contexts containing the types of the declaration parameters (along with any 
                            previous \contO) should type $P$.
                            Process \emph{definition} $\cDef{\ \cDecl{X}{x_1:\tType_1,\dotsc,x_n:\tType_n}{P\ }}{\ Q}$ is
                            typed under $\contO\cdot\contG$ if its declaration is well typed and $Q$ is typed under \contG\ and
                            \contO\ composed with the new process variable.
    \item[\raiseup{\trChoice}] {Non-deterministic choice} is well typed if processes are typed by $\contO\cdot\contG$ in isolation. This is in line with how \emph{case} or \emph{if-then-else} processes are typed.
    \item[\raiseup{\trLift}] We annotate the typing judgement $\contO\cdot\contG \vdash P$ with the buffer-tracker to obtain $\contO\cdot\contG\ \vdash_\bufferTracker P$,
                             denoting that the sessions in \bufferTracker\ occur in $P$.
                             The lifting rule annotates the typing judgement with an empty buffer-tracker if the buffer-less judgement 
                             ($\vdash$) types $P$ (using the rules mentioned thus far). 
    \item[\raiseup{\trEmpty}] In standard asynchronous MPST theory, the empty buffer $\cBuffer{s}{\emptyBuffer}$ is typed under the empty context 
    $\emptyset$, ensuring a one-to-one correlation between buffer types in the context and messages in a session buffer.
    However, since our calculus \emph{models message loss}, it is possible that a context contains buffer types 
    for messages that have been dropped from the process buffer.
    Thus, our theory types $\cBuffer{s}{\emptyBuffer}$ under a \emph{garbage collected} \contG.
    The predicate \fPayloads\ is defined in \cref{fig:payloads}, and states that valid leftover types $\fPayloads(\contG)$ are:
    \begin{enumerate*}[label=(\textit{\roman*})]
        \item empty;
        \item empty buffer types;
        \item message buffer types with basic-type payloads; or
        \item message buffer types with channel payloads that are typed under \contG.
    \end{enumerate*}
    \item[\raiseup{\trBuffer{1}} \raiseup{\trBuffer{2}}] An entry in a session buffer $\cBuffer{s}{\cBufferEntry{p}{q}{m}{w}\cBufferCons\sigma}$ is typed 
    under a context containing a mapping from \cSessionRole{s}{p} to a buffer type \tBuffer{q}{m}{\tType}{\tQueue}, matching 
    the recipient and message label.
    The message payload $w$ must be of type \tType, indicated by the buffer type, and buffer continuation $s:\sigma$
    should be typed under the buffer continuation type \tQueue\ in the case that it is not empty (\trBuffer{2}). 
    \item[\raiseup{\trBuffer{w}}] Weakening allows a session buffer to be typed under a larger context if the addition can be garbage collected 
                                  and inserted into the original context using the message insertion function (\cref{fig:msg-insert}).
                                  This is partial function that either appends a message to an existing buffer type, or inserts it as the head of a new buffer type.
                                  Put differently, weakening allows a buffer to be typed under a larger context containing irrelevant types that can be garbage collected.
    \item[\raiseup{\trPar}] If a process $P_1$ is typed by \contGb{1}, and process $P_2$ is typed by \contGb{2}, then the composition
                            \contGb{1},\contGb{2} types the \emph{parallel} composition \cPar{$P_1$}{$P_2$}.
                            It is also required that parallel processes \emph{cannot} each contain a buffer for the same session $s$. 
                            This guarantees the uniqueness of one session buffer per restricted session.
    \item[\raiseup{\trNew}] {Session restriction} $\cRestriction{s\!:\!\contGp}{P}$ requires session $s$ to be instantiated with a 
                            \contGp\ mapping each session with role to its session-buffer type. 
                            \fSafe(\contGp) must hold to ensure subject reduction, as discussed in \cref{subsec:types-contexts}.
                            Session $s$ should not be present in a previous context \contG, and process $P$ should be typed 
                            under the composition of the previous and newly instantiated context with the updated buffer-tracker $\contO\cdot\contG,\contGp\ \vdash_\bufferTracker P$ (since the buffer for $s$ is contained in $P$).
\end{description}

\begin{example}[Ping Pong: Type Context]
\label{ping:types}
Recalling the ping pong example, the whole system can then be described by a parallel composition of the three processes representing each role \role p, \role q, \role r together with an empty buffer, which is closed under a type context $\contG$ with the following typing assumptions.
\footnotesize{\begin{align*}
        \begin{array}{c}
            \contG = \{\cSessionRole{s}{p}:\tSession_\role{p},\ \cSessionRole{s}{q}:\tSession_\role{q}, \cSessionRole{s}{r}:\tSession_\role{r}\}
            \\[1em]
            P_{\textit{ping}} = \cRestriction{s:\contG}{\ \cPar{P_\role{p}\ }{\ \cPar{P_\role{q}\ }{\ \cPar{P_\role{r}\ }{\ \cBuffer{s}{\emptyBuffer}}}}}
         \end{array}
    \end{align*}}
\end{example}

\section{Type Properties}
\label{sec:types-props}
\label{subsec:types-props}

The main results of our MPST system for \magpi\ processes are \emph{subject reduction}
(\cref{thm:sr}) and \emph{session fidelity} (\cref{thm:sf}).
%
%
It is key to note that our results are parametric on the reliability function \rel.
Thus, the theorems we present hold for \emph{any} configuration of reliability, \ie\ 
from \emph{no} reliable communication all the way to \emph{completely} reliable networks.

In order to synchronise reliability assumptions between types and processes, we define the \emph{reliable process reduction} $\longrightarrow_\rel$, such that $\longrightarrow_\rel\ \subseteq\ \longrightarrow$.
\begin{definition}[Reliable process reduction]\label{def:rel-red}
    The reliable process reduction $\longrightarrow_\rel$ is inductively defined by the 
    same reduction rules for $\longrightarrow$  (in \cref{fig:calc-semantics}), with the following changes
    \footnote{For a fully unreliable network, \ie\ $\forall\role{p}\in\roles\cdot\rel(\role{p})=\emptyset$, 
    $\longrightarrow_\rel$ is equivalent to $\longrightarrow$.}:
    \textnormal{\footnotesize \begin{align*}
        \begin{array}{lcr}
            \rTimeOut \Sep
            & \cPar{\cBranchQuant{\cSessionRole{s}{q}}{\cRcv{p$_i$}{m$_i$}{\cVar{x_i}}{\cVar{P_i}}, \cTimeOut{\cVar{Q}}}}{\cBuffer{s}{\sigma}} 
              \longrightarrow_\rel
              \cPar{\cVar{Q}}{\cBuffer{s}{\sigma}}\Sep
            & \text{if } \exists k \in I : \role{p$_k$} \not\in \rel(\role{q}) \\\\
            \rDrop
            & \cBuffer{s}{\cBufferEntry{p}{q}{m}{w}\cBufferCons\sigma} \longrightarrow_\rel \cBuffer{s}{\sigma} 
            & \text{for } \role{q} \not\in \rel(\role{p})
        \end{array}
    \end{align*}}
\end{definition}
Intuitively, the reliable process reduction disregards network faults for reliable communication.
Concretely, a timeout reduction \rTimeOut\ is only possible if \emph{at least one role} in the branch is unreliable;
and message loss \rDrop\ can only occur for messages that are \emph{not} reliable from the viewpoint of the sender.
This ensures that no messages are ignored or lost for reliable communication.
Proofs of our technical results, along with 
any auxiliary results, are given in \cref{app:SR,app:SF,app:proc-props}.

\subsection{Subject Reduction}

Using $\longrightarrow_\rel$, we now present our result of \emph{subject reduction}.
Intuitively, subject reduction states that, if a process $P$ is typed under a safe context, and $P$ reliably reduces to some process $P'$, 
then the context also reduces (in 0 or 1 steps) to a safe context, which types the new process $P'$.

\begin{theorem}[Subject Reduction]\label{thm:sr}
    \textnormal{\begin{multline*}
        \contO\cdot\contG\ \vdash_\bufferTracker\ P\ \ \textit{and }\ 
        \safe{\bufferTracker}{\rel}{\contG}\ \ \textit{and }\ 
        P \rightarrow_\rel P^\prime\ \implies\\
        \exists \contGp: \contG \rightarrow_{(\bufferTracker;\rel)}^{\{0,1\}} \contGp\ \ \textit{and }\ 
        \safe{\bufferTracker}{\rel}{\contGp}\ \ \textit{and }\ 
        \contO\cdot\contGp \vdash_\bufferTracker P^\prime
    \end{multline*}}
\end{theorem}

A key novel result of our type system is that no unexpected network failures can occur at runtime, \ie\ a 
process always has a failure-handling subprotocol defined for unreliable communication.
This follows from the definition of our safety property $\fSafe$ (\cref{def:safety-prop}) and holds 
through subject reduction.
We state the result in \cref{cor:failure-handling}.
More precisely, this corollary states that timeout branches are guaranteed to be defined for unreliable 
communication.
The inverse is stated in \cref{cor:rel-adherence}, \ie\ timeouts are not defined for branches containing 
\emph{only} reliable sources.
\begin{corollary}[Failure handling safety]\label{cor:failure-handling}
    Given a reliability function \rel\ \textnormal{:}\\
    \textnormal{$\role{p}\not\in \rel(\role{q})$} and $\contO\cdot\contG\ \vdash_\bufferTracker\ P$ with 
    $(\bufferTracker;\rel)$-$\fSafe(\contG)$ and $P\longrightarrow_\rel^* P^\prime \equiv \ctx[Q]$
    implies \textnormal{$Q \neq \cBranchQuant{\cSessionRole{s}{q}}{\dots,\cRcv{p}{m}{x}{Q^\prime}}$}.
    \Ie\ $Q$ cannot be a branch at \textnormal{\role{q}} receiving from \textnormal{\role{p}} and not define a timeout.
\end{corollary}

\begin{corollary}[Reliability adherence]\label{cor:rel-adherence}
    Given a reliability function \rel\ \textnormal{:}\\
    \textnormal{$\rel(\role{q}) = \relSet_\role{q}$} and $\contO\cdot\contG\ \vdash_\bufferTracker\ P$ with 
    $(\bufferTracker;\rel)$-$\fSafe(\contG)$ and $P\longrightarrow_\rel^* P^\prime \equiv \ctx[Q]$
    implies \textnormal{$Q \neq \cBranchQuant{\cSessionRole{s}{q}}{\cRcv{p$_i$}{m$_i$}{x_i}{Q_i},\ \cTimeOut{Q^\prime}}$}
    \st\ \textnormal{$\forall i \in I : \role{p$_i$} \in \relSet_\role{q}$}.
    \Ie\ $Q$ cannot be a branch at \textnormal{\role{q}} only receiving from reliable roles \textnormal{\role{p$_i$}} and define a timeout.
\end{corollary}

\subsection{Session Fidelity}
\label{subset: SF}

\emph{Session fidelity} states the opposite implication of subject 
reduction,  \ie\ if \contG\ types a process $P$, and \contG\ can reduce, then 
$P$ can match at least one of the context reductions.

Consequently, relevant properties of process $P$ can be deduced from the behaviour of its type context \contG\ (as we will see in 
\cref{thm:proc-props}).
However, as shown by Scalas and Yoshida \cite[sec. 5.2]{DBLP:journals/pacmpl/ScalasY19}, the result 
does \emph{not} hold for all well-typed processes.
Concretely, session fidelity is violated by: 
\begin{enumerate*}[label=(\textit{\roman*})]
    \item processes that recurse infinitely without being productive 
    (\eg\ $\cDef{\;\cDecl{X}{x}{\cCall{X}{x}}\;}{\;\cCall{X}{\cSessionRole{s}{p}}}$); and
    \item processes that deadlock by interleaving communication across multiparty sessions.
\end{enumerate*}
Hence, we assume the necessary conditions on processes to restrict the aforementioned violations, by adapting \cite[def. 5.3]{DBLP:journals/pacmpl/ScalasY19}.

\begin{definition}[Conditions for session fidelity]\label{def:sf-conds}
    Assuming $\emptyset\cdot\contG\ \vdash_{\{s\}} P$. We say that P:
    \begin{enumerate}
        \item \textbf{has guarded definitions} iff each process definition in P of the form 
        \textnormal{\[\cDef{\ \cDecl{X}{x_1:\tType,\dots,x_n:\tType}{Q}\ }{\ P^\prime}\]}
        $\forall j\in 1..n :$ if $\;\tType_j$ is a session type, then a process call 
        $\;\cCall{Y}{\dots,x_j,\dots}$ can only occur in $Q$ as a subterm of
        \textnormal{\[\ \cBranchQuant{x_j}{\cRcv{p$_i$}{m$_i$}{y_i}{P_i[,\ \cTimeOut{P_t}]}}\  \text{ or } \ \cSnd{x_j}{p}{m}{y}{P^{\prime\prime}},\]}
        \ie\ after $x_j$ is used for input or output.\\
        \label{def:has-guarded}

        \item \textbf{only plays role \textnormal{\role{p}} in} $s$\textbf{, by} \contG\ iff:
            \begin{enumerate*}[label=\textbf{(\roman*)}]
                \item $P$ has guarded definitions (from \ref{def:has-guarded});
                \item \textnormal{$\fv(P) = \emptyset$};
                \item \textnormal{$\contG = \contGb{0},\cSessionRole{s}{p}:\tSessionQueue$} with \textnormal{$\tSessionQueue \neq \tEnd$} and \textnormal{$\fend(\contGb{0})$}; and
                \item for all $\cRestriction{s^\prime\!:\!\contGp}{P^\prime}$ subterm of $P$, $\fend(\contGp)$.
            \end{enumerate*}
        \label{def:plays-by}
    \end{enumerate}
    We say \quot{$P$ only plays role \textnormal{\role{p}} in $s$} iff 
    $\ \exists \contG: \emptyset\cdot\contG\ \vdash_{\{s\}} P$ and condition \ref{def:plays-by} holds.
\end{definition}

\Cref{def:sf-conds} formalises guarded recursion in condition \ref{def:has-guarded}, and the notion 
of only playing a single role for a given session in condition \ref{def:plays-by}.
Together, these conditions ensure that session fidelity, stated in \cref{thm:sf}, holds for 
all well-typed processes.

\begin{theorem}[Session Fidelity]\label{thm:sf}
    Assuming
        $\emptyset\cdot\contG\ \vdash_{\bufferTracker} P$ with $(\bufferTracker;\rel)$-$\fSafe(\contG)$,
        \textnormal{$P \equiv \cPar{(\Pi_{\role{p}\in I}\, P_\role{p})}{\cBuffer{s}{\sigma}}$} and \textnormal{$\contG = \bigcup_{\role{p}\in I}\, \contGb{\role{p}}$},
        and for each \textnormal{$P_\role{p}$}:
        \begin{enumerate*}[label=(\roman*)]
            \item \textnormal{$\emptyset\cdot\contGb{\role{p}}\ \vdash_{\bufferTracker} P_\role{p}$}, and
            \item \textnormal{$P_\role{p}$} being \inaction\ (up-to-$\equiv$) \emph{or} only plays role \textnormal{\role{p}} in $s$, by \textnormal{\contGb{\role{p}}}.
        \end{enumerate*}
    Then,\\[0.5em]
    $\contG\longrightarrow_{(\bufferTracker;\rel)}$ \ implies \ $\exists \contGp,P^\prime$: 
    \begin{enumerate*}[label=(\roman*)]
        \item $\contG\longrightarrow_{(\bufferTracker;\rel)}\contGp$,
        \item $P\longrightarrow^+_\rel P^\prime$,
        \item $\emptyset\cdot\contGp\ \vdash_{\bufferTracker} P^\prime$ with $(\bufferTracker;\rel)$-$\fSafe(\contGp)$,
        \item \textnormal{$P^\prime = \cPar{(\Pi_{\role{p}\in I}\, P_\role{p}^\prime)}{\cBuffer{s}{\sigma^\prime}}$} and \textnormal{$\contGp = \bigcup_{\role{p}\in I}\, \contGbp{\role{p}}$}, and 
        \item for each \textnormal{$P_\role{p}^\prime$}: \textnormal{$\emptyset\cdot\contGbp{\role{p}}\ \vdash_{\bufferTracker} P_\role{p}^\prime$}, and \textnormal{$P_\role{p}^\prime$} is \inaction\ (up-to-$\equiv$) \emph{or} only plays role \textnormal{\role{p}} in $s$, by \textnormal{\contGbp{\role{p}}}.
    \end{enumerate*}
\end{theorem}

\subsection{Process Properties}
%
%
Our result of session fidelity (\cref{subset: SF}) allows us to infer runtime properties about programs in \magpi\ from their 
types.
We proceed by defining desirable runtime properties on processes (\cref{def:proc-props}); expressing the 
equivalence of these properties at type-level (\cref{def:type-props}); and presenting our 
result of \emph{process properties verification} (\cref{thm:proc-props}), linking process properties to their 
type-level equivalences.

From \cref{def:proc-props} below, a process is:
\ref{prop:rel-com-safe}~\emph{$\rel_F$-communication-safe} (new \wrt \cite{DBLP:journals/pacmpl/ScalasY19}) if it reaches the end of communication over reliable reductions and has \emph{no leftover messages} in the buffer;
\ref{prop:proc-df}~\emph{deadlock-free} if it either reduces or it is inaction;
\ref{prop:proc-term}~\emph{terminating} if it is deadlock free and can reach inaction in a finite number of steps;
\ref{prop:proc-nterm}~\emph{never-terminating} if it can always infinitely reduce; and
\ref{prop:proc-live}~\emph{live} if, for every reliable branch it can reduce to, it can 
eventually reduce to some branch continuation.
We need not consider branches with timeouts since these are trivially live, given that a process can always 
reduce over the timeout.
\begin{definition}[Process properties\textnormal{, adapted from \cite{DBLP:journals/pacmpl/ScalasY19}}] \label{def:proc-props}
    For some reliability function \rel, and full reliability function $\rel_F$, a process $P$ is said to be:
    \begin{enumerate}[label=\textbf{(\roman*)}]
        \item \label{prop:rel-com-safe}
              \textbf{$\rel_F$-communication-safe} iff  
              \begin{gather*}
                P\longrightarrow_{\rel_F}^* P^\prime \not\longrightarrow_{\rel_F} \textit{ and }\ P^\prime = \ctx[\cBuffer{s}{\sigma}]\ \textit{ implies }\ \sigma = \emptyBuffer;
              \end{gather*}

        \item \label{prop:proc-df}   
              \textbf{deadlock-free} iff $P\longrightarrow_\rel^* P^\prime \not\longrightarrow_\rel\ $ implies $\ P^\prime \equiv \inaction$;
        
        \item \label{prop:proc-term} 
              \textbf{terminating} iff it is deadlock free, and 
                \[\exists\, i\ \text{finite}\ \st\ \forall n \geq i : P = P_0 \longrightarrow_\rel P_1 \longrightarrow_\rel \dots \longrightarrow_\rel P_n\ \text{ implies }\ P_n \not\longrightarrow_\rel\textit{;}\] 
        
        \item \label{prop:proc-nterm}
              \textbf{never-terminating} iff $P\longrightarrow_\rel^* P^\prime\ $ implies $\ P^\prime \longrightarrow_\rel$;
        
        \item \label{prop:proc-live} 
              \textbf{live} iff $P\longrightarrow_\rel^* P^\prime \equiv \ctx[Q]\ $ implies
              \textnormal{\begin{multline*}
                \textit{if }\ Q = \cBranchQuant{c}{\cRcv{q$_i$}{m$_i$}{x_i}{Q^\prime_i}}, \textit{ then} \\
                \exists\, \ctx^\prime,k\in I, w : P^\prime\longrightarrow^*_{\rel} \ctx^\prime[\,Q^\prime_k[^w/_{x_k}]\,].
              \end{multline*}}
    \end{enumerate}
\end{definition}

Note that, differently from other works~\cite{DBLP:journals/corr/abs-2207-02015,DBLP:journals/pacmpl/ScalasY19}, our definition 
of liveness only speaks about receiving processes, and not sending.
Typically, liveness also requires that a sent message---in the case of \magpi, any message 
in a session buffer---is always eventually consumed.
However, because of the failures that our calculus models, it is possible that a process is 
live and still have unconsumed messages in the buffer (\eg, as a result of timing out due to a message delay).
Additionally, for a \emph{$\rel_F$-communication-safe} process it follows that all sent messages are 
consumed in the reliable case.
Hence, the traditional definition of liveness still holds for reliable network configurations, 
and our new definition provides the largest guarantees possible given the failure assumptions.

We now present the type-level equivalences of the above process properties.
For liveness, we generalise to the largest liveness property, as done with safety 
in \cref{def:safety-prop}, allowing users to define more fine-grained notions of 
liveness, if required.

From \cref{def:type-props} below, a type context is:
\ref{prop:type-rel-com-safe}~\emph{$\rel_F$-communication-safe} if it has no populated buffer types when it can no longer reliably reduce;
\ref{prop:type-df}~\emph{deadlock-free} if the reason why it can no longer reduce is because it is \fend\ typed (and 
possibly, as a result of network failures, has some leftover types that can be garbage collected);
\ref{prop:type-term}~\emph{terminating} if it is deadlock free and can reach the end of the protocol in a finite number of steps;
\ref{prop:type-nterm}~\emph{never-terminating} if it can always infinitely reduce; and
\ref{prop:type-live}~\emph{live} if, for every reliable branch it can reduce to, there is a series of steps that 
can reduce to a continuation of that branch.

\begin{definition}[Type context properties]\label{def:type-props}
    For some reliability function \rel, a full reliability function $\rel_F$, and a set of sessions \bufferTracker, we say context \contG\ is:
    \begin{enumerate}[label=\textbf{(\roman*)}]
        \item \label{prop:type-rel-com-safe}
              $(\bufferTracker;\rel_F)$\textbf{-communication-safe} iff 
              \begin{center}
              $\contG\longrightarrow_{(\bufferTracker;\rel_F)}^* \contGp \not\longrightarrow_{(\bufferTracker;\rel_F)}$ and \textnormal{$\cSessionRole{s}{p}:\tQueue \in \contGp$} implies $\tQueue = \emptyBuffer$;
              \end{center}

        \item \label{prop:type-df}   
              $(\bufferTracker;\rel)$\textbf{-deadlock-free} iff 
              \[\contG\longrightarrow_{(\bufferTracker;\rel)}^* \contGp \not\longrightarrow_{(\bufferTracker;\rel)}\ \text{ implies }\ \contGp = \contGbp{0},\contGpp\ \st\ \fend(\contGbp{0}) \text{ and } \fPayloads(\contGpp)\textit{;}\]
        
        \item \label{prop:type-term} 
            $(\bufferTracker;\rel)$\textbf{-terminating} iff it is $(\bufferTracker;\rel)$-deadlock-free, and $\exists\, i\ \text{finite}\ \st$
                \[\forall n \geq i : \contG = \contGb{0} \longrightarrow_{(\bufferTracker;\rel)} \contGb{1} \longrightarrow_{(\bufferTracker;\rel)} \dots \longrightarrow_{(\bufferTracker;\rel)} \contGb{n}\ \text{implies}\ \contGb{n} \not\longrightarrow_{(\bufferTracker;\rel)}\textit{;}\]
        
        \item \label{prop:type-nterm}
            $(\bufferTracker;\rel)$\textbf{-never-terminating} iff $\contG\longrightarrow_{(\bufferTracker;\rel)}^* \contGp\ $ implies $\ \contGp \longrightarrow_{(\bufferTracker;\rel)}$;
        
        \item \label{prop:type-live} 
              $(\bufferTracker;\rel)$\textbf{-live} iff it obeys some liveness property $(\bufferTracker;\rel)$-\fLive\ \st
                \begin{center}
                    \begin{tabular}{l}
                        \textnormal{$(\bufferTracker;\rel)$-$\fLive(\contG,\ \contV{\cSessionRole{s}{p}}{\tSession})$} and \textnormal{$\tSession = \tBranchQuant{\tRcv{q$_i$}{m$_i$}{\tType_i}{\tSession_i}}$} \\
                            \textnormal{$\Sep\Sep\Sep\implies \ \exists \contGp, k \in I : \contG,\cSessionRole{s}{p}:\tSession \longrightarrow_{(\bufferTracker;\rel)}^* \contGp,\cSessionRole{s}{p}:\tSession_k$} \\[0.2em]
                        \textnormal{$(\bufferTracker;\rel)$-$\fLive(\contG,\ \cSessionRole{s}{p}:\tRec{t}{\tSession}) \ \implies \ (\bufferTracker;\rel)$-$\fLive(\contG,\ \cSessionRole{s}{p}:\tSession[^{\tRec{t}{\tSession}}/_\recVar{t}])$} \\[0.2em]
                        $(\bufferTracker;\rel)$-$\fLive(\contG) \text{ and } \contG\rightarrow_{(\bufferTracker;\rel)}\contGp \ \implies \ (\bufferTracker;\rel)$-$\fLive(\contG)$
                    \end{tabular}
                \end{center}
    \end{enumerate}
\end{definition}

We are now ready to use these type-level equivalent properties to infer behaviours of the processes 
they type.
We present our result in \cref{thm:proc-props} which formally states that, under the same 
assumptions given in session fidelity (\cref{thm:sf}), if a process is typed under some type context, and a 
property holds on that context, then the same property holds for the process itself.

\begin{theorem}[Process properties verification]\label{thm:proc-props}
    Assuming:
        $\emptyset\cdot\contG \vdash_\bufferTracker P$ with $(\bufferTracker;\rel)$-$\fSafe(\contG)$,
        \textnormal{$P \equiv \cPar{(\Pi_{\role{p}\in I}\, P_\role{p})}{\cBuffer{s}{\sigma}}$} and \textnormal{$\contG = \bigcup_{\role{p}\in I}\, \contGb{\role{p}}$}.
        Further, for each \textnormal{$P_\role{p}$}:
        \begin{enumerate*}[label=(\roman*)]
            \item \textnormal{$\emptyset\cdot\contGb{\role{p}}\ \vdash_\bufferTracker P_\role{p}$}, and
            \item \textnormal{$P_\role{p} \equiv \inaction\ $ or $\ P_\role{p}$} only plays role \textnormal{\role{p}} in $s$, by \textnormal{\contGb{\role{p}}}.
        \end{enumerate*}
        Then, $\forall \phi \in$ \{$\rel_F$-communication-safe, deadlock-free, terminating, never-terminating, live\}, 
        if $(\bufferTracker;\rel)$-$\phi(\contG)$, then P is $\phi$.
\end{theorem}

\subsection{Decidability}

Since \magpi\ is Turing-complete, determining the properties listed in 
\cref{def:proc-props} from \emph{processes} is \emph{undecidable}~\cite{DBLP:journals/mscs/BusiGZ09}.
A benefit of our generalised 
theory 
is that undecidable process properties can be inferred from 
\emph{decidable} type-level properties.

\begin{theorem}[Decidability]\label{thm:decidability}
    If $(\bufferTracker;\rel)$-$\phi(\contG)$ is decidable, then \quot{$\contO\cdot\contG\ \vdash_\bufferTracker P$ with $(\bufferTracker;\rel)$-$\phi(\contG)$} is decidable.
\end{theorem}

Our decidability result (\cref{thm:decidability}) states that for any decidable 
type-level property, type-checking with that property is decidable.
However, since \magpi\ is \emph{asynchronous}, we have \emph{no} results on decidability 
of $\phi$.
On the contrary, as discussed in~\cite[sec. 7]{DBLP:journals/pacmpl/ScalasY19}, 
type-level properties for \emph{asynchronous} type theories are, in some cases,
\emph{undecidable}.
This is a result of pairing buffer types with session types---which makes the 
type system Turing-powerful~\cite[thm. 2.5]{DBLP:journals/corr/abs-1211-2609}.
Scalas and Yoshida~\cite{DBLP:journals/pacmpl/ScalasY19} address this issue 
through two methods:
\begin{enumerate*}[label=(\textit{\roman*})]
    \item standard global types produce type contexts that can be captured through a 
          \emph{decidable \textbf{consistency}} property; and 
    \item restricting the size of the message buffer to make properties decidable. 
\end{enumerate*} 
The former ensures decidability by restricting communication to match the expressivity 
of global types.
For the latter, they show that any type context that remains bound within a finite-sized 
buffer is decidable (since the type has a finite state transition system representation).
In line with their results, we lift their definition of \emph{boundedness}, \ie\ a restriction 
on the size of a buffer, to \magpi's type system.

\begin{definition}[Boundedness\textnormal{, from~\cite{DBLP:journals/pacmpl/ScalasY19}}]\label{def:boundedness}
    We say \contG\ is $(\bufferTracker;\rel)$-\textbf{bound}$_k$ iff
    \textnormal{$\exists k \in \mathbb{N} : \contG \longrightarrow_{(\bufferTracker;\rel)}^* \contGp,\cSessionRole{s}{p}:\tQueue\ $} 
    implies $\ |\tQueue| < k$. \\
    \noindent
    We say \contG\ is $(\bufferTracker;\rel)$-\textbf{bounded} iff $\ \exists k\ \text{finite} : (\bufferTracker;\rel)$-\textbf{bound}$_k(\contG)$.
\end{definition}

\noindent
Using \cref{def:boundedness}, we present our result of decidable bounded properties in 
\cref{thm:decidable-props}.
\begin{theorem}[Decidable bounded properties]\label{thm:decidable-props}
    $(\bufferTracker;\rel)$-bound$_k(\contG)$ is decidable for all $\bufferTracker, \rel,$ and $k$.
    Furthermore, if $(\bufferTracker;\rel)$-bounded$(\contG)$, then 
    $\forall \phi \in $ \{$\rel_F$-communication-safe, deadlock-free, terminating, never-terminating, live\}, 
    it holds that $(\bufferTracker;\rel)$-$\phi(\contG)$ is decidable.
\end{theorem}

Thus, decidability is guaranteed for all protocols expressible through standard
\emph{asynchronous} global type theory, and all protocols that use finite
message buffers---now with the benefit of reasoning about and handling network
errors!


\begin{example}[Ping Pong: Properties]\label{ex:ping-props}
    Inspecting the types in \cref{ping-pong} and \cref{ping:types},
    we can conclude that $\contG = \{\cSessionRole{s}{p}: \tSession_\role{p}, \cSessionRole{s}{q}: \tSession_\role{q}, \cSessionRole{s}{r}: \tSession_\role{r}\}$ is bound$_4$.
   By \cref{thm:decidable-props}, $\contG$ is decidable to check for type-level properties.
    On doing so, we determine that \contG\ is:
    \begin{enumerate*}[label=(\textit{\roman*})]
        \item \emph{safe}, it satisfies the safety property (\cref{def:safety-prop})
              required for subject reduction;
        \item \emph{$\rel_F$-communication-safe}, since if we only consider reliable reductions, no buffer types remain populated;
        \item \emph{terminating}, since we can count the number of steps taken to reach the end of the protocol; and
        \item \emph{live}, as reliable communication $\tSession_\role{r}$ always reduces---\ie\ a result is always obtained.
    \end{enumerate*}
\end{example}

\section{Generalising Network Assumptions}
\label{sec:network_assumptions}

The work presented thus far covers worst-case network assumptions 
for communication.
As beneficial as this may be for low-level networks programming, and for complex 
distributed applications with minimal assumptions (\eg\ consensus protocols), 
not all applications are built on these pessimistic conditions.
\Eg\ many distributed applications operate over the Transmission Control 
Protocol (TCP), and thus assume that if consecutive messages are received from the 
same source, then they are guaranteed to arrive in the order in which they were sent.

We now showcase the few changes to \magpi\ required to alter its 
network assumptions.
It is key to note that these changes produce a \emph{subset} of \magpi, thus all relevant properties continue to be valid for its TCP-compliant version.

\subsection{From Total to Partial Reordering}

In a \emph{reliable} network configuration designed to run over TCP, message reordering 
for communication between \emph{two parties} is guaranteed to \emph{not occur}.
Therefore, we can adjust the message reordering of \magpi\ to model this environment, 
and strengthen our safety property \fSafe\ to \emph{TCP-safe} communication.
\magpi\ models message reordering through buffer congruence rules.
Therefore, strengthening congruence suffices to restrict communication to 
the TCP-safe assumptions.
\begin{definition}[TCP process-congruence]\label{def:TCP-proc-cong}
    The process congruence for the TCP-compliant subset of \magpi, $\equiv_{\tcp}$,
    is inductively defined using the same rules defining $\equiv$ (in \cref{def:buffer_cong}), but with the following change:
    \textnormal{\footnotesize{\begin{gather*}
        \cBuffer{s\!}{\!\sigma_1\cBufferCons h_1 \cBufferCons h_2 \cBufferCons\sigma_2} \;\equiv\; \cBuffer{s\!}{\!\sigma_1\cBufferCons h_2 \cBufferCons h_1 \cBufferCons\sigma_2} \\
        \textit{ replaced by }\\
        \inferrule{\role{p$_1$} \neq \role{p$_2$}\ \textit{or}\ \role{q$_1$} \neq \role{q$_2$}}
                    {\cBuffer{s\!}{\!\sigma_1\cBufferCons \cBufferEntry{\role{p$_1$}}{\role{q$_1$}}{m$_1$}{w_1} \cBufferCons \cBufferEntry{\role{p$_2$}}{\role{q$_2$}}{m$_2$}{w_2} \cBufferCons\sigma_2} \\\equiv_\tcp\; \cBuffer{s\!}{\!\sigma_1\cBufferCons \cBufferEntry{\role{p$_2$}}{\role{q$_2$}}{m$_2$}{w_2} \cBufferCons \cBufferEntry{\role{p$_1$}}{\role{q$_1$}}{m$_1$}{w_1} \cBufferCons\sigma_2}}
    \end{gather*}}}
\end{definition}

To obtain the TCP-compliant subset of \magpi, we assume reductions over fully reliable networks 
and adopt TCP process congruence from \cref{def:TCP-proc-cong}, which no longer allows reordering 
of messages for each role couple.
%
We now reflect this definition of TCP congruence at the type-level in \cref{def:TCP-type-cong}, and use this to define 
a TCP-safety property on type contexts in \cref{def:tcp-safety-prop}.

\begin{definition}[TCP type-congruence]\label{def:TCP-type-cong}
    The type congruence for the TCP-compliant subset of \magpi, $\equiv_{\tcp}$,
    is inductively defined using the same rules as $\equiv$ (\cref{fig:type-cong}), but with the following change:
    \textnormal{\footnotesize{\begin{gather*}
        \inferrule{\ }{\tQueue_1\cBufferCons\tQueue_2 \equiv \tQueue_2\cBufferCons\tQueue_1}
        \Sep\textit{ replaced by }\Sep
        \inferrule{\role{p} \neq \role{q}}{\tBuffer{p}{m$_1$}{\tType_1}{\tBuffer{q}{m$_2$}{\tType_2}{\tQueue}} \\\\\equiv_\tcp \tBuffer{q}{m$_2$}{\tType_2}{\tBuffer{p}{m$_1$}{\tType_1}{\tQueue}}}
    \end{gather*}}}
\end{definition}

\begin{definition}[TCP safety]\label{def:tcp-safety-prop}
    Predicate $\fTCP$ is a $\bufferTracker$-\textbf{TCP-safety} property on typing contexts iff: 
    \begin{center}\textnormal{
        \begin{tabular}{l}
            $\fTCP(\contG,\ \contV{\cSessionRole{s}{p}}{\tBranchQuant{\tRcv{q$_i$}{m$_i$}{\tType_i}{\tSession_i}}},\, \contV{\cSessionRole{s}{q}\!}{\!\tQueue})$ \\
            $\Sep\text{and } \tQueue \equiv_\tcp \tBuffer{p}{m}{\tType}{\tQueue^\prime}$ \\
            $\Sep\text{and } \exists\, k \in I : \role{q$_k$} = \role{q} \ \implies \ \msgLabel{m$_k$} = \msgLabel{m}\ \land\ \tType_k = \tType$ \\[0.2em]
            $\fTCP(\contG,\ \cSessionRole{s}{p}:\tRec{t}{\tSession}) \Sep\ \  \implies \ \fTCP(\contG,\ \cSessionRole{s}{p}:\tSession[^{\tRec{t}{\tSession}}/_\recVar{t}])$ \\[0.2em]
            $\fTCP(\contG) \text{ and } \contG\rightarrow_\bufferTracker \contGp \ \ \implies \ \fTCP(\contGp)$
        \end{tabular}}
    \end{center}
\end{definition}

Similar to our previous definition of safety in \cref{def:safety-prop}, TCP safety ensures that payload 
types of communicating entities match.
In addition, it also requires correct ordering of messages (up to $\equiv_\tcp$) by checking 
message labels---this is possible since messages between two parties do not get reordered, and so 
they must be received in the same order they are sent.
In order to benefit from the session theorems proved in \cref{sec:types-props}, all that is required is 
to show that $\fTCP \subseteq \fSafe$, \ie\ any context that is TCP-safe is also safe.
This is the only requirement since all theorems in \cref{sec:types-props}
\begin{enumerate*}[label=(\textit{\roman*})]
    \item are parametric on the reliability function \rel, including fully reliable networks; and
    \item are proven for $(\bufferTracker;\rel)$-$\fSafe(\contG)$.
\end{enumerate*}

\begin{proposition}[Containment of \fTCP\ in \fSafe]
    $\forall \contG \in \fTCP : \contG \in \fSafe$.
\end{proposition}

\begin{proof}
    \fTCP\ uses a fully reliable configuration of \magpi---\ie\ is void of failure-handling timeouts---and thus 
    trivially abides by \spR{1} and \spR{2}. 
    \spRec\ is reflected directly in \fTCP.
    \spRed\ is reflected for $\rel = \rel_F$, \ie\ for a fully reliable configuration.
    \spCom\ is never violated by $\contG \in \fTCP$ since $\equiv_\tcp\ \subset\ \equiv$. \qed
\end{proof}

%% file: sections/4-examples.tex
This work presents the Ping (\cref{ping-pong,ex:ping-procs,ex:ping-rel,ping:types,ex:ping-props}) and Domain Name System (\cref{subsec:DNS-cs}) examples as they are widely known, and between them cover the full range of our contributions. 
Previous related works are \emph{not} expressive enough to model either protocol with our range of failure assumptions. 
Thus Ping and DNS are suitable to illustrate how \magpi\ pushes the boundaries of MPST.
Additional examples are provided in \cref{app:examples}.

\subsection{DNS}\label{subsec:DNS-cs}

We now demonstrate the key features of \magpi\ through a case study.
We present a multiparty example of a Domain Name 
System (DNS) with a cache and inbuilt load-balancer.
This example:
\begin{enumerate*}[label=(\textit{\roman*})]
    \item reasons about failures in its unreliable connections that are specified using 
          our novel \emph{viewpoint-specific} reliability sets;
    \item defines \emph{failure-handling} protocols for these possible failures;
    \item is \emph{bounded} (\cref{def:boundedness}), and thus has decidable type-level properties; and
    \item is \emph{safe, $\rel_F$-communication-safe, deadlock-free, terminating,} and \emph{live}. 
\end{enumerate*}
Typically, DNS is implemented over TCP, however the distributed components can still suffer hardware failures. 
To cater for this, and for better demonstration of our contributions, we describe the protocol in our failure-prone setting.
%


\subsubsection{Specification}\label{subsec:DNS}

We consider a specification of a client-{DNS} interaction, where the client 
consults a {cache}, and the {DNS} delegates requests to workers.

The client, represented by role \role c, wishes to retrieve a web-address for a 
particular URL, and can do so by issuing a \msgLabel{req}uest to the DNS. 
As an optimisation, the client also stores recently retrieved addresses in a local 
and reliable \role{cache}---thus before issuing new requests to the DNS, it first consults 
this cache.
Upon receiving a request, the \role{DNS} offloads processing work to one of two 
workers, represented by roles \role{w$_1$} and \role{w$_2$}. 
After retrieving the appropriate address, the worker sends the response to the client.

The reliability configuration of this application is as such:
the client and cache have reliable connections, formally $\rel(\role{c}) = \{\role{cache}\}$ and $\rel(\role{cache}) = \{\role{c}\}$;
the DNS and workers have reliable connections, formally $\rel(\role{DNS}) = \{\role{w$_1$}, \role{w$_2$}\}$ and $\rel(\role{w$_1$}) = \rel(\role{w$_2$}) = \{\role{DNS}\}$;
all other communications are unreliable.

We now present the session types specifying the communication protocol for this distributed application.
We adopt shorthand notion for singleton selections, and omit payload types for simplicity, as with the ping example.

\begin{example}[DNS protocol]\label{ex:DNS}
    \footnotesize{\begin{align*}
        \begin{array}{l}
            \tSession_\role{c} = \tSnd{cache}{req}{}{\tBranchL{
                \begin{array}{l}
                    \tRcv{cache}{ans}{}{\tEnd}, \\
                    \tRcv{cache}{404}{}{\tSnd{DNS}{req}{}{\tBranchL{
                        \begin{array}{l}
                            \tRcv{w$_1$}{ans}{}{\tSnd{cache}{new}{}{\tEnd}}, \\
                            \tRcv{w$_2$}{ans}{}{\tSnd{cache}{new}{}{\tEnd}}, \\
                            \cTimeOut{\tSnd{cache}{ko}{}{\tEnd}}
                        \end{array}
                    }}}
                \end{array}
            }}\\[2em]
            \tSession_\role{cache} = \tBranchL{\tRcv{c}{req}{}{\tSelectL{
                \begin{array}{l}
                    \tSnd{c}{ans}{}{\tEnd}, \\
                    \tSnd{c}{404}{}{\tBranchL{
                        \begin{array}{l}
                            \tRcv{c}{new}{}{\tEnd}, \\
                            \tRcv{c}{ko}{}{\tEnd}
                        \end{array}
                    }}
                \end{array}
            }}}\\[2em]
            \tSession_\role{DNS} = \tBranchL{
                \begin{array}{l}
                    \tRcv{c}{req}{}{\tSelectL{
                        \begin{array}{l}
                            \tSnd{w$_1$}{req}{}{\tSnd{w$_2$}{ko}{}{\tEnd}} \\
                            \tSnd{w$_2$}{req}{}{\tSnd{w$_1$}{ko}{}{\tEnd}}
                        \end{array}
                    }} \\
                    \cTimeOut{\tSnd{w$_1$}{ko}{}{\tSnd{w$_2$}{ko}{}{\tEnd}}}
                \end{array}
            }\\[2em]
            \tSession_\role{w$_i$} = \tBranchL{
                \begin{array}{l}
                    \tRcv{DNS}{req}{}{\tSnd{c}{ans}{}{\tEnd}}, \\
                    \tRcv{DNS}{ko}{}{\tEnd}
                \end{array}
            }
        \end{array}
    \end{align*}}
\end{example}

Our viewpoint-specific definition of reliability is necessary 
to specify the reliable connections with the DNS and workers whilst maintaining unreliable 
connections with the client.
Additionally, the client type $\tSession_\role{c}$ (resp. the DNS type $\tSession_\role{DNS}$) is dependant 
on using undirected branching (resp. selection).
Hence this example is not expressible using previous theory~\cite{DBLP:journals/corr/abs-2207-02015,DBLP:journals/pacmpl/ScalasY19}.

%% file: sections/5-related.tex
Modelling failures has become a relevant 
and widely researched topic in recent years.
We elaborate on how our generic type system and modular language 
differs from, and in some cases may possibly subsume, related work.

Majumdar~\etal~\cite{DBLP:conf/concur/MajumdarMSZ21} introduce undirected branching 
as a means of catering for the non-deterministic \emph{partial} reordering of 
messages that is possible in networks using the Transmission Control Protocol (TCP).
As shown in \cref{sec:network_assumptions}, the modularity of our type system allows \magpi\ to
be adapted to support this network configuration, as well as other settings with lower 
levels of abstraction.

Affine type systems define types that can be used \emph{at most once}.
Affine session types~\cite{DBLP:journals/lmcs/MostrousV18,DBLP:journals/pacmpl/FowlerLMD19,DBLP:journals/mscs/CapecchiGY16}
use affine typing metatheory to allow sessions to be prematurely cancelled in the event of failure.
These works only model application-level failure (using try/catch blocks) and do 
not necessarily describe \emph{how} a failure is handled, but only allow the initial 
protocol to be abandoned if failure occurs.

Viering~\etal~\cite{DBLP:journals/pacmpl/VieringHEZ21} present a MPST theory for 
event-driven distributed systems, where processes are restarted by monitors if 
they crash.
This approach \emph{requires} a centralised reliable node, a notion that is subsumed 
by our \emph{view-point specific} definition of \emph{reliability}, 
\cref{def:reliability}.

Chen~\etal~\cite{DBLP:conf/forte/ChenVBZE16} remove the need for a centralised 
reliable node.
They equip their type system with \emph{synchronisation points} 
capable of detecting and handling failures raised by the nodes that 
experience them.
Similarly, Adameit~\etal~\cite{DBLP:conf/forte/AdameitPN17} consider an environment 
free from a centralised reliable node where unstable \emph{links} between participants
can fail.
They introduce the concept of \emph{optional blocks}, allowing \emph{default values} 
to substitute data not received due to communication failure.
Viering~\etal~\cite{DBLP:conf/esop/VieringCEHZ18}, motivated by consensus algorithms, 
delegate a group of processes as a permanently available recovery system capable of 
monitoring processes and informing them of failures.
Thus, they no longer rely on \emph{one} centralised robust node, but instead assume 
that at least some of the processes that make up the coordinator are alive at any 
given time.
The drawback in these approaches is their reliance on coordination to 
handle faults.
This may not be suitable with certain network configurations and failure-models.
Since our type system handles failure through low-level techniques, it remains agnostic 
to the types of failures, and is suitable for any non-Byzantine network 
configuration.

Recent work by Peters~\etal~\cite{DBLP:conf/forte/PetersNW22} extends global type theory 
with \emph{failure annotations}---marking the communications which are 
susceptible to failures and the kind of failure (specifically either 
process crashes or message loss).
They handle failure by defining \emph{default values and branches}.
Since the theory is an extension of global types, it suffers from the 
same problems that are addressed through generalised session types. 
Additionally, the work is not agnostic to failure-models, and thus it is uncertain if
the theory is generic enough to model failures other than the two they consider.

Most similar to \magpi\ is work by 
Barwell~\etal~\cite{DBLP:journals/corr/abs-2207-02015},
where generalised session type theory is extended to reason about crash-stop 
failures.
They reserve the \msgLabel{crash} message label, which can be used in receive
branches to detect node failure and specify failure-handling subprotocols.
In line with our research, their type system is generic, thus improving its 
expressiveness. 
However, unlike \magpi, their theory is not asynchronous, does not support undirected 
branching/selection, and assumes crash-stops to be the only possible faults---we address and capture a range of failures such as crash failures, link failures, message loss, delays and reordering and network partitioning.

Distributed variations of the 
\tpi-calculus~\cite{DBLP:conf/coordination/Amadio97,DBLP:journals/tcs/RielyH01,DBLP:books/el/01/Castellani01,DBLP:books/daglib/0018113}
introduce the concept of process \emph{locations}---representations of 
real-world physical hardware.
Processes are assigned to locations to form a process topology, and said locations 
can be crashed to model failures.
However, none of these calculi model the range of failures that 
are supported by \magpi, nor do they have type systems capable of 
ensuring communication-safe failure handling.

To conclude the paper, we presented \magpi---a Multiparty, Asynchronous and Generalised \tpi-calculus which addresses the widest set of non-Byzantine faults by using timeouts and the most general reliability definition. Our language builds on the \emph{generalised} and \emph{asynchronous} MPST, which is the most flexible for distributed programming. We prove subject reduction and session fidelity; a series of process properties, as well as fault-handling safety and reliability adherence.
As future work, we aim to investigate linear logic for Curry-Howard correspondences in order to understand the foundational and canonical meaning of faults and reliability. We aim to investigate Byzantine faults in combination with the non-Byzantine faults addressed here.
Lastly, we will explore the use of model checking to streamline the verification of process properties.

%% file: sections/appendix/subject-reduction.tex
\begin{lemma}\label{lem:safety-preservation}
    $(\bufferTracker;\rel)\textnormal{-}\fSafe(\contG,\contGp)\  \implies\ (\bufferTracker;\rel)\textnormal{-}\fSafe(\contG)$
\end{lemma}
\begin{proof}
    Proof by contradiction. \\
    Assume \contG\ is \emph{not} $(\bufferTracker;\rel)$-\textbf{safe}. \\
    Then, $\contG \longrightarrow^*_{(\bufferTracker;\rel)} \contGpp$ \st\ \contGpp\ \emph{violates} 
    $(\bufferTracker;\rel)$-$\fSafe$ (by violating any of \spR{1}, \spR{2}, or \spCom). \\
    But, by \crrCong, $\ \contG,\contGp \longrightarrow^*_{(\bufferTracker;\rel)} \contGpp,\contGp$; 
    where $\contGpp,\contGp$ is \emph{not} $(\bufferTracker;\rel)$-\textbf{safe} (since \contGpp\ violates it). \\
    $\therefore\ \contG,\contGp$ violates \spRed, thus $\contG,\contGp$ is \emph{not} $(\bufferTracker;\rel)$-\textbf{safe}; \emph{contradiction}. \\
    Hence, $(\bufferTracker;\rel)$-$\fSafe(\contG)$. \qed
\end{proof}

\begin{lemma}\label{lem:safety-composition}
    $(\bufferTracker_1;\rel)\textnormal{-}\fSafe(\contGb{1})\ \text{and}\ (\bufferTracker_2;\rel)\textnormal{-}\fSafe(\contGb{2})\ \text{and}\ \bufferTracker_1 \cap \bufferTracker_2 = \emptyset  \implies (\bufferTracker_1 \cup \bufferTracker_2;\rel)\textnormal{-}\fSafe(\contGb{1},\contGb{2})$
\end{lemma}
\begin{proof}
    Communication in \contGb{1} is contained under sessions in $\bufferTracker_1$.\\
    Likewise, communication in \contGb{2} is contained under sessions in $\bufferTracker_2$. \\
    $\therefore\ \contGb{1},\contGb{2}$ does not introduce any new communication that could 
    violate \spCom, since their sessions are disjoint ($\bufferTracker_1 \cap \bufferTracker_2 = \emptyset$). \\
    Additionally, context composition does not alter the structure of branches, 
    thus $\contGb{1},\contGb{2}$ does not violate \spR{1} or \spR{2}. \\
    \spRec\ and \spRed\ are preserved through \crrRec\ and \crrCong. \\
    Hence, $(\bufferTracker_1 \cup \bufferTracker_2;\rel)\textnormal{-}\fSafe(\contGb{1},\contGb{2})$.
    \qed
\end{proof}

\begin{lemma}[Message Garbage Collection] \label{lem:type-drop}
    If $\contO\cdot\contG\ \vdash_\bufferTracker\ \cBuffer{s}{h\cBufferCons\sigma}$,
    then $\contO\cdot\contG\ \vdash_\bufferTracker\ \cBuffer{s}{\sigma}$, where 
    $h$ is some \textnormal{\cBufferEntry{p}{q}{m}{w}}.
\end{lemma}
\begin{proof} \setcounter{equation}{0}
    Proof by induction on the structure of $\sigma$.
    \begin{enumerate}[label=\textbf{(BC)}]
        \item Base case: Let $\sigma = \emptyBuffer$ \label[case]{case:mgc-bc}
        \begin{flalign}
\label{eq:lem-mgc-1}
            & \contO\cdot\contG\ \vdash_\bufferTracker\ \cBuffer{s}{h\cBufferCons\emptyBuffer}
                & \reason{by the implication hypothesis} \\
\label{eq:lem-mgc-2}
            & \contG = \contGb{0}, \contGp, \cSessionRole{s}{p}:\tBuffer{q}{m}{\tType}{\tQueue}
                & \reason{by (\ref{eq:lem-mgc-1}) and inv. of \trBuffer{1}/\trBuffer{2}}
        \end{flalign}
        We now perform a sub-proof by induction on the derivation of (\ref{eq:lem-mgc-1}).
        \begin{enumerate}[label=\textbf{(SC\arabic*)}, itemsep=1em]
            \item Subcase \trBuffer{1} \label[subcase]{subcase:mgc-sc1}
            \begin{flalign}
\label{eq:lem-mgc-3}
                & \contO\cdot\contGp\ \vdash_\bufferTracker\ \cBuffer{s}{\emptyBuffer}\ \ \text{and}\ \ 
                  \contGb{0}\ \vdash\ w:\tType \ \ \text{and}\ \ 
                  \tQueue = \emptyBuffer
                    & \reason{by inv. of \trBuffer{1}} \\
\label{eq:lem-mgc-4}
                & \contGb{0} = w:\tType
                    & \reason{by (\ref{eq:lem-mgc-3}), \trVar} \\
\label{eq:lem-mgc-5}
                & \fPayloads(\contGp) 
                    & \reason{by (\ref{eq:lem-mgc-3}), \trEmpty} 
            \end{flalign}
            By (\ref{eq:lem-mgc-2}) and ctx comp., we know $\cSessionRole{s}{p} \not\in \contGp$.
            Hence:
            \begin{flalign}
\label{eq:lem-mgc-6}
                & \fPayloads(\contGp, \cSessionRole{s}{p}:\emptyBuffer) 
                    & \reason{by (\ref{eq:lem-mgc-5}) and \cref{fig:payloads}} \\
\label{eq:lem-mgc-7}
                & \fPayloads(\contG)
                    & \reason{by (\ref{eq:lem-mgc-2}), (\ref{eq:lem-mgc-4}), (\ref{eq:lem-mgc-6}), and \cref{fig:payloads}} \\
\label{eq:lem-mgc-8}
                & \contO\cdot\contG\ \vdash_\bufferTracker\ \cBuffer{s}{\emptyBuffer}
                    & \reason{by (\ref{eq:lem-mgc-7}) and \trEmpty}
            \end{flalign}
            $\therefore$ by (\ref{eq:lem-mgc-8}), \cref{lem:type-drop} holds for \ref{case:mgc-bc}:\ref{subcase:mgc-sc1}.

            \item Subcase \trBuffer{2} \label[subcase]{subcase:mgc-sc2}
            \begin{flalign}
\label{eq:lem-mgc-9}
                & \contO\cdot\contGp,\cSessionRole{s}{p}:\tQueue\ \vdash_\bufferTracker\ \cBuffer{s}{\emptyBuffer}\ \ \text{and}\ \ 
                    \contGb{0}\ \vdash\ w:\tType
                    & \reason{by inv. of \trBuffer{1}} \\
\label{eq:lem-mgc-10}
                & \contGb{0} = w:\tType
                    & \reason{by (\ref{eq:lem-mgc-9}) and \trVar} \\
\label{eq:lem-mgc-11}
                & \fPayloads(\contGp,\cSessionRole{s}{p}:\tQueue) 
                    & \reason{by (\ref{eq:lem-mgc-9}) and \trEmpty} 
            \end{flalign}
            \begin{flalign}
\label{eq:lem-mgc-12}
                & \fPayloads(\contG)
                    & \reason{by (\ref{eq:lem-mgc-2}), (\ref{eq:lem-mgc-10}), (\ref{eq:lem-mgc-11}), \cref{fig:payloads}} \\
\label{eq:lem-mgc-13}
                & \contO\cdot\contG\ \vdash_\bufferTracker\ \cBuffer{s}{\emptyBuffer}
                    & \reason{by (\ref{eq:lem-mgc-12}) and \trEmpty}
            \end{flalign}
            $\therefore$ by (\ref{eq:lem-mgc-13}), \cref{lem:type-drop} holds for \ref{case:mgc-bc}:\ref{subcase:mgc-sc2}.
        \end{enumerate}
        Hence, since \cref{lem:type-drop} holds for all subcases of the base case, then it holds for \ref{case:mgc-bc}.
    \end{enumerate}
    \begin{enumerate}[label=\textbf{(IC)}, itemsep=1em]
        \item Inductive case: Let $\sigma = h^\prime\cBufferCons\sigma^\prime$ \label[case]{case:mgc-ic}
        %
        %
        %
        \begin{flalign}
\label{eq:lem-mgc-14}
            & \contO\cdot\contG\ \vdash_\bufferTracker\ \cBuffer{s}{\cBufferEntry{p}{q}{m}{w}\cBufferCons h^\prime\cBufferCons\sigma^\prime} 
                & \reason{implication hypothesis}
        \end{flalign}
        We now perform a sub-proof by induction on the derivation of (\ref{eq:lem-mgc-14})
        \begin{enumerate}[label=\textbf{(SC\arabic*)}, itemsep=1em]
            \item Subcase \trBuffer{1} \label[subcase]{subcase:mgc-ic-sc1}
            \begin{subequations}\addtocounter{equation}{-1}
            \begin{flalign}
\label{eq:lem-mgc-15}
                & \contG = \contGb{0},\contGb{1},\contGb{2},\ \ \st
                    & \reason{by (\ref{eq:lem-mgc-14}) and inv. of \trBuffer{1}} \\
\label{eq:lem-mgc-15a}
                & \contO\cdot\contGb{1}\ \vdash_\bufferTracker\ \cBuffer{s}{h^\prime\cBufferCons\sigma^\prime}  & \\
\label{eq:lem-mgc-15b}
                & \contGb{2}\ \vdash\ w:\tType & \\
\label{eq:lem-mgc-15c}
                & \contGb{0} = \cSessionRole{s}{p}:\tBuffer{q}{m}{\tType}{\emptyBuffer} &
            \end{flalign}
            \end{subequations}
            \begin{flalign}
\label{eq:lem-mgc-16}
                & \contG = (\contGb{0}\msgInsert\contGb{1}),\contGb{2}
                    & \reason{by (\ref{eq:lem-mgc-15}), (\ref{eq:lem-mgc-15a}), (\ref{eq:lem-mgc-15c}), and def. of \msgInsert\ in \cref{fig:msg-insert}} \\
\label{eq:lem-mgc-17}
                & \fPayloads(\contGb{0},\contGb{2})
                    & \reason{by (\ref{eq:lem-mgc-15b}), (\ref{eq:lem-mgc-15c}), and def. of \fPayloads\ in \cref{fig:payloads}} \\
\label{eq:lem-mgc-18}
                & \contO\cdot\contG\ \vdash_\bufferTracker\ \cBuffer{s}{h^\prime\cBufferCons\sigma^\prime}
                    & \reason{by (\ref{eq:lem-mgc-16}), (\ref{eq:lem-mgc-15a}), (\ref{eq:lem-mgc-17}), and \trBuffer{w}}
            \end{flalign}
            $\therefore$ by (\ref{eq:lem-mgc-17}), \cref{lem:type-drop} holds for \ref{case:mgc-ic}:\ref{subcase:mgc-ic-sc1}.

            \item Subcase \trBuffer{2} \label[subcase]{subcase:mgc-ic-sc2}
            \begin{subequations}\addtocounter{equation}{-1}
            \begin{flalign}
\label{eq:lem-mgc-19}
                & \contG = \contGb{0},\contGb{1},\contGb{2},\ \ \st
                    & \reason{by (\ref{eq:lem-mgc-14}) and inv. of \trBuffer{2}} \\
\label{eq:lem-mgc-19a}
                & \contO\cdot\contGb{1},\cSessionRole{s}{p}:\tQueue\ \vdash_\bufferTracker\ \cBuffer{s}{h^\prime\cBufferCons\sigma^\prime} & \\
\label{eq:lem-mgc-19b}
                & \contGb{2}\ \vdash\ w:\tType & \\
\label{eq:lem-mgc-19c}
                & \contGb{0} = \cSessionRole{s}{p}:\tBuffer{q}{m}{\tType}{\tQueue} & 
            \end{flalign}
            \end{subequations}
            \begin{flalign}
\label{eq:lem-mgc-20}
                & \contG = (\cSessionRole{s}{p}:\tBuffer{q}{m}{\tType}{\emptyBuffer}\msgInsert\contGb{1},\cSessionRole{s}{p}:\tQueue),\contGb{2}
                    & \reason{by (\ref{eq:lem-mgc-19}), (\ref{eq:lem-mgc-19a}), (\ref{eq:lem-mgc-19c}), and def. of \msgInsert\ in \cref{fig:msg-insert}} \\
\label{eq:lem-mgc-21}
                & \fPayloads(\cSessionRole{s}{p}:\tBuffer{q}{m}{\tType}{\emptyBuffer},\contGb{2})
                    & \reason{by (\ref{eq:lem-mgc-19b}), and def. of \fPayloads\ in \cref{fig:payloads}} \\
\label{eq:lem-mgc-22}
                & \contO\cdot\contG\ \vdash_\bufferTracker\ \cBuffer{s}{h^\prime\cBufferCons\sigma^\prime}
                    & \reason{by (\ref{eq:lem-mgc-20}), (\ref{eq:lem-mgc-19a}), (\ref{eq:lem-mgc-21}), and \trBuffer{w}}
            \end{flalign}
            $\therefore$ by (\ref{eq:lem-mgc-22}), \cref{lem:type-drop} holds for \ref{case:mgc-ic}:\ref{subcase:mgc-ic-sc2}.
            
        \end{enumerate}
        Hence, since \cref{lem:type-drop} holds for all subcases of the inductive case, then it holds for \ref{case:mgc-ic}. \qed
    \end{enumerate}
\end{proof}

\begin{customthm}{1}[Subject Reduction]
If a process is typed under a safe context, and this process reduces, then 
the context reduces (in 0 or 1 steps) to a safe context which types 
the new process. 
Formally:
\textnormal{\begin{multline*}
    \contO\cdot\contG \vdash_\bufferTracker P\ \land\ 
    \safe{\bufferTracker}{\rel}{\contG}\ \land\
    P \rightarrow_\rel P^\prime\ \implies\\
    \exists \contGp: \contG \rightarrow_{(\bufferTracker;\rel)}^{\{0,1\}} \contGp\ \land\
    \safe{\bufferTracker}{\rel}{\contGp}\ \land\
    \contO\cdot\contGp \vdash_\bufferTracker P^\prime
\end{multline*}}
where $\rightarrow^{\{0,1\}}$ refers to a reduction of 0 or 1 steps.
\end{customthm}

\begin{proof}
We prove subject reduction by induction on the derivation of $P\rightarrow_\rel P^\prime$.
We may assume the following:
\begin{enumerate}[label=\textbf{(A\arabic*)}]
    \item $\contO\cdot\contG \vdash_\bufferTracker P$ \label{ass:1}
    \item $\safe{\bufferTracker}{\rel}{\contG}$ \label{ass:2}
    \item $P \rightarrow_\rel P^\prime$ \label{ass:3}
\end{enumerate}
With these assumptions, we prove \cref{thm:sr} for all cases of $P\rightarrow_\rel P^\prime$:
\begin{enumerate}[label=\textbf{(C\arabic*)}, itemsep=1em]
\item Case \rSelection \label[case]{case:sr-selection}
    \begin{flalign}
        &P = \cPar{\cSnd{\cSessionRole{s}{p}}{q}{m}{w}{Q}\ }{\ \cBuffer{s\!}{\!\sigma}}  
            &\reason{by \ref{ass:3} and inv. of \rSelection} \\
        &P^\prime = \cPar{Q\ }{\ \cBuffer{s\!}{\!\cBufferEntry{p}{q}{m}{w}\cBufferCons\sigma\cBufferCons\emptyBuffer}}
            &\reason{by \ref{ass:3} and inv. of \rSelection}
    \end{flalign}
    \begin{subequations}\addtocounter{equation}{-1}
    \begin{flalign}
\label{eq:sr-sel-1}
        &\contG = \contGb{\selection},\contGb{s}\ \text{and}\ \bufferTracker = \bufferTracker_1 \cup \bufferTracker_2\ \ \st
            &\reason{by \ref{ass:1} and inv. of \trPar} \\
\label{eq:sr-sel-2}
        &\contO\cdot\contGb{\selection}\!\ \vdash_{\bufferTracker_1}\ \cSnd{\cSessionRole{s}{p}}{q}{m}{w}{Q} & \\
\label{eq:sr-sel-buffer}
        &\contO\cdot\contGb{s}\ \vdash_{\bufferTracker_2}\ \cBuffer{s\!}{\!\sigma} & \\
        &\bufferTracker_1 \cap \bufferTracker_2 = \emptyset &
    \end{flalign}
    \end{subequations}
    The only rule that can type the process in (\ref{eq:sr-sel-2}) is \trSelection. 
    Hence, we can conclude that:
    \begin{flalign}
\label{eq:sr-sel-3}
        &\bufferTracker_1 = \emptyset\ \ \text{and}\ \ 
         \contO\cdot\contGb{\selection} \vdash \cSnd{\cSessionRole{s}{p}}{q}{m}{w}{Q}
            &\reason{by (\ref{eq:sr-sel-2}), \trSelection, and \trLift}
    \end{flalign}
    \begin{subequations}\addtocounter{equation}{-1}
    \begin{flalign}
        &\contGb{\selection} = \contGb{0},\contGb{1},\contGb{2}\ \ \st &\reason{by (\ref{eq:sr-sel-3}) and inv. of \trSelection} \\
\label{eq:sr-sel-4}
        &\contGb{1}\ \vdash\ \cSessionRole{s}{p}:\tSelectQuant{\tSnd{q$_i$}{m$_i$}{\tType_i}{\tSession_i}} & \\
\label{eq:sr-sel-5b}
        &\contGb{2}\ \vdash\ w:\tType & \\
\label{eq:sr-sel-5c}
        &\contO\cdot\contGb{0}, \cSessionRole{s}{p}: \tSession\ \vdash\ Q & \\
\label{eq:sr-sel-kinI}
        &\exists k \in I : \role{q$_k$} = \role{q},\ \msgLabel{m$_k$} = \msgLabel{m},\ \tType_k = \tType,\ \tSession_k = \tSession &
    \end{flalign}
    \end{subequations}
    \begin{flalign}
        &\contGb{1} = \cSessionRole{s}{p}:\tSelectQuant{\tSnd{q$_i$}{m$_i$}{\tType_i}{\tSession_i}} 
            & \reason{by (\ref{eq:sr-sel-4}) and \trVal}
    \end{flalign}
    \begin{flalign}
\label{eq:sr-sel-5}
        &\bufferTracker_2 = \{s\} & \reason{by induction on the derivation of (\ref{eq:sr-sel-buffer})} \\
\label{eq:sr-sel-s-in-buff}
        &\bufferTracker\  = \{s\},\ \ \therefore\ s \in \bufferTracker & \reason{by (\ref{eq:sr-sel-1}), (\ref{eq:sr-sel-3}), and (\ref{eq:sr-sel-5})}
    \end{flalign}
    We consider two subcases: when \cSessionRole{s}{p} is contained in \contGb{s}, and when it is not.
    \begin{enumerate}[label=\textbf{(SC\arabic*)}, itemsep=1em]
        \item Subcase: Let $\cSessionRole{s}{p} \in \contGb{s}$ \label[subcase]{subcase:sr-sel-in}
            \begin{subequations}\addtocounter{equation}{-1}
            \begin{flalign}
\label{eq:sr-sel-subin-red}
                &\contG \xrightarrow{\contSnd{s}{p}{q}{m}{\tType}}_\bufferTracker \contGp = \contGb{0},\contGb{2},\contGbp{1},\contGbp{s}\ \ \st
                    &\reason{by (\ref{eq:sr-sel-kinI}), (\ref{eq:sr-sel-s-in-buff}), and \crrSnd{2}} \\
                &\contGbp{1} = \cSessionRole{s}{p}:\tTuple{\tQueue\cBufferCons\tBuffer{q}{m}{\tType}{\emptyBuffer}}{\tSession} & \\
                &\contGbp{s} = \contGb{s}\setminus\{\cSessionRole{s}{p}:\tQueue\} & 
            \end{flalign}
            \end{subequations}
            \begin{subequations}\addtocounter{equation}{-1}
            \begin{flalign}
                &\contGbp{1} = \contGbp{1M},\contGbp{1s}\ \ \st
                    &\reason{by ctx comp. and cong. (\cref{def:ctx-syntax})} \\
                &\contGbp{1M} = \cSessionRole{s}{p}:\tBuffer{q}{m}{\tType}{\tQueue\cBufferCons\emptyBuffer} & \\
\label{eq:sr-sel-10b}
                &\contGbp{1s}\ = \cSessionRole{s}{p}:\tSession & 
            \end{flalign}
            \end{subequations}
            \begin{flalign}
\label{eq:sr-sel-11}
                &\contO\cdot\contGb{0},\contGbp{1s}\ \vdash\ Q & \reason{by (\ref{eq:sr-sel-5c}) and (\ref{eq:sr-sel-10b})} \\
\label{eq:sr-sel-12}
                &\contO\cdot\contGb{2},\contGbp{s},\contGbp{1M}\ \vdash_\bufferTracker\ \cBuffer{s}{\cBufferEntry{p}{q}{m}{w}\cBufferCons\sigma}
                    & \reason{by (\ref{eq:sr-sel-buffer}), (\ref{eq:sr-sel-5b}), \trBuffer{2}} \\
\label{eq:sr-sel-13}
                &\contO\cdot\contGb{0},\contGbp{1s}\ \vdash_\emptyset\ Q & \reason{by (\ref{eq:sr-sel-11}) and \trLift}
            \end{flalign}
            \begin{flalign}
\label{eq:sr-sel-14}
                &\contO\cdot\contGp\ \vdash_\bufferTracker\ P^\prime & \reason{by (\ref{eq:sr-sel-12}), (\ref{eq:sr-sel-13}), since $\bufferTracker\cap\emptyset = \emptyset$, and by \trPar} \\
\label{eq:sr-sel-15}
                &\contG\rightarrow_{(\bufferTracker;\rel)} \contGp & \reason{holds trivially by (\ref{eq:sr-sel-subin-red}) since \rel\ is not used} \\
\label{eq:sr-sel-16}
                &\safe{\bufferTracker}{\rel}{\contGp} & \reason{since \safe{\bufferTracker}{\rel}{\contG} and by (\ref{eq:sr-sel-15}) and \spRed}
            \end{flalign}
            $\therefore$ by (\ref{eq:sr-sel-14}), (\ref{eq:sr-sel-15}), and (\ref{eq:sr-sel-16}), \cref{thm:sr} holds for \ref{case:sr-selection}:\ref{subcase:sr-sel-in}.

        \item Subcase: Let $\cSessionRole{s}{p} \not\in \contGb{s}$ \label[subcase]{subcase:sr-sel-notin}
            \begin{subequations}\addtocounter{equation}{-1}
            \begin{flalign}
\label{eq:sr-sel-17}
                &\contG \xrightarrow{\contSnd{s}{p}{q}{m}{\tType}}_\bufferTracker \contGp = \contGb{0},\contGbp{1},\contGb{2},\contGb{s}\ \st
                    &\reason{by (\ref{eq:sr-sel-kinI}), (\ref{eq:sr-sel-s-in-buff}), and \crrSnd{1}} \\
                &\contGbp{1} = \cSessionRole{s}{p}:\tTuple{\tBuffer{q}{m}{\tType}{\emptyBuffer}}{\tSession} &
            \end{flalign}
            \end{subequations}
            \begin{subequations}\addtocounter{equation}{-1}
            \begin{flalign}
                &\contGbp{1} = \contGbp{1M},\contGbp{1s}\ \ \st
                    &\reason{by ctx comp. and cong. (\cref{def:ctx-syntax})} \\
                &\contGbp{1M} = \cSessionRole{s}{p}:\tBuffer{q}{m}{\tType}{\emptyBuffer} & \\
\label{eq:sr-sel-18b}
                &\contGbp{1s}\ = \cSessionRole{s}{p}:\tSession & 
            \end{flalign}
            \end{subequations}
            \begin{flalign}
\label{eq:sr-sel-19}
                &\contO\cdot\contGb{0},\contGbp{1s}\ \vdash\ Q & \reason{by (\ref{eq:sr-sel-5c}) and (\ref{eq:sr-sel-18b})} \\
\label{eq:sr-sel-20}
                &\contO\cdot\contGb{2},\contGb{s},\contGbp{1M}\ \vdash_\bufferTracker\ \cBuffer{s}{\cBufferEntry{p}{q}{m}{w}\cBufferCons\sigma}
                    & \reason{by (\ref{eq:sr-sel-buffer}), (\ref{eq:sr-sel-5b}), \trBuffer{1}} \\
\label{eq:sr-sel-21}
                &\contO\cdot\contGb{0},\contGbp{1s}\ \vdash_\emptyset\ Q & \reason{by (\ref{eq:sr-sel-19}) and \trLift}
            \end{flalign}
            \begin{flalign}
\label{eq:sr-sel-22}
                &\contO\cdot\contGp\ \vdash_\bufferTracker\ P^\prime & \reason{by (\ref{eq:sr-sel-20}), (\ref{eq:sr-sel-21}), since $\bufferTracker\cap\emptyset = \emptyset$, and by \trPar} \\
\label{eq:sr-sel-23}
                &\contG\rightarrow_{(\bufferTracker;\rel)} \contGp & \reason{holds trivially by (\ref{eq:sr-sel-17}) since \rel\ is not used} \\
\label{eq:sr-sel-24}
                &\safe{\bufferTracker}{\rel}{\contGp} & \reason{since \ref{ass:2} and by (\ref{eq:sr-sel-23}) and \spRed}
            \end{flalign}
            $\therefore$ by (\ref{eq:sr-sel-22}), (\ref{eq:sr-sel-23}), and (\ref{eq:sr-sel-24}), \cref{thm:sr} holds for \ref{case:sr-selection}:\ref{subcase:sr-sel-notin}.
        \end{enumerate}
        Hence, since \cref{thm:sr} holds for all subcases, it holds for \cref{case:sr-selection}.

\item Case \rBranch \setcounter{equation}{0} \label[case]{case:sr-branch}
    \begin{flalign}
\label{eq:sr-br-1}
        & P = \cPar{\cBranchQuant{\cSessionRole{s}{q}}{\cRcv{p$_i$}{m$_i$}{x_i}{Q_i}\,[,\ \cTimeOut{Q^\prime}]}}{\cBuffer{s}{\cBufferEntry{p}{q}{m}{w}\cBufferCons\sigma}}
            &\reason{by \ref{ass:3}, inv. of \rBranch} \\
\label{eq:sr-br-2}
        & P^\prime = \cPar{Q_k [^w/_{x_k}]\ }{\ \cBuffer{s}{\sigma}}\ \text{ and }\ \exists k \in I : \role{p$_k$} = \role{p},\ \msgLabel{m$_k$} = \msgLabel{m} 
            & \reason{by \ref{ass:3}, inv. of \rBranch}
    \end{flalign}
    \begin{subequations}\addtocounter{equation}{-1}
    \begin{flalign}
\label{eq:sr-br-3}
        &\contG = \contGb{\branch},\contGb{s}\ \text{and}\ \bufferTracker = \bufferTracker_1 \cup \bufferTracker_2\ \ \st
            &\reason{by \ref{ass:1} and inv. of \trPar} \\
\label{eq:sr-br-3a}
        &\contO\cdot\contGb{\branch}\ \vdash_{\bufferTracker_1}\ \cBranchQuant{\cSessionRole{s}{q}}{\cRcv{p$_i$}{m$_i$}{x_i}{Q_i}\,[,\ \cTimeOut{Q^\prime}]} & \\
\label{eq:sr-br-3b}
        &\contO\cdot\contGb{s}\ \ \vdash_{\bufferTracker_2}\ \cBuffer{s}{\cBufferEntry{p$_k$}{q}{m$_k$}{w}\cBufferCons\sigma} & \\
\label{eq:sr-br-3c}
        &\bufferTracker_1 \cap \bufferTracker_2 = \emptyset &
    \end{flalign}
    \end{subequations}
    \begin{flalign}
\label{eq:sr-br-4}
        &\bufferTracker_1 = \emptyset\ \ \text{and}\ \ 
            \contO\cdot\contGb{\branch} \vdash \cBranchQuant{\cSessionRole{s}{q}}{\cRcv{p$_i$}{m$_i$}{x_i}{Q_i}\,[,\ \cTimeOut{Q^\prime}]}
            &\reason{by (\ref{eq:sr-br-3a}), \trBranch, and \trLift}
    \end{flalign}
    \begin{subequations}\addtocounter{equation}{-1}
    \begin{flalign}
\label{eq:sr-br-5}
        &\contGb{\branch} = \contGb{0},\contGb{1}\ \ \st &\reason{by (\ref{eq:sr-br-4}) and inv. of \trBranch} \\
\label{eq:sr-br-5a}
        &\contGb{1}\ \vdash\ \cSessionRole{s}{q}:\tBranchQuant{\tRcv{p$_i$}{m$_i$}{\tType_i}{\tSession_i\,[,\ \tSession^\prime]}} & \\
\label{eq:sr-br-5b}
        & [\,\contO\cdot\contGb{0}, \cSessionRole{s}{q}: \tSession^\prime\ \vdash\ Q^\prime\,] & \\
\label{eq:sr-br-5c}
        &\exists i \in I : \contO\cdot\contGb{0}, x_i : \tType_i, \cSessionRole{s}{q}:\tSession_i\ \vdash\ Q_i &
    \end{flalign}
    \end{subequations}
    \begin{flalign}
\label{eq:sr-br-6}
        &\contGb{1} = \cSessionRole{s}{q}:\tBranchQuant{\tRcv{p$_i$}{m$_i$}{\tType_i}{\tSession_i\,[,\ \tSession^\prime]}}
            & \reason{by (\ref{eq:sr-br-5a}) and \trVal}
    \end{flalign}
    \begin{flalign}
\label{eq:sr-br-7}
        &\bufferTracker_2 = \{s\} & \reason{by induction on the derivation of (\ref{eq:sr-br-3b})} \\
\label{eq:sr-br-8}
        &\bufferTracker\  = \{s\},\ \ \therefore\ s \in \bufferTracker & \reason{by (\ref{eq:sr-br-3}), (\ref{eq:sr-br-4}), and (\ref{eq:sr-br-7})}
    \end{flalign}
    We now perform a sub-proof by induction on the derivation of (\ref{eq:sr-br-3b}) 
    \begin{enumerate}[label=\textbf{(SC\arabic*)}, itemsep=1em] 
        \item Subcase \trBuffer{1} \label[subcase]{subcase:sr-br-sc1}
        \begin{subequations}\addtocounter{equation}{-1}
        \begin{flalign}
\label{eq:sr-br-9}
            &\contGb{s} = \contGb{s0},\contGb{s1},\cSessionRole{s}{p$_k$}:\tBuffer{q}{m$_k$}{\tType}{\emptyBuffer}\ \ \st &\reason{by (\ref{eq:sr-br-3b}) and inv. of \trBuffer{1}} \\
\label{eq:sr-br-9a}
            &\contO\cdot\contGb{s1}\ \vdash_\bufferTracker\ \cBuffer{s}{\sigma} & \\
\label{eq:sr-br-9b}
            & \contGb{s0}\ \vdash\ w:\tType &
        \end{flalign}
        \end{subequations}
        \begin{flalign}
\label{eq:sr-br-10}
            & \exists k \in I : \tType_k = \tType & \reason{since \ref{ass:2}, and by (\ref{eq:sr-br-6}), (\ref{eq:sr-br-9}), (\ref{eq:sr-br-2}), and \spCom} 
        \end{flalign}
        There are two possible context reductions for \contG, but we only need to find \emph{one} that satisfies the necessary conditions. Consider:
        \begin{multline}
\label{eq:sr-br-11}
            \contG \xrightarrow{\contCom{s}{p$_k$}{q}{m$_k$}}_\bufferTracker \contGp = \contGb{0},\contGb{s0},\contGb{s1},\cSessionRole{s}{p$_k$}:\emptyBuffer, \cSessionRole{s}{q}:\tSession \text{ and } \exists k \in I : \tSession_k = \tSession \\
            \reason{by (\ref{eq:sr-br-8}), (\ref{eq:sr-br-2}), (\ref{eq:sr-br-10}), and \crrCom} 
        \end{multline}
        \begin{flalign}
\label{eq:sr-br-12}
            & \contO\cdot\contGb{0},\contGb{s0},\cSessionRole{s}{q}:\tSession_k\ \vdash_\emptyset\ Q_k[^w/_{x_k}]
                & \reason{by (\ref{eq:sr-br-5c}), (\ref{eq:sr-br-9b}), (\ref{eq:sr-br-11}), and \trLift} \\
\label{eq:sr-br-13}
            & \contO\cdot\contGb{s1},\cSessionRole{s}{p$_k$}:\emptyBuffer\ \vdash_\bufferTracker\ \cBuffer{s}{\sigma}
                & \reason{by (\ref{eq:sr-br-9a}), (\ref{eq:sr-br-11}), \fPayloads, and since \cSessionRole{s}{p$_k$} $\not\in$ \dom(\contGb{s1})}
        \end{flalign}
        \begin{flalign}
\label{eq:sr-br-14}
            & \contO\cdot\contGp\ \vdash_\bufferTracker\ P^\prime
                & \reason{since $\bufferTracker\cap\emptyset = \emptyset$, and by (\ref{eq:sr-br-12}), (\ref{eq:sr-br-13}), and \trPar}\\
\label{eq:sr-br-15}
            &\contG\rightarrow_{(\bufferTracker;\rel)} \contGp & \reason{holds trivially by (\ref{eq:sr-br-11}) since \rel\ is not used} \\
\label{eq:sr-br-16}
            &\safe{\bufferTracker}{\rel}{\contGp} & \reason{since \safe{\bufferTracker}{\rel}{\contG} and by (\ref{eq:sr-br-15}) and \spRed}
        \end{flalign}
        $\therefore$ by (\ref{eq:sr-br-14}), (\ref{eq:sr-br-15}), and (\ref{eq:sr-br-16}), \cref{thm:sr} holds for \ref{case:sr-branch}:\ref{subcase:sr-br-sc1}.

        \item Subcase \trBuffer{2} \label[subcase]{subcase:sr-br-sc2} is similar to \cref{subcase:sr-br-sc1}.
    \end{enumerate}
    
    Hence, since \cref{thm:sr} holds for all subcases, it holds for \cref{case:sr-branch}.

\item Case \rTimeOut \setcounter{equation}{0} \label[case]{case:sr-timeout}
    \begin{flalign}
\label{eq:sr-to-1}
        & P = \cPar{\cBranchQuant{\cSessionRole{s}{q}}{\cRcv{p$_i$}{m$_i$}{x_i}{Q_i},\ \cTimeOut{Q^\prime}}}{\cBuffer{s}{\sigma}}
            &\reason{by \ref{ass:3}, inv. of \rTimeOut} \\
\label{eq:sr-to-2}
        & P^\prime = \cPar{Q^\prime\ }{\ \cBuffer{s}{\sigma}} & \reason{by \ref{ass:3}, inv. of \rTimeOut}
    \end{flalign}
    \begin{subequations}\addtocounter{equation}{-1}
    \begin{flalign}
\label{eq:sr-to-3}
        & \contG = \contGb{\branch},\contGb{s}\ \text{and}\ \bufferTracker = \bufferTracker_1\cup\bufferTracker_2\ \ \st 
            & \reason{by \ref{ass:1}, inv. of \trPar} \\
\label{eq:sr-to-3a}
        & \contO\cdot\contGb{\branch} \vdash_{\bufferTracker_1}
          \cBranchQuant{\cSessionRole{s}{q}}{\cRcv{p$_i$}{m$_i$}{x_i}{Q_i},\ \cTimeOut{Q^\prime}} & \\
\label{eq:sr-to-3b}
        & \contO\cdot\contGb{s}\ \vdash_{\bufferTracker_2}\ \cBuffer{s}{\sigma} & \\
\label{eq:sr-to-3c}
        & \bufferTracker_1 \cap \bufferTracker_2 = \emptyset & 
    \end{flalign}
    \end{subequations}
    The only rule that can type the process in (\ref{eq:sr-to-3a}) is \trBranch. 
    Hence we can conclude that:
    \begin{flalign}
\label{eq:sr-to-4}
        & \bufferTracker_1 = \emptyset\ \text{and}\ \contO\cdot\contGb{\branch}\ \vdash\ 
          \cBranchQuant{\cSessionRole{s}{q}}{\cRcv{p$_i$}{m$_i$}{x_i}{Q_i},\ \cTimeOut{Q^\prime}} 
            & \reason{by (\ref{eq:sr-sel-2}), \trBranch, and \trLift}
    \end{flalign}
    \begin{subequations}\addtocounter{equation}{-1}
    \begin{flalign}
\label{eq:sr-to-5}
        &\contGb{\branch} = \contGb{0},\contGb{1}\ \ \st & \reason{by (\ref{eq:sr-to-4}) and inv. of \trBranch} \\
\label{eq:sr-to-5a}
        &\contGb{1}\ \vdash\ 
         \cSessionRole{s}{q}:\tBranchQuant{\tRcv{p$_i$}{m$_i$}{\tType_i}{\tSession_i},\ \cTimeOut{\tSession^\prime}} & \\
\label{eq:sr-to-5b}
        &\contO\cdot\contGb{0}, \cSessionRole{s}{q}:\tSession^\prime\ \vdash\ Q^\prime & \\
\label{eq:sr-to-5c}
        &\forall i \in I : \contO\cdot\contGb{0}, x_i:\tType_i, \cSessionRole{s}{q}:\tSession_i\ \vdash\ Q_i &
    \end{flalign}
    \end{subequations}
    \begin{flalign}
\label{eq:sr-to-6}
        &\contGb{1} = \cSessionRole{s}{q}:\tBranchQuant{\tRcv{p$_i$}{m$_i$}{\tType_i}{\tSession_i},\ \cTimeOut{\tSession^\prime}}
            & \reason{by (\ref{eq:sr-to-5a}) and \trVar} \\
\label{eq:sr-to-7}
        &\bufferTracker = \bufferTracker_2 = \{s\},\  \therefore s \in \bufferTracker
            & \reason{by (\ref{eq:sr-to-3c}) and ind. on (\ref{eq:sr-to-3b})} \\
\label{eq:sr-to-8}
        &\exists k \in I : \role{p$_k$} \not\in \rel(\role{q})
            & \reason{by \ref{ass:2}, (\ref{eq:sr-to-6}) and \spR{2}}
    \end{flalign}
    \begin{subequations}\addtocounter{equation}{-1}
    \begin{flalign}
\label{eq:sr-to-9}
        &\contG = \contGb{0},\contGb{1},\contGb{s} \xrightarrow{\contTimeout{s}{p}}_\bufferTracker
         \contGp = \contGb{0},\contGbp{1},\contGb{s}\ \ \st
            & \reason{by (\ref{eq:sr-to-7}), (\ref{eq:sr-to-8}), and \crrTO} \\
\label{eq:sr-to-9a}
        &\contGbp{1} = \cSessionRole{s}{q}:\tSession^\prime
    \end{flalign}
    \end{subequations}
    \begin{flalign}
\label{eq:sr-to-10}
        &\contO\cdot\contGb{0},\contGbp{1}\ \vdash\ Q^\prime
            &\reason{by (\ref{eq:sr-to-5b}) and (\ref{eq:sr-to-9a})} \\
\label{eq:sr-to-11}
        &\contO\cdot\contGb{0},\contGbp{1}\ \vdash_\emptyset\ Q^\prime
            & \reason{by (\ref{eq:sr-to-10}) and \trLift} \\
\label{eq:sr-to-12}
        &\contO\cdot\contGp\ \vdash_\bufferTracker\ P^\prime
            &\reason{by \trPar\ using (\ref{eq:sr-to-11}), ((\ref{eq:sr-to-7}) and (\ref{eq:sr-to-3b})), 
                     $\bufferTracker\cap\emptyset = \emptyset$} \\
\label{eq:sr-to-13}
        &\contG\rightarrow_{(\bufferTracker;\rel)} \contGp 
            & \reason{by (\ref{eq:sr-to-8}) and (\ref{eq:sr-to-9})} \\
\label{eq:sr-to-14}
        &\safe{\bufferTracker}{\rel}{\contGp} 
            & \reason{since \ref{ass:2} and by (\ref{eq:sr-to-13}) and \spRed}
    \end{flalign}
    $\therefore$ by (\ref{eq:sr-to-12}), (\ref{eq:sr-to-13}), and (\ref{eq:sr-to-14}), \cref{thm:sr} holds for \cref{case:sr-timeout}.

\item Case \rChoice \setcounter{equation}{0} \label[case]{case:sr-choice}
    \begin{flalign}
\label{eq:sr-ch-1}
        & P = \cChoice{P_1}{P_2} &\reason{by \ref{ass:3} and inv. of \rChoice}
    \end{flalign}
    The only way \ref{ass:1} holds is if $\bufferTracker = \emptyset$ (since 
    otherwise there is no applicable typing rule), hence we conclude that:
    \begin{flalign}
\label{eq:sr-ch-2}
        & \bufferTracker = \emptyset & \reason{holds from \ref{ass:1}}\\
\label{eq:sr-ch-3}
        & \contO\cdot\contG\ \vdash\ P & \reason{by (\ref{eq:sr-ch-2}) and inv. of \trLift} \\
\label{eq:sr-ch-4}
        & \contO\cdot\contG\ \vdash\ P_1\ \text{and}\ 
          \contO\cdot\contG\ \vdash\ P_2
            & \reason{by (\ref{eq:sr-ch-3}) and inv. of \trChoice}
    \end{flalign}
    We now consider two subcases, one for each possible reduction of the 
    non-deterministic choice:

    \begin{enumerate}[label=\textbf{(SC\arabic*)}, itemsep=1em]
        \item Subcase: Let $P^\prime = P_1$ \label[subcase]{subcase:sr-ch-1}
        \begin{flalign}
\label{eq:sr-ch-5}
            & \contO\cdot\contG\ \vdash_\bufferTracker\ P^\prime 
                & \reason{by (\ref{eq:sr-ch-4}), (\ref{eq:sr-ch-2}), and \trLift}
        \end{flalign}
        $\therefore\ $ by (\ref{eq:sr-ch-5}), 
        $\contG \rightarrow_{(\bufferTracker;\rel)}^0 \contG$, and \ref{ass:2},
        \cref{thm:sr} holds for \ref{case:sr-choice}:\ref{subcase:sr-ch-1}.

        \item Subcase: Let $P^\prime = P_2$ \label[subcase]{subcase:sr-ch-2}
        \begin{flalign}
\label{eq:sr-ch-6}
            & \contO\cdot\contG\ \vdash_\bufferTracker\ P^\prime 
                & \reason{by (\ref{eq:sr-ch-4}), (\ref{eq:sr-ch-2}), and \trLift}
        \end{flalign}
        $\therefore\ $ by (\ref{eq:sr-ch-6}), 
        $\contG \rightarrow_{(\bufferTracker;\rel)}^0 \contG$, and \ref{ass:2},
        \cref{thm:sr} holds for \ref{case:sr-choice}:\ref{subcase:sr-ch-2}.
    \end{enumerate}
    Hence, since \cref{thm:sr} holds for all subcases, it holds for \cref{case:sr-choice}.

\item Case \rCall \setcounter{equation}{0} \label[case]{case:sr-call}
    \begin{flalign}
\label{eq:sr-call-1}
        & P\ = \cDef{\ \cDecl{X}{x_1,\dotsc,\cVar{x_n}}{\cVar{Q\ }}}{\ (\cPar{\cCall{X}{w_1,\dotsc,w_n}}{\cVar{Q^\prime}})}
            & \reason{by inv. of \rCall}\\
\label{eq:sr-call-2}
        & P^\prime = \cDef{\ \cDecl{X}{\cVar{x_1},\dotsc,\cVar{x_n}}{\cVar{Q\ }}}{\ (\cPar{\cVar{Q}[^{w_1}/_{x_1}]\,\dotsb[^{w_n}/_{x_n}]}{\cVar{Q^\prime}})}
            & \reason{by inv. of \rCall}
    \end{flalign}
    \begin{flalign}
\label{eq:sr-call-3}
        & \contO,X:\tType_1,\dotsc,\tType_n\cdot x_1:\tType_1,\dotsc,x_n:\tType_n\ \vdash\ Q
            & \reason{by \ref{ass:1}, (\ref{eq:sr-call-1}) and inv. of \trDef} \\
\label{eq:sr-call-4}
        & \contO,X:\tType_1,\dotsc,\tType_n\cdot\contG\ \vdash\ \cPar{\cCall{X}{w_1,\dotsc,w_n}\ }{\ \cVar{Q^\prime}}
            & \reason{by \ref{ass:1}, (\ref{eq:sr-call-1}) and inv. of \trDef} 
    \end{flalign}
    \begin{subequations}\addtocounter{equation}{-1}
    \begin{flalign}
\label{eq:sr-call-5}
        & \contG = \contGb{1}, \contGb{2}\ \ \st
            & \reason{by (\ref{eq:sr-call-4}) and inv. of \trPar} \\
\label{eq:sr-call-5a}
        & \contO,X:\tType_1,\dotsc,\tType_n\cdot\contGb{1}\ \vdash\ \cCall{X}{w_1,\dotsc,w_n} & \\
\label{eq:sr-call-5b}
        & \contO,X:\tType_1,\dotsc,\tType_n\cdot\contGb{2}\ \vdash\ Q^\prime &
    \end{flalign}
    \end{subequations}
    \begin{subequations}\addtocounter{equation}{-1}
    \begin{flalign}
\label{eq:sr-call-6}
        & \contGb{1} = \contGb{11},\dotsc,\contGb{1n},\contGbp{1}\ \ \st
            & \reason{by (\ref{eq:sr-call-5a}) and inv. of \trCall} \\
\label{eq:sr-call-6a}
        & \forall i \in 1..n \cdot \contGb{1i}\ \vdash\ w_i:\tType_i & \\
\label{eq:sr-call-6b}
        & \fend(\contGbp{1}) &
    \end{flalign}
    \end{subequations}
    \begin{flalign}
\label{eq:sr-call-7}
        & \contO,X:\tType_1,\dotsc,\tType_n\cdot \contGb{1}\ \vdash\ \cVar{Q}[^{w_1}/_{x_1}]\,\dotsb[^{w_n}/_{x_n}]
            & \reason{by (\ref{eq:sr-call-3}), (\ref{eq:sr-call-6a}), (\ref{eq:sr-call-6b})} \\
\label{eq:sr-call-8}
        & \contO,X:\tType_1,\dotsc,\tType_n\cdot\contG\ \ \vdash\ \cPar{\cVar{Q}[^{w_1}/_{x_1}]\,\dotsb[^{w_n}/_{x_n}]\ }{\ \cVar{Q^\prime}}
            & \reason{by (\ref{eq:sr-call-7}), (\ref{eq:sr-call-5b}), and \trPar} \\
\label{eq:sr-call-9}
        & \contO\cdot\contG\ \ \vdash\ P^\prime
            & \reason{by (\ref{eq:sr-call-3}), (\ref{eq:sr-call-8}), and \trDef} 
    \end{flalign}
    $\therefore\ $ by (\ref{eq:sr-call-9}), 
    $\contG \rightarrow_{(\bufferTracker;\rel)}^0 \contG$, and \ref{ass:2},
    \cref{thm:sr} holds for \cref{case:sr-call}.

\item Case \rCtx \setcounter{equation}{0} \label{case:sr-ctx} \\
    We now perform a further proof by induction on the structure of \ctx.
    From \cref{def:ctx-reduction}, it follows that we must prove \cref{thm:sr} for 
    the following three cases:
    \begin{enumerate}[label=\textbf{(SC\arabic*)}]
        \item $P = \cPar{Q}{P_\ctx}\ $ and $\ P^\prime = \cPar{Q^\prime}{P_\ctx}\ $ and $\ Q \longrightarrow Q^\prime$ \label{subcase:sr-ctx-sc1}
        \begin{subequations}\addtocounter{equation}{-1}
        \begin{flalign}
\label{eq:sr-ctx-sc1-1}
            & \contG = \contGb{Q},\contGb{\ctx}\ \ \st 
                & \reason{by inv. of \trPar} \\
\label{eq:sr-ctx-sc1-1a}
            & \contO\cdot\contGb{Q}\ \vdash_{\bufferTracker_Q}\ Q & \\
\label{eq:sr-ctx-sc1-1b}
            & \contO\cdot\contGb{\ctx}\ \vdash_{\bufferTracker_\ctx}\ P_\ctx & \\
\label{eq:sr-ctx-sc1-1c}
            & \bufferTracker_Q \cap \bufferTracker_\ctx = \emptyset\ \text{and}\ \bufferTracker = \bufferTracker_Q \cup \bufferTracker_\ctx &
        \end{flalign}
        \end{subequations}
        \begin{flalign}
\label{eq:sr-ctx-sc1-2}
            & (\bufferTracker;\rel)\textnormal{-}\fSafe(\contGb{Q})
                & \reason{by \ref{ass:2} and \cref{lem:safety-preservation}}
        \end{flalign}
        \begin{subequations}\addtocounter{equation}{-1}
        \begin{flalign}
\label{eq:sr-ctx-sc1-3}
            & \exists \contGbp{Q}\ \ \st & \reason{by (\ref{eq:sr-ctx-sc1-1a}), (\ref{eq:sr-ctx-sc1-2}), since $Q\longrightarrow Q\prime$, and by the \ih} \\
\label{eq:sr-ctx-sc1-3a}
            & \contO\cdot\contGbp{Q}\ \vdash_{\bufferTracker_Q}\ Q^\prime & \\
\label{eq:sr-ctx-sc1-3b}
            & \contGb{Q} \longrightarrow_{(\bufferTracker_Q;\rel)}\contGbp{Q} & \\
\label{eq:sr-ctx-sc1-3c}
            & (\bufferTracker_Q;\rel)\textnormal{-}\fSafe(\contGbp{Q}) & 
        \end{flalign}
        \end{subequations}
        \begin{flalign}
\label{eq:sr-ctx-sc1-4}
            & \contO\cdot\contGbp{Q},\contGb{\ctx}\ \vdash_\bufferTracker\ P^\prime 
                & \reason{by (\ref{eq:sr-ctx-sc1-3a}), (\ref{eq:sr-ctx-sc1-1b}), (\ref{eq:sr-ctx-sc1-1c}), and \trPar} \\
\label{eq:sr-ctx-sc1-5}
            & \contG \longrightarrow_{(\bufferTracker;\rel)} \contGbp{Q},\contGb{\ctx}
                & \reason{by (\ref{eq:sr-ctx-sc1-3b}), \crrCong, and since $\bufferTracker_Q \subseteq \bufferTracker$}
        \end{flalign}
        Since $P_\ctx$ does not communicate over sessions in $\bufferTracker_Q$ we obtain:
        \begin{flalign}
\label{eq:sr-ctx-sc1-6}
            & (\bufferTracker_\ctx;\rel)\textnormal{-}\fSafe(\contGb{\ctx})
                & \reason{since \ref{ass:2}, \cref{lem:safety-preservation}, and (\ref{eq:sr-ctx-sc1-1c})} \\
\label{eq:sr-ctx-sc1-7}
            & (\bufferTracker;\rel)\textnormal{-}\fSafe(\contGbp{Q},\contGb{\ctx})
                & \reason{by (\ref{eq:sr-ctx-sc1-3c}), (\ref{eq:sr-ctx-sc1-6}), (\ref{eq:sr-ctx-sc1-1c}), and \cref{lem:safety-composition}}
        \end{flalign}
        $\therefore$ by (\ref{eq:sr-ctx-sc1-4}), (\ref{eq:sr-ctx-sc1-5}), and (\ref{eq:sr-ctx-sc1-7}), \cref{thm:sr} holds for \ref{subcase:sr-ctx-sc1}.

        \item $P = \cRestriction{s}{Q}\ $ and $\ P^\prime = \cRestriction{s}{Q^\prime}\ $ and $\ Q \longrightarrow Q^\prime$ \label{subcase:sr-ctx-sc2}
        \begin{subequations}\addtocounter{equation}{-1}
        \begin{flalign}
\label{eq:sr-ctx-sc2-8}
            & s:\contGp\ \ \st & \reason{by inv. of \trNew} \\
\label{eq:sr-ctx-sc2-8a}
            & \contGp = \{\cSessionRole{s}{p}:\tSessionQueue_\role{p}\}_{\role{p}\in\roles} & \\
\label{eq:sr-ctx-sc2-8b}
            & s\not\in\contG & \\
\label{eq:sr-ctx-sc2-8c}
            & (\{s\}\,;\rel)\textnormal{-}\fSafe(\contGp) & \\
\label{eq:sr-ctx-sc2-8d}
            & \contO\cdot\contG,\contGp\ \vdash_{\bufferTracker \cup \{s\}}\ Q & \\
\label{eq:sr-ctx-sc2-8e}
            & s \not\in \bufferTracker &
        \end{flalign}
        \end{subequations}
        \begin{flalign}
\label{eq:sr-ctx-sc2-9}
            & (\bufferTracker \cup \{s\}\,;\rel)\textnormal{-}\fSafe(\contG,\contGp)
                & \reason{by \ref{ass:2}, (\ref{eq:sr-ctx-sc2-8c}), (\ref{eq:sr-ctx-sc2-8e}), and \cref{lem:safety-composition}}
        \end{flalign}
        Since the sessions of \contG\ and \contGp\ are disjoint, we may consider a reduction $Q \longrightarrow Q^\prime$ with 
        session $s$, effecting \contGp; or a reduction over another session, effecting \contG.\\
        For latter case, \cref{thm:sr} holds directly from the \ih. \\
        For the former case, the following considers a reduction involving session $s$:
        \begin{subequations}\addtocounter{equation}{-1}
        \begin{flalign}
\label{eq:sr-ctx-sc2-10}
            & \exists \contGpp\ \ \st
                & \reason{since $Q \longrightarrow Q^\prime$ and by (\ref{eq:sr-ctx-sc2-8d}), (\ref{eq:sr-ctx-sc2-9}), and the \ih} \\
\label{eq:sr-ctx-sc2-10a}
            & \contG,\contGp \longrightarrow_{(\{s\};\rel)} \contG,\contGpp & \\
\label{eq:sr-ctx-sc2-10b}
            & \contO\cdot\contG,\contGpp\ \vdash_{\bufferTracker\cup\{s\}} Q^\prime & \\
\label{eq:sr-ctx-sc2-10c}
            & (\bufferTracker \cup \{s\}\,;\rel)\textnormal{-}\fSafe(\contG,\contGpp)&
        \end{flalign}
        \end{subequations}
        \begin{flalign}
\label{eq:sr-ctx-sc2-11}
            & \contGp \longrightarrow_{(\{s\};\rel)} \contGpp
                & \reason{by (\ref{eq:sr-ctx-sc2-10a}) and \crrCong} \\
\label{eq:sr-ctx-sc2-12}
            & (\{s\}\,;\rel)\textnormal{-}\fSafe(\contGpp)
                & \reason{by (\ref{eq:sr-ctx-sc2-8c}), (\ref{eq:sr-ctx-sc2-11}), \spRed} 
        \end{flalign}
        Since there is no reduction for (\ref{eq:sr-ctx-sc2-11}) which can add a new session with role 
        into the typing context, we obtain:
        \begin{flalign}
\label{eq:sr-ctx-sc2-13}
            & \contGpp = \{\cSessionRole{s}{p}:\tSessionQueue_\role{p}^\prime\}_{\role{p}\in\roles}
                & \reason{by (\ref{eq:sr-ctx-sc2-8a}) and (\ref{eq:sr-ctx-sc2-11})} \\
\label{eq:sr-ctx-sc2-14}
            & \contO\cdot\contG\ \vdash_{\bufferTracker}\ \cRestriction{s\!:\!\contGpp}{Q^\prime}
                & \reason{by (\ref{eq:sr-ctx-sc2-13}), (\ref{eq:sr-ctx-sc2-8b}), (\ref{eq:sr-ctx-sc2-12}), (\ref{eq:sr-ctx-sc2-10b}), and \trNew}
        \end{flalign}
        $\therefore$ since $\contG\longrightarrow^0_{\bufferTracker;\rel}\contG$, and by \ref{ass:2} and (\ref{eq:sr-ctx-sc2-14}), 
        \cref{thm:sr} holds for \ref{subcase:sr-ctx-sc2}.

        \item $P = \cDef{\ D\ }{\ Q}\ $ and $\ P^\prime = \cDef{\ D\ }{\ Q^\prime}\ $ and $\ Q \longrightarrow Q^\prime$ \label{subcase:sr-ctx-sc3} \\
        This subcase holds directly from the \ih. 
    \end{enumerate}
    
    $\therefore$ since \cref{thm:sr} holds for all three subcases, it also holds for \cref{case:sr-ctx}.

\item Case \rDrop \setcounter{equation}{0} \label[case]{case:sr-drop}
    This case holds directly from \ref{ass:1}, \ref{ass:2},
    $\contG \rightarrow_{(\bufferTracker;\rel)}^0 \contG$, and \cref{lem:type-drop}.
\end{enumerate}

$\therefore$ \cref{thm:sr} holds since it holds for all possible cases. \qed
\end{proof}

\begin{customcor}{1}[Failure handling safety]
    Given a reliability function \rel\ \textnormal{:}
    \textnormal{$\role{p}\not\in \rel(\role{q})$} and $\contO\cdot\contG\ \vdash_\bufferTracker\ P$ with 
    $(\bufferTracker;\rel)$-$\fSafe(\contG)$ and $P\longrightarrow_\rel^* P^\prime \equiv \ctx[Q]$
    implies \textnormal{$Q \neq \cBranchQuant{\cSessionRole{s}{q}}{\dots,\cRcv{p}{m}{x}{Q^\prime}}$}.
    \Ie\ $Q$ cannot be a branch at \textnormal{\role{q}} receiving from \textnormal{\role{p}} and not define a timeout.
\end{customcor}

\begin{proof}[by contradiction]
    Assume $Q = \cBranchQuant{\cSessionRole{s}{q}}{\dots,\cRcv{p}{m}{x}{Q^\prime}}$, by \cref{thm:sr}, it holds 
    that $\exists \contGp : \fSafe(\contGp)$ and $ \contO\cdot\contGp \vdash_\bufferTracker\ Q$. 
    By inv. of \trBranch\ it holds that \contGp\ is of the form $\cSessionRole{s}{q}: \tBranchQuant{\tRcv{p$_i$}{m$_i$}{x_i}{\tSession_i}}$
    and $\exists k \in I : \role{p$_k$} = \role{p}$. 
    But $\role{p} \not\in \rel(\role q)$, thus violates \spR{1}.
    \emph{Contradiction}, $\therefore\ Q \neq \cBranchQuant{\cSessionRole{s}{q}}{\dots,\cRcv{p}{m}{x}{Q^\prime}}$. \qed
\end{proof}

\begin{customcor}{2}[Reliability Adherence]
    Given a reliability function \rel\ \textnormal{:}
    \textnormal{$\rel(\role{q}) = \relSet_\role{q}$} and $\contO\cdot\contG\ \vdash_\bufferTracker\ P$ with 
    $(\bufferTracker;\rel)$-$\fSafe(\contG)$ and $P\longrightarrow_\rel^* P^\prime \equiv \ctx[Q]$
    implies \textnormal{$Q \neq \cBranchQuant{\cSessionRole{s}{q}}{\cRcv{p$_i$}{m$_i$}{x_i}{Q_i},\ \cTimeOut{Q^\prime}}$}
    \st\ \textnormal{$\forall i \in I : \role{p$_i$} \in \relSet_\role{q}$}.
    \Ie\ $Q$ cannot be a branch at \textnormal{\role{q}} only receiving from reliable roles \textnormal{\role{p$_i$}} and define a timeout.
\end{customcor}

\begin{proof}
    Proof is similar to that of \cref{cor:failure-handling}, but results in a contradiction by violation of \spR{2}. \qed
\end{proof}

%% file: sections/appendix/session-fidelity.tex
\begin{proposition}[Subject Congruence]
    Assume $\contO\cdot\contG \vdash_\bufferTracker P$ and $P \equiv P^\prime$, then 
    $\exists \contGp$ s.t. $\contG \equiv \contGp$ and $\contO\cdot\contGp \vdash_\bufferTracker P^\prime$.
\end{proposition}
\begin{proof}
    Proof is by induction on the derivation of $P \equiv P^\prime$.
    Most cases are straightforward and hold for $\contG = \contGp$.
    Below we elaborate the two buffer congruence cases:
    \begin{enumerate}
        \item $P = \cRestriction{s : \contGpp}{\cBuffer{s}{\sigma}} \equiv \inaction = P^\prime$.
              We must have $\contO\cdot\contGpp \vdash_{\{s\}} s:\sigma$ and $\fend(\contG)$.
              This is because \contGpp\ contains all the types required to check $s:\sigma$ (guaranteed by the buffer tracker).
              Thus $\fPayloads(\contG,\contGpp)$ holds, and since $s \not\in \contG$, then $\fend(\contG)$.
              $\therefore$ we conclude with $\contGp = \contG$ by typing rule \trInaction.
        \item $P \equiv P^\prime$ by swapping message positions in the buffer. 
              For this case we mirror the message reordering operation in the type by applying corresponding buffer type message swaps (\cref{fig:type-cong}). \qed
    \end{enumerate}
\end{proof}

\begin{proposition}[Normal Form]\label{prop:nf}
    For all $P$, $P \equiv \cDef{\ \tilde{D}\ }{\ \cRestriction{\tilde{s}}{\cPar{P_1}{\cPar{\cdots}{P_n}}}}$
    where $\forall i \in 1..n, P_i$ is either a branch (with possible timeout), selection, process call or session buffer.
\end{proposition}

\begin{proof}
    Similar to \cite[Prop. L.1]{DBLP:journals/pacmpl/ScalasY19}. \qed
\end{proof}

\begin{lemma}[Session Inversion]\label{lem:session-inv}
    Assume $\emptyset\cdot\contG\ \vdash_\bufferTracker\ (\Pi_{\role{p}\in I} P_\role{p})$ with each $P_\role{p}$ being \inaction\ (up-to-$\equiv$) or only plays role \role{p} in $s$.
    Then $\contG = \{\cSessionRole{s}{p}:\tSession_\role{p}\}_{\role{p}\in I^\prime},\{\cSessionRole{s}{p}:\tQueue_\role{p}\}_{\role{p}\in I^{\prime\prime}}$ for some $I^\prime, I^{\prime\prime}$.
    Moreover, $\forall \role{p} \in I^\prime$:
    \begin{enumerate}
        \item if $\ \tSession_\role{p} = \tSelectQuant{\tSnd{q$_j$}{m$_j$}{\tType_j}{\tSession_j}}$ then $\role{p} \in I$ and for some $\ctx, \ctx^\prime$ and $k \in J$, either:
            \begin{flalign*}
                & P_\role{p}\ \equiv\ \ctx[\cSnd{\cSessionRole{s}{p}}{q$_k$}{m$_k$}{w_k}{P_\role{p}^\prime}]\ \ \text{or}
                    &   \begin{tabular}{l}
                        \text{where $w$ is some} \\ \text{value $v$ or some $\cSessionRole{s^\prime}{\role{r}}$}
                    \end{tabular} \\
                & P_\role{p}\ \equiv\ \ctx\left[
                    \begin{array}{l}
                        \cDef{\cDecl{\ X}{x_1:\tType_1, \dots, x_n : \tType_n}{\\ \Sep\Sep\ctx^\prime[\cSnd{x_l}{q$_k$}{m$_k$}{d_k}{P_\role{p}^\prime}]\ \ }}
                            {\\ \cCall{X}{w_1, \dots, w_{l-1}, \cSessionRole{s}{p}, w_{l+1}, \dots, w_n}}
                    \end{array}
                \right]
                    & \begin{tabular}{l}
                        \text{with $1 \leq l \leq n$} \\ \text{and each $w_m$ is either} \\ \text{some value $v_m$ or some $\cSessionRole{s^\prime_m}{r$_m$}$}
                    \end{tabular}
            \end{flalign*}
        \item if $\ \tSession_\role{p} = \tBranchQuantx{\tRcv{q$_j$}{m$_j$}{\tType_j}{\tSession_j}}{j}{J}$ then $\role{p} \in I$ and for some $\ctx, \ctx^\prime$, either:
            \begin{flalign*}
                & P_\role{p}\ \equiv\ \ctx[\cBranchQuantx{\cSessionRole{s}{p}}{\cRcv{q$_j$}{m$_j$}{x_j}{P_{\role{p$_j$}}^\prime}}{j}{J}]\ \ \text{or}
                    &   \begin{tabular}{l}
                        \text{where $w$ is some} \\ \text{value $v$ or some $\cSessionRole{s^\prime}{\role{r}}$}
                    \end{tabular} \\
                & P_\role{p}\ \equiv\ \ctx\left[
                    \begin{array}{l}
                        \cDef{\cDecl{\ X}{x_1:\tType_1, \dots, x_n : \tType_n}{\\ \Sep\Sep\ctx^\prime[\cBranchQuantx{x_l}{\cRcv{q$_j$}{m$_j$}{x_j}{P_{\role{p$_j$}}^\prime}}{j}{J}]\ \ }}
                            {\\ \cCall{X}{w_1, \dots, w_{l-1}, \cSessionRole{s}{p}, w_{l+1}, \dots, w_n}}
                    \end{array}
                \right]
                    & \begin{tabular}{l}
                        \text{with $1 \leq l \leq n$} \\ \text{and each $w_m$ is either} \\ \text{some value $v_m$ or some $\cSessionRole{s^\prime_m}{r$_m$}$}
                    \end{tabular}
            \end{flalign*}
        \item if $\ \tSession_\role{p} = \tBranchQuantx{\tRcv{q$_j$}{m$_j$}{\tType_j}{\tSession_j},\ \cTimeOut \tSession^\prime}{j}{J}$ then $\role{p} \in I$ and for some $\ctx, \ctx^\prime$, either:
            \begin{flalign*}
                & P_\role{p}\ \equiv\ \ctx[\cBranchQuantx{\cSessionRole{s}{p}}{\cRcv{q$_j$}{m$_j$}{x_j}{P_{\role{p$_j$}}^\prime},\ \cTimeOut P_\role{p$_t$}^\prime}{j}{J}]\ \ \text{or}
                    &   \begin{tabular}{l}
                        \text{where $w$ is some} \\ \text{value $v$ or some $\cSessionRole{s^\prime}{\role{r}}$}
                    \end{tabular} \\
                & P_\role{p}\ \equiv\ \ctx\left[
                    \begin{array}{l}
                        \cDef{\cDecl{\ X}{x_1:\tType_1, \dots, x_n : \tType_n}{\\ \Sep\Sep\ctx^\prime[\cBranchQuantx{x_l}{\cRcv{q$_j$}{m$_j$}{x_j}{P_{\role{p$_j$}}^\prime,\ \cTimeOut P_\role{p$_t$}^\prime}}{j}{J}]\ \ }}
                            {\\ \cCall{X}{w_1, \dots, w_{l-1}, \cSessionRole{s}{p}, w_{l+1}, \dots, w_n}}
                    \end{array}
                \right]
                    & \begin{tabular}{l}
                        \text{with $1 \leq l \leq n$} \\ \text{and each $w_m$ is either} \\ \text{some value $v_m$ or some $\cSessionRole{s^\prime_m}{r$_m$}$}
                    \end{tabular}
            \end{flalign*}
    \end{enumerate}

    Furthermore, 4. $\forall \role{p} \in I \setminus I^\prime : P_\role{p} \equiv \inaction$.
\end{lemma}

\begin{proof}
    Proof is similar to \cite[Prop. L.5]{DBLP:journals/pacmpl/ScalasY19} but uses our definition of normal form (\cref{prop:nf}). \qed
\end{proof}

\begin{customthm}{2}[Session Fidelity]
    Assuming
        $\emptyset\cdot\contG\ \vdash_{\bufferTracker} P$ with $(\bufferTracker;\rel)$-$\fSafe(\contG)$,
        \textnormal{$P \equiv \cPar{(\Pi_{\role{p}\in I}\, P_\role{p})}{\cBuffer{s}{\sigma}}$} and \textnormal{$\contG = \bigcup_{\role{p}\in I}\, \contGb{\role{p}}$},
        and for each \textnormal{$P_\role{p}$}:
        \begin{enumerate*}[label=(\roman*)]
            \item \textnormal{$\emptyset\cdot\contGb{\role{p}}\ \vdash_{\bufferTracker} P_\role{p}$}, and
            \item \textnormal{$P_\role{p}$} being \inaction\ (up-to-$\equiv$) \emph{or} only plays role \textnormal{\role{p}} in $s$, by \textnormal{\contGb{\role{p}}}.
        \end{enumerate*}
    Then,\\[0.5em]
    \begin{tabular}{l l c l}
        $\contG\longrightarrow_{(\bufferTracker;\rel)}$ & \ implies \  & $\exists \contGp,P^\prime$: \ & $\contG\longrightarrow_{(\bufferTracker;\rel)}\contGp$  \\
        && and & $P\longrightarrow^+_\rel P^\prime$ \\
        && and & $\emptyset\cdot\contGp\ \vdash_{\bufferTracker} P^\prime$ with $(\bufferTracker;\rel)$-$\fSafe(\contGp)$\\
        && and & \textnormal{$P^\prime = \cPar{(\Pi_{\role{p}\in I}\, P_\role{p}^\prime)}{\cBuffer{s}{\sigma^\prime}}$} and \textnormal{$\contGp = \bigcup_{\role{p}\in I}\, \contGbp{\role{p}}$} \\
        && and & for each \textnormal{$P_\role{p}^\prime$}: \\
        &&& $\sep$ \textnormal{$\emptyset\cdot\contGbp{\role{p}}\ \vdash_{\bufferTracker} P_\role{p}^\prime$}, and \\
        &&& $\sep$ \textnormal{$P_\role{p}^\prime$} is \inaction\ (up-to-$\equiv$) \emph{or} only plays role \textnormal{\role{p}} in $s$, by \textnormal{\contGbp{\role{p}}}
    \end{tabular}
\end{customthm}

\begin{proof}
    Proof is by induction on the reduction of $\contG$.
    
    Case $\contG\xrightarrow{\contTimeout{s}{p}}_{(\bufferTracker;\rel)}$.
    \begin{flalign}
\label{eq:sf-1}
    &\inferrule
        {s \in \bufferTracker \\ \exists\, k \in I : \role{q$_k$} \not\in \rel(\role{p})}
        {\contG = \contG_{I\setminus\role{p}},\contG_{\role{p}\tQueue},\cSessionRole{s}{p}:\!\tBranchQuant{\tRcv{q$_i$}{m$_i$}{\tType_i}{\tSession_i},\ \cTimeOut{\tSession^\prime}} \xrightarrow{\contTimeout{s}{p}}_{(\bufferTracker;\rel)} \cSessionRole{s}{p}:\tSession^\prime,\contG_{\role{p}\tQueue},\contG_{I\setminus\role{p}} = \contGp}
        &\reason{by \crrTO} 
    \end{flalign}
    By the hypothesis, (\ref{eq:sf-1}) and \cref{lem:session-inv}, the process $P_\role{p}$ 
    is a branch with a defined timeout (possibly occurring under some process definition).\\[1em]
    $\therefore P\longrightarrow^+_\rel P^\prime$ by either one reduction of \rTimeOut, or applying \rTimeOut\ after a finite number of \rCall. \\[1em]
    Then by \cref{thm:sr}, it holds that $\emptyset\cdot\contGp\ \vdash_{\bufferTracker} P^\prime$ with $(\bufferTracker;\rel)$-$\fSafe(\contGp)$.\\[1em]
    By analysis of reduction rules \rTimeOut\ and \rCall, we observe that $P^\prime$ maintains the required structure of $P^\prime = \cPar{(\Pi_{\role{q}\in I\setminus\role{p}}\, P_\role{q})}{\cPar{P_\role{p}^\prime}{\cBuffer{s}{\sigma^\prime}}}$ (for $\sigma = \sigma^\prime$).
    Similarly, $\contGp = \bigcup_{\role{q}\in I\setminus\role{p}}\, \contGb{\role{q}} \cup \contGbp{\role{p}} = \contG_{\role{p}\tQueue},\cSessionRole{s}{p} : \tSession^\prime$ for some possibly empty queue type $\contG_{\role{p}\tQueue}$.\\[1em]
    Lastly, we already know that for each $P_\role{q}$: $\emptyset\cdot\contGb{\role{q}}\ \vdash_{\bufferTracker} P_\role{q}$ and $P_\role{q}$ is \inaction\ (up-to-$\equiv$) \emph{or} only plays role \textnormal{\role{q}} in $s$, by \textnormal{\contGb{\role{q}}}.
    Thus, all that is left to show is that $\emptyset\cdot\contGbp{\role{p}}\ \vdash_{\bufferTracker} P_\role{p}^\prime$ and $P_\role{p}^\prime$ is \inaction\ (up-to-$\equiv$) \emph{or} only plays role \textnormal{\role{p}} in $s$, by \textnormal{\contGbp{\role{p}}}---which follow from rule \trBranch.\\[1em]

    Cases $\contG\xrightarrow{\contSnd{s}{p}{q}{m}{\tType}}_{(\bufferTracker;\rel)}$ and $\contG\xrightarrow{\contCom{s}{p}{q}{m}}_{(\bufferTracker;\rel)}$ follow similar steps to the first case. 
    We use context reduction rules to infer the structure of \contG, and \cref{lem:session-inv} to obtain the structure of $P$. 
    We proceed by applying similar reasoning to the steps demonstrated above. \qed
    
\end{proof}

%% file: sections/appendix/process-properties.tex
\begin{customthm}{3}[Process properties verification]
    Assuming:
        $\emptyset\cdot\contG\ \vdash_{\{s\}} P$ with $(\{s\};\rel)$-$\fSafe(\contG)$,
        \textnormal{$P \equiv \cPar{(\Pi_{\role{p}\in I}\, P_\role{p})}{\cBuffer{s}{\sigma}}$} and \textnormal{$\contG = \bigcup_{\role{p}\in I}\, \contGb{\role{p}}$}.
        Further, for each \textnormal{$P_\role{p}$}:
        \begin{enumerate*}[label=(\roman*)]
            \item \textnormal{$\emptyset\cdot\contGb{\role{p}}\ \vdash_{\{s\}} P_\role{p}$}, and
            \item \textnormal{$P_\role{p} \equiv \inaction\ $ or $\ P_\role{p}$} only plays role \textnormal{\role{p}} in $s$, by \textnormal{\contGb{\role{p}}}.
        \end{enumerate*}
        
        Then, $\forall \phi \in$ \{reliable-communication-free, deadlock-free, terminating, never-terminating, live\}, 
        if $(\bufferTracker;\rel)$-$\phi(\contG)$, then P is $\phi$.
\end{customthm}

\begin{proof}
    These results follow from \cref{thm:sf}.
\end{proof}

\begin{customthm}{4}[Decidability]
    If $(\bufferTracker;\rel)$-$\phi(\contG)$ is decidable, then \quot{$\contO\cdot\contG\ \vdash_\bufferTracker P$ with $(\bufferTracker;\rel)$-$\phi(\contG)$} is decidable.
\end{customthm}

\begin{proof}
    A decidable algorithm for type checking can be obtained by inversion of the typing rules in \cref{fig:type-rules}.
    It is key to note that:
    \begin{enumerate}[label=(\textit{\arabic*})]
        \item for each shape of $P$, at most one rule can resolve $\contO\cdot\contG\ \vdash_\bufferTracker\ P$, with the 
              exception of $\contO\cdot\contG\ \vdash_\bufferTracker\ \cBuffer{s}{\sigma}$, where \trBuffer{w} should be 
              ignored since leftover buffer types can be collectively garbage collected using \trEmpty;
        \item each rule deterministically populates \contO\ and \contG;
        \item every rule produces smaller subterms of P, ensuring that type checking of recursion terminates.
    \end{enumerate}
    Lastly, decidability of rule \trNew\ is dependant on decidability of \fSafe, which is decidable by the hypothesis. \qed
\end{proof}

\begin{customthm}{5}[Decidable bounded properties]
    $(\bufferTracker;\rel)$-bound$_k(\contG)$ is decidable for all $\bufferTracker, \rel, k$.
    Furthermore, if $(\bufferTracker;\rel)$-bounded$(\contG)$, then 
    $\forall \phi \in $ \{safe, reliable-com.-safe, deadlock-free, terminating, never-terminating, live\}, 
    it holds that $(\bufferTracker;\rel)$-$\phi(\contG)$ is decidable.
\end{customthm}

\begin{proof}
    $(\bufferTracker;\rel)$-bound$_k(\contG)$ is decidable since it is determined in a finite number of 
    reductions---\ie\ even if \contG\ can infinitely reduce, the violation of bound $k$ occurs before 
    infinitely reducing.
    Furthermore, if $(\bufferTracker;\rel)$-bound$_k(\contG)$, then \contG\ can be represented as a finite 
    state transition system, therefore $\forall \phi \in $ \{safe, reliable-com.-safe, deadlock-free, terminating, never-terminating, live\}, 
    it holds that $(\bufferTracker;\rel)$-$\phi(\contG)$ is decidable. \qed
\end{proof}

%% file: sections/appendix/leader-election.tex
\subsection{Leader Election}

We now present an example type for a leader election protocol (based off of the Raft consensus algorithm~\cite{DBLP:conf/usenix/OngaroO14}) with the following attributes:
\begin{enumerate}
    \item The election cluster consists of 3 static participants;
    \item Participants are aware of the cluster size;
    \item The protocol is fault-tolerant for all non-Byzantine faults;
    \item Participants are fail-stop, \ie\ a fault to a node implies it immediately crashes;
    \item The protocol assumes a maximum threshold of $f$ faults for a cluster size of $2f + 1$, \ie\ a maximum of 1 crashed node at a time.
\end{enumerate}

\subsubsection{Specification.}
At any point in time, a participant is either a \emph{follower}, \emph{candidate}, or \emph{leader}.
All participants begin as a \emph{follower}.
Whilst in this \emph{follower} state, participants cannot initiate communication, but rather only reply when spoken to.
\emph{Followers} may grant or deny votes when requested by \emph{candidates}.
If a \emph{follower} waits too long to receive any message, it will non-deterministically timeout and become a \emph{candidate}.
\emph{Candidates} send requests for votes to the whole cluster, and require a simple majority of votes to be elected.
Given the cluster size of 3 participants, and since \emph{candidates} always vote for themselves, only 1 external vote is required to win an election.
When a \emph{candidate} wins an election, it promotes to \emph{leader}.
\emph{Leaders} send out periodic heartbeat messages to stop \emph{followers} from timing out and starting new unnecessary elections.
Interested readers are directed to~\cite{DBLP:conf/usenix/OngaroO14} for more details.

\begin{example}[Leader Election]
    The algorithm under discussion is (quasi-)symmetric, meaning the type can be described for one 
    (generic) role and replicated for each participant.
    We use roles \role p, \role q, and \role r, to distinguish between the three distinct participants in the type.
    Note that unlike the rest of the work presented, \role p, \role q, and \role r are not representing 
    specific roles, but are placeholders to be instantiated.

    \footnotesize{\begin{align*}
        \begin{array}{l}
            \tSession_\role{p} = \tRec{F}{
                \tBranchL{
                    \begin{array}{l}
                        \tRcv{q}{hb}{}{\tSelect{\tSnd{q}{ok}{}{\recVar{F}},\ \tSnd{q}{ko}{}{\recVar{F}}}}, \\
                        \tRcv{r}{hb}{}{\tSelect{\tSnd{r}{ok}{}{\recVar{F}},\ \tSnd{r}{ko}{}{\recVar{F}}}}, \\[1em]
                        \tRcv{q}{rv}{}{\tSelect{\tSnd{q}{yes}{}{\recVar{F}},\ \tSnd{q}{no}{}{\recVar{F}}}}, \\
                        \tRcv{r}{rv}{}{\tSelect{\tSnd{r}{yes}{}{\recVar{F}},\ \tSnd{r}{no}{}{\recVar{F}}}}, \\[1em]
                        \tRcv{q}{yes}{}{F},\sep
                        \tRcv{r}{yes}{}{F},\sep
                        \tRcv{q}{no}{}{F},\sep
                        \tRcv{r}{no}{}{F},\\
                        \tRcv{q}{ok}{}{F},\sep
                        \tRcv{r}{ok}{}{F},\sep
                        \tRcv{q}{ko}{}{F},\sep
                        \tRcv{r}{ko}{}{F},\\[1em]
                        \cTimeOut{\ 
                            \tRec{C}{\ 
                                \tSnd{q}{rv}{}{\tSnd{r}{rv}{}{\ 
                                    \tRec{C$^\prime$}{\ 
                                        \tBranchL{
                                            \begin{array}{l}
                                                \tRcv{q}{hb}{}{\tSelect{\tSnd{q}{ok}{}{\recVar{F}},\ \tSnd{q}{ko}{}{\recVar{C$^\prime$}}}}, \\
                                                \tRcv{r}{hb}{}{\tSelect{\tSnd{r}{ok}{}{\recVar{F}},\ \tSnd{r}{ko}{}{\recVar{C$^\prime$}}}}, \\[1em]
                                                \tRcv{q}{rv}{}{\tSelect{\tSnd{q}{yes}{}{\recVar{F}},\ \tSnd{q}{no}{}{\recVar{C$^\prime$}}}}, \\
                                                \tRcv{r}{rv}{}{\tSelect{\tSnd{r}{yes}{}{\recVar{F}},\ \tSnd{r}{no}{}{\recVar{C$^\prime$}}}}, \\[1em]
                                                \tRcv{q}{no}{}{\recVar{C$^\prime$}},\sep
                                                \tRcv{r}{no}{}{\recVar{C$^\prime$}},\\
                                                \tRcv{q}{ok}{}{\recVar{C$^\prime$}},\sep
                                                \tRcv{r}{ok}{}{\recVar{C$^\prime$}},\\
                                                \tRcv{q}{ko}{}{\recVar{C$^\prime$}},\sep
                                                \tRcv{r}{ko}{}{\recVar{C$^\prime$}},\\[1em]
                                                \tRcv{q}{yes}{}{\tSession_\role{p}^L},\ 
                                                \tRcv{r}{yes}{}{\tSession_\role{p}^L}\\[1em]
                                                \cTimeOut{\ \recVar{C}}
                                            \end{array}
                                        }
                                    }
                                }}
                            }
                        }
                    \end{array}
                }
            } \\[14em]
            \tSession_\role{p}^L = \tRec{L}{
                \tSnd{q}{hb}{}{\tSnd{r}{hb}{}{
                    \tRec{L$^\prime$}{\ 
                        \tBranchL{
                            \begin{array}{l}
                                \tRcv{q}{hb}{}{\tSelect{\tSnd{q}{ok}{}{\recVar{F}},\ \tSnd{q}{ko}{}{\recVar{L$^\prime$}}}}, \\
                                \tRcv{r}{hb}{}{\tSelect{\tSnd{r}{ok}{}{\recVar{F}},\ \tSnd{r}{ko}{}{\recVar{L$^\prime$}}}}, \\[1em]
                                \tRcv{q}{rv}{}{\tSelect{\tSnd{q}{yes}{}{\recVar{F}},\ \tSnd{q}{no}{}{\recVar{L$^\prime$}}}}, \\
                                \tRcv{r}{rv}{}{\tSelect{\tSnd{r}{yes}{}{\recVar{F}},\ \tSnd{r}{no}{}{\recVar{L$^\prime$}}}}, \\[1em]
                                \tRcv{q}{yes}{}{\recVar{L$^\prime$}},\sep
                                \tRcv{r}{yes}{}{\recVar{L$^\prime$}},\sep
                                \tRcv{q}{no}{}{\recVar{L$^\prime$}},\sep
                                \tRcv{r}{no}{}{\recVar{L$^\prime$}},\\
                                \tRcv{q}{ok}{}{\recVar{L$^\prime$}},\sep
                                \tRcv{r}{ok}{}{\recVar{L$^\prime$}},\sep
                                \tRcv{q}{ko}{}{\recVar{L$^\prime$}},\sep
                                \tRcv{r}{ko}{}{\recVar{L$^\prime$}},\\[1em]
                                \cTimeOut{\ \recVar{L}}
                            \end{array}
                        }
                    }
                }}
            }
        \end{array}
    \end{align*}}

\end{example}

%% file: paper.bbl
\begin{thebibliography}{10}
\providecommand{\url}[1]{\texttt{#1}}
\providecommand{\urlprefix}{URL }
\providecommand{\doi}[1]{https://doi.org/#1}

\bibitem{DBLP:conf/forte/AdameitPN17}
Adameit, M., Peters, K., Nestmann, U.: Session types for link failures. In:
  Bouajjani, A., Silva, A. (eds.) Formal Techniques for Distributed Objects,
  Components, and Systems - 37th {IFIP} {WG} 6.1 International Conference,
  {FORTE} 2017, Held as Part of the 12th International Federated Conference on
  Distributed Computing Techniques, DisCoTec 2017, Neuch{\^{a}}tel,
  Switzerland, June 19-22, 2017, Proceedings. Lecture Notes in Computer
  Science, vol. 10321, pp. 1--16. Springer (2017).
  \doi{10.1007/978-3-319-60225-7\_1}

\bibitem{DBLP:conf/coordination/Amadio97}
Amadio, R.M.: An asynchronous model of locality, failure and process mobility.
  In: Garlan, D., M{\'{e}}tayer, D.L. (eds.) Coordination Languages and Models,
  Second International Conference, {COORDINATION} '97, Berlin, Germany,
  September 1-3, 1997, Proceedings. Lecture Notes in Computer Science,
  vol.~1282, pp. 374--391. Springer (1997). \doi{10.1007/3-540-63383-9\_92}

\bibitem{DBLP:journals/corr/abs-1211-2609}
Bartoletti, M., Scalas, A., Tuosto, E., Zunino, R.: Honesty by typing. Log.
  Methods Comput. Sci.  \textbf{12}(4) (2016). \doi{10.2168/LMCS-12(4:7)2016}

\bibitem{DBLP:journals/corr/abs-2207-02015}
Barwell, A.D., Scalas, A., Yoshida, N., Zhou, F.: Generalised multiparty
  session types with crash-stop failures. In: Klin, B., Lasota, S., Muscholl,
  A. (eds.) 33rd International Conference on Concurrency Theory, {CONCUR} 2022,
  September 12-16, 2022, Warsaw, Poland. LIPIcs, vol.~243, pp. 35:1--35:25.
  Schloss Dagstuhl - Leibniz-Zentrum f{\"{u}}r Informatik (2022).
  \doi{10.4230/LIPIcs.CONCUR.2022.35}

\bibitem{DBLP:journals/mscs/BusiGZ09}
Busi, N., Gabbrielli, M., Zavattaro, G.: On the expressive power of recursion,
  replication and iteration in process calculi. Math. Struct. Comput. Sci.
  \textbf{19}(6),  1191--1222 (2009). \doi{10.1017/S096012950999017X}

\bibitem{DBLP:journals/mscs/CapecchiGY16}
Capecchi, S., Giachino, E., Yoshida, N.: Global escape in multiparty sessions.
  Math. Struct. Comput. Sci.  \textbf{26}(2),  156--205 (2016).
  \doi{10.1017/S0960129514000164}

\bibitem{DBLP:books/el/01/Castellani01}
Castellani, I.: Process algebras with localities. In: Bergstra, J.A., Ponse,
  A., Smolka, S.A. (eds.) Handbook of Process Algebra, pp. 945--1045.
  North-Holland / Elsevier (2001). \doi{10.1016/b978-044482830-9/50033-3}

\bibitem{DBLP:conf/forte/ChenVBZE16}
Chen, T., Viering, M., Bejleri, A., Ziarek, L., Eugster, P.: A type theory for
  robust failure handling in distributed systems. In: Albert, E., Lanese, I.
  (eds.) Formal Techniques for Distributed Objects, Components, and Systems -
  36th {IFIP} {WG} 6.1 International Conference, {FORTE} 2016, Held as Part of
  the 11th International Federated Conference on Distributed Computing
  Techniques, DisCoTec 2016, Heraklion, Crete, Greece, June 6-9, 2016,
  Proceedings. Lecture Notes in Computer Science, vol.~9688, pp. 96--113.
  Springer (2016). \doi{10.1007/978-3-319-39570-8\_7}

\bibitem{DBLP:books/daglib/0007403}
Clarke, E.M., Grumberg, O., Peled, D.A.: Model checking. {MIT} Press (2001)

\bibitem{DBLP:conf/laser/ClarkeKNZ11}
Clarke, E.M., Klieber, W., Nov{\'{a}}cek, M., Zuliani, P.: Model checking and
  the state explosion problem. In: {LASER} Summer School. Lecture Notes in
  Computer Science, vol.~7682, pp. 1--30. Springer (2011).
  \doi{10.1007/978-3-642-35746-6\_1}

\bibitem{Dezani_2006}
Dezani-Ciancaglini, M., Mostrous, D., Yoshida, N., Drossopoulou, S.: Session
  {Types} for {Object}-{Oriented} {Languages}. In: {ECOOP} 2006. vol.~4067, pp.
  328--352. Springer Berlin Heidelberg (2006). \doi{10.1007/11785477_20},
  \url{http://link.springer.com/10.1007/11785477_20}, lecture Notes in Computer
  Science

\bibitem{DBLP:journals/pacmpl/FowlerLMD19}
Fowler, S., Lindley, S., Morris, J.G., Decova, S.: Exceptional asynchronous
  session types: session types without tiers. Proc. {ACM} Program. Lang.
  \textbf{3}({POPL}),  28:1--28:29 (2019). \doi{10.1145/3290341}

\bibitem{DBLP:books/daglib/0018113}
Hennessy, M.: A distributed Pi-calculus. Cambridge University Press (2007)

\bibitem{DBLP:conf/concur/Honda93}
Honda, K.: Types for dyadic interaction. In: Best, E. (ed.) {CONCUR} '93, 4th
  International Conference on Concurrency Theory, Hildesheim, Germany, August
  23-26, 1993, Proceedings. Lecture Notes in Computer Science, vol.~715, pp.
  509--523. Springer (1993). \doi{10.1007/3-540-57208-2\_35}

\bibitem{HVK98}
Honda, K., Vasconcelos, V.T., Kubo, M.: Language primitives and type discipline
  for structured communication-based programming. In: ESOP. LNCS, vol.~1381,
  pp. 122--138. Springer (1998). \doi{10.1007/BFb0053567}

\bibitem{DBLP:journals/jacm/HondaYC16}
Honda, K., Yoshida, N., Carbone, M.: Multiparty asynchronous session types. J.
  {ACM}  \textbf{63}(1),  9:1--9:67 (2016). \doi{10.1145/2827695}

\bibitem{KokkeD21}
Kokke, W., Dardha, O.: Deadlock-free session types in linear haskell. In: Hage,
  J. (ed.) Haskell 2021: Proceedings of the 14th {ACM} {SIGPLAN} International
  Symposium on Haskell, Virtual Event, Korea, August 26-27, 2021. pp. 1--13.
  {ACM} (2021). \doi{10.1145/3471874.3472979}

\bibitem{KouzapasDPG16}
Kouzapas, D., Dardha, O., Perera, R., Gay, S.J.: Typechecking protocols with
  mungo and stmungo. In: Cheney, J., Vidal, G. (eds.) Proceedings of the 18th
  International Symposium on Principles and Practice of Declarative
  Programming, Edinburgh, United Kingdom, September 5-7, 2016. pp. 146--159.
  {ACM} (2016). \doi{10.1145/2967973.2968595}

\bibitem{DBLP:conf/ecoop/LagaillardieNY22}
Lagaillardie, N., Neykova, R., Yoshida, N.: Stay safe under panic: Affine rust
  programming with multiparty session types. In: 36th European Conference on
  Object-Oriented Programming, {ECOOP} 2022, June 6-10, 2022, Berlin, Germany.
  LIPIcs, vol.~222, pp. 4:1--4:29. Schloss Dagstuhl - Leibniz-Zentrum f{\"{u}}r
  Informatik (2022). \doi{10.4230/LIPIcs.ECOOP.2022.4}

\bibitem{DBLP:journals/tocs/Lamport98}
Lamport, L.: The part-time parliament. {ACM} Trans. Comput. Syst.
  \textbf{16}(2),  133--169 (1998). \doi{10.1145/279227.279229}

\bibitem{DBLP:conf/popl/LangeNTY17}
Lange, J., Ng, N., Toninho, B., Yoshida, N.: Fencing off go: liveness and
  safety for channel-based programming. In: Proceedings of the 44th {ACM}
  {SIGPLAN} Symposium on Principles of Programming Languages, {POPL} 2017,
  Paris, France, January 18-20, 2017. pp. 748--761. {ACM} (2017).
  \doi{10.1145/3009837.3009847}

\bibitem{laprie1985dependable}
Laprie, J.C.: Dependable computing and fault-tolerance. Digest of Papers
  FTCS-15  \textbf{10}(2), ~124 (1985)

\bibitem{DBLP:conf/concur/MajumdarMSZ21}
Majumdar, R., Mukund, M., Stutz, F., Zufferey, D.: Generalising projection in
  asynchronous multiparty session types. In: Haddad, S., Varacca, D. (eds.)
  32nd International Conference on Concurrency Theory, {CONCUR} 2021, August
  24-27, 2021, Virtual Conference. LIPIcs, vol.~203, pp. 35:1--35:24. Schloss
  Dagstuhl - Leibniz-Zentrum f{\"{u}}r Informatik (2021).
  \doi{10.4230/LIPIcs.CONCUR.2021.35}

\bibitem{DBLP:journals/lmcs/MostrousV18}
Mostrous, D., Vasconcelos, V.T.: Affine sessions. Log. Methods Comput. Sci.
  \textbf{14}(4) (2018). \doi{10.23638/LMCS-14(4:14)2018}

\bibitem{DBLP:conf/usenix/OngaroO14}
Ongaro, D., Ousterhout, J.K.: In search of an understandable consensus
  algorithm. In: {USENIX} Annual Technical Conference. pp. 305--319. {USENIX}
  Association (2014),
  \url{https://www.usenix.org/conference/atc14/technical-sessions/presentation/ongaro}

\bibitem{orchard2017session}
Orchard, D., Yoshida, N.: Session types with linearity in haskell. Behavioural
  Types: from Theory to Tools p.~219 (2017)

\bibitem{DBLP:conf/forte/PetersNW22}
Peters, K., Nestmann, U., Wagner, C.: Fault-tolerant multiparty session types.
  In: Mousavi, M.R., Philippou, A. (eds.) Formal Techniques for Distributed
  Objects, Components, and Systems - 42nd {IFIP} {WG} 6.1 International
  Conference, {FORTE} 2022, Held as Part of the 17th International Federated
  Conference on Distributed Computing Techniques, DisCoTec 2022, Lucca, Italy,
  June 13-17, 2022, Proceedings. Lecture Notes in Computer Science, vol. 13273,
  pp. 93--113. Springer (2022). \doi{10.1007/978-3-031-08679-3\_7}

\bibitem{DBLP:books/daglib/0005958}
Pierce, B.C.: Types and programming languages. {MIT} Press (2002)

\bibitem{DBLP:journals/tcs/RielyH01}
Riely, J., Hennessy, M.: Distributed processes and location failures. Theor.
  Comput. Sci.  \textbf{266}(1-2),  693--735 (2001).
  \doi{10.1016/S0304-3975(00)00326-1}

\bibitem{DBLP:journals/fac/Rossi21}
Rossi, M.: Modeling and analysis of communicating systems. Formal Aspects
  Comput.  \textbf{33}(2),  297--298 (2021). \doi{10.1007/s00165-021-00533-8}

\bibitem{ScalasDHY17}
Scalas, A., Dardha, O., Hu, R., Yoshida, N.: A linear decomposition of
  multiparty sessions for safe distributed programming. In: M{\"{u}}ller, P.
  (ed.) 31st European Conference on Object-Oriented Programming, {ECOOP} 2017,
  June 19-23, 2017, Barcelona, Spain. LIPIcs, vol.~74, pp. 24:1--24:31. Schloss
  Dagstuhl - Leibniz-Zentrum f{\"{u}}r Informatik (2017).
  \doi{10.4230/LIPIcs.ECOOP.2017.24}

\bibitem{DBLP:journals/pacmpl/ScalasY19}
Scalas, A., Yoshida, N.: Less is more: multiparty session types revisited.
  Proc. {ACM} Program. Lang.  \textbf{3}({POPL}),  30:1--30:29 (2019).
  \doi{10.1145/3290343}

\bibitem{DBLP:conf/agere/TaboneF21}
Tabone, G., Francalanza, A.: Session types in elixir. In: Castegren, E.,
  Koster, J.D., Fowler, S. (eds.) {AGERE} 2021: Proceedings of the 11th {ACM}
  {SIGPLAN} International Workshop on Programming Based on Actors, Agents, and
  Decentralized Control, Virtual Event / Chicago, IL, USA, 17 October 2021. pp.
  12--23. {ACM} (2021). \doi{10.1145/3486601.3486708}

\bibitem{THK94}
Takeuchi, K., Honda, K., Kubo, M.: An interaction-based language and its typing
  system. In: PARLE '94. LNCS, vol.~817, pp. 398--413. Springer (1994)

\bibitem{V12}
Vasconcelos, V.T.: Fundamentals of session types. Inf. Comput.  \textbf{217},
  52--70 (2012). \doi{10.1016/j.ic.2012.05.002}

\bibitem{DBLP:conf/esop/VieringCEHZ18}
Viering, M., Chen, T., Eugster, P., Hu, R., Ziarek, L.: A typing discipline for
  statically verified crash failure handling in distributed systems. In: Ahmed,
  A. (ed.) Programming Languages and Systems - 27th European Symposium on
  Programming, {ESOP} 2018, Held as Part of the European Joint Conferences on
  Theory and Practice of Software, {ETAPS} 2018, Thessaloniki, Greece, April
  14-20, 2018, Proceedings. Lecture Notes in Computer Science, vol. 10801, pp.
  799--826. Springer (2018). \doi{10.1007/978-3-319-89884-1\_28}

\bibitem{DBLP:journals/pacmpl/VieringHEZ21}
Viering, M., Hu, R., Eugster, P., Ziarek, L.: A multiparty session typing
  discipline for fault-tolerant event-driven distributed programming. Proc.
  {ACM} Program. Lang.  \textbf{5}({OOPSLA}),  1--30 (2021).
  \doi{10.1145/3485501}

\end{thebibliography}
